\documentclass[reqno,oneside,10pt]{amsart}

\usepackage{amssymb,graphicx,braket}
\usepackage{hyperref}
\usepackage[utf8]{inputenc}
\usepackage{braket}
\usepackage{color}

\usepackage{thmtools}
\usepackage{thm-restate}
\usepackage{tikz-cd}
\usepackage{comment}
\usepackage{tensor}
\usepackage{mathtools}
\usepackage{wasysym}
\newcommand{\vi}[1]{v^{#1}(z)}

\theoremstyle{plain}
\newtheorem{Th}{Theorem}[section]
\newtheorem{Cor}[Th]{Corollary}
\newtheorem{Lem}[Th]{Lemma}
\newtheorem{Prop}[Th]{Proposition}

\theoremstyle{definition}
\newtheorem{Def}[Th]{Definition}
\newtheorem{Rem}[Th]{Remark}

\newtheorem{Ex}[Th]{Example}
\newtheorem{?}[Th]{Problem}

\numberwithin{equation}{section}

\let\frontmatter\relax
\def\mainmatter{\def\baselinestretch{1}\normalfont}

\usepackage[margin=1.2in]{geometry}

\makeatletter
\renewcommand{\section}{\@startsection
{section}%                   % the name
{1}%                         % the level
{\z@}%                       % the indent / 0mm
{-\baselineskip}%            % the before skip / -3.5ex \@plus -1ex \@minus -.2ex
{0.8\baselineskip}%          % the after skip / 2.3ex \@plus .2ex
{\centering\scshape\large}} % the style

\renewcommand{\subsection}{\@startsection
{subsection}%                   % the name
{2}%                         % the level
{\z@}%                       % the indent / 0mm
{-0.8\baselineskip}%            % the before skip / -3.5ex \@plus -1ex \@minus -.2ex
{0.5\baselineskip}%          % the after skip / 2.3ex \@plus .2ex
{\normalfont \bf \normalsize}} % the style

\renewcommand{\subsubsection}{\@startsection
{subsubsection}%                   % the name
{3}%                         % the level
{\z@}%                       % the indent / 0mm
{-0.8\baselineskip}%            % the before skip / -3.5ex \@plus -1ex \@minus -.2ex
{0.5\baselineskip}%          % the after skip / 2.3ex \@plus .2ex
{\normalfont \it \normalsize}} % the style
\makeatother

\usepackage{mathtools}

\newcommand{\e}{\boldsymbol{e}}

\definecolor{darkgreen}{rgb}{0.1, 0.8, 0.1}

\makeatletter
\def\l@subsection{\@tocline{2}{0pt}{2.5pc}{5pc}{}}

\renewcommand{\tocsection}[3]{%
  \indentlabel{\@ifnotempty{#2}{\bfseries\ignorespaces#1 #2.~}}\bfseries#3}
\renewcommand{\tocsubsection}[3]{%
  \indentlabel{\@ifnotempty{#2}{\ignorespaces#1 #2.~}}#3}

\makeatother

\begin{document}
\frontmatter

\title{Airy ideals, transvections and $\mathcal{W}(\mathfrak{sp}_{2N})$-algebras}

\author{Vincent Bouchard}
\author{Thomas Creutzig}
\author{Aniket Joshi}

 %\subjclass[2010]{Primary: 05C??. Secondary: 05C??}

%\keywords{Hecke operators, vector-valued modular forms, Weil representation} 

\begin{abstract} 
In the first part of the paper we propose a different viewpoint on the theory of higher Airy structures (or Airy ideals) which may shed light on its origin. We define Airy ideals in the $\hbar$-adic completion of the Rees Weyl algebra, and show that Airy ideals are defined exactly such that they are always related to the canonical left ideal generated by derivatives by automorphisms of the Rees Weyl algebra of a simple type, which we call transvections. The standard existence and uniqueness result in the theory of Airy structures then follows immediately.

In the second part of the paper we construct Airy ideals generated by the non-negative modes of the strong generators of the  principal $\mathcal W$-algebra of $\mathfrak{sp}_{2N}$ at level $-N-1/2$, following the approach developed in \cite{Airy}. This provides an example of an Airy ideal in the Heisenberg algebra that requires realizing the zero modes as derivatives instead of variables, which leads to an interesting interpretation for the resulting partition function.
\end{abstract}

\maketitle
\tableofcontents
\mainmatter

\section{Introduction}

Airy structures were introduced in \cite{ks} (see also \cite{ABCD}) as an algebraic reformulation of the topological recursion of Chekhov-Eynard-Orantin \cite{CE06,EO,EORev}. They were then generalized to higher Airy structures in \cite{Airy}, paralleling the construction of the higher topological recursion from \cite{BE2,BE,BHLMR}. The key result in the theory of Airy structures is the existence and uniqueness of a partition function annihilated by the differential operators that form the Airy structure. This partition function often plays the role of a generating series for interesting enumerative invariants, such as Hurwitz numbers or Gromov-Witten invariants, and the existence of an Airy structure implies that the generating series is uniquely fixed by a set of differential constraints, such as Virasoro or $\mathcal{W}$-constraints.

\subsection{A different viewpoint on Airy structures}

In this paper we first propose a different point of view on the definition of higher Airy structures. To be precise, according to the literature on the subject we consider in this paper ``quantum higher crosscapped Airy structures''. For simplicity, from now on we will call such objects simply Airy structures, or more precisely Airy ideals, as they really are left ideals.

Airy ideals are usually understood in terms of a collection of $\hbar$-dependent differential operators in the Weyl algebra -- which may be in terms of a finite or countably infinite number of variables -- that satisfy a certain number of properties. While the proof of these properties, such as existence and uniqueness of the partition function, is fairly straightforward, the definition of Airy ideals may seem to come out of nowhere. Our aim in the first part of this paper is to shed light on where the definition of Airy ideals comes from. As we will see, it is intimately connected to the existence of certain automorphisms of the Rees Weyl algebra which we call transvections.

The starting point of the story is the standard Rees construction, which turns filtered algebras and modules into graded algebras and modules by introducing a parameter $\hbar$. We apply this construction to the Weyl algebra  (in a finite or countably infinite number of variables) and its polynomial left module. We further construct the $\hbar$-adic completions of the Rees Weyl algebras and modules. This provides a clean and natural interpretation of the introduction of an $\hbar$ parameter, which plays a crucial role in the theory of Airy ideals. It avoids the need to deal with formal power series in the variables; however, it is absolutely key that we consider the $\hbar$-adic completion of the algebra and its modules (i.e., that we consider formal power series in $\hbar$).

Airy ideals will live in the completed Rees Weyl algebra (they are $\hbar$-dependent differential operators). But before we define Airy ideals, we consider a simple left ideal, which plays a fundamental role: the canonical left ideal $\mathcal{I}_{\text{can}}$ in the completed Rees Weyl algebra generated by the derivatives $\hbar \partial_a$. It is clear that the quotient of the completed Rees Weyl algebra by this left ideal is isomorphic, as a left module, to the completed Rees polynomial module, since each equivalence class clearly contains a polynomial representative. As we will see, the canonical left ideal $\mathcal{I}_{\text{can}}$ plays a key role, because Airy ideals will be defined exactly such that they are always the image of the canonical left ideal under a certain type of automorphisms of the completed Rees Weyl algebra.

To understand this statetement, we introduce a class of automorphisms of the completed Rees Weyl algebra that we call ``transvections''.  Those automorphism transform the derivatives $\hbar \partial_a$, but keep the variables $\hbar x_a$ and $\hbar$ invariant. A transvection maps the canonical left ideal $\mathcal{I}_{\text{can}}$ into a left ideal $\mathcal{I}$ generated by a complete set of commuting first-order differential operators $\bar H_a$. Since it is an automorphism, it follows that the quotient of the completed Rees Weyl algebra by $\mathcal{I}$ is canonically isomorphic as a left module to the completed Rees polynomial module twisted by the automorphism $\phi$. We can reformulate this in terms of the standard (untwisted) action of differential operators using what we call ``modules of exponential type''. As a result, we find a unique exponential solution to the differential equations $\mathcal{I} \cdot Z = 0$ after imposing a suitable initial condition. 

All this structure provides the foundations to understand the definition of Airy ideals. Airy ideals (usually called Airy structures in the literature) are left ideals in the completed Rees Weyl algebra generated by some operators $H_a$ (not necessarily first-order) satisfying a set of conditions. We show that given an Airy ideal $\mathcal{I}$, there always exists a stable transvection $\phi$ such that $\mathcal{I}$ is isomorphic to $\phi(\mathcal{I}_{\text{can}})$. This is rather striking; it is the key result in the theory of Airy ideals. One could say that the definition of Airy ideals is precisely such that this result holds. It implies that the quotient of the completed Rees Weyl algebra by an Airy ideal $\mathcal{I}$ is canonically isomorphic as a left module to the completed Rees polynomial module twisted by $\phi$. In terms of the standard (untwisted) action of differential operators, the quotient is isomorphic to a module of exponential type, which in turn implies the existence of a unique exponential solution $Z$ to the equations $\mathcal{I} \cdot Z = 0$ after imposing a suitable initial condition. This is the fundamental existence and uniqueness theorem first proved by Kontsevich and Soibelman in \cite{ks}. It appears here as a byproduct of the fact that Airy ideals are always related to the canonical left ideal $\mathcal{I}_{\text{can}}$ generated by derivatives by a transvection. Neat!

With this viewpoint on Airy ideals, we further explore the connection with the Heisenberg algebra and the free boson VOA, which is where many examples of Airy ideals are found. Indeed, the Heisenberg algebra is intimately connected to the Weyl algebra -- it is almost the same, but one needs to be careful with the existence of zero modes in the Heisenberg algebra, which have no counterparts \emph{a priori} in the Weyl algebra. We investigate this carefully, and show that we can construct Airy ideals in the Heisenberg algebra in two different ways: by either considering the zero modes to act as variables or as derivatives. So far in the literature all examples of Airy ideals constructed within the Heisenberg algebra (such as Airy ideals for $\mathcal{W}$-algebras in \cite{Airy,Whittaker,bks,BM,Milanov}) have used the first scenario, where the zero modes are understood as variables (or often set to zero). The example that we investigate in the remainder of the paper uses the second scenario, where the zero modes act as derivatives. This leads to an interesting new interpretation for the partition function, as it does not live in the Fock module of the Heisenberg algebra anymore -- it involves conjugate modes to the zero modes of the Heisenberg algebra and thus is an element of an infinite length indecomposable module of the Heisenberg algebra. 

\subsection{Airy ideals for $\mathcal{W}(\mathfrak{sp}_{2N})$}

In the rest of the paper we propose a new class of Airy ideals constructed as modules of $\mathcal{W}$-algebras, in the spirit of \cite{Airy}. We study the principal $\mathcal{W}(\mathfrak{sp}_{2N})$-algebra at level $-N-1/2$. It is naturally realized within the symplectic fermion algebra of rank $N$. Using the boson-fermion correspondence, we show that it can also be obtained within the rank $N$ free boson VOA. 

We then construct a twisted module for the rank $N$ free boson VOA, which restricts to a (untwisted) module for the $\mathcal{W}(\mathfrak{sp}_{2N})$-algebra. Our goal is to show that the non-negative modes of the strong generators of $\mathcal{W}(\mathfrak{sp}_{2N})$ in this module generate an Airy ideal. This is not directly the case, but we show that we can introduce an automorphism (known as ``dilaton shift'') on the algebra of bosonic modes such that the image of the non-negative modes of the strong generators of $\mathcal{W}(\mathfrak{sp}_{2N})$ under this automorphism generate an Airy ideal.

The upshot is that all results in the theory of Airy ideals then apply to $\mathcal{W}(\mathfrak{sp}_{2N})$. We get that the quotient of the completed Rees algebra of bosonic modes by the Airy ideal is canonically isomorphic to the $\hbar$-adically completed Fock module for the rank $N$ free boson VOA, but twisted by a transvection. Similarly, this quotient is also isomorphic to a module of exponential type generated by a particular state (the partition function associated to the Airy ideal), which we can construct recursively order by order in $\hbar$. What is interesting is that this state does not live in the Fock module for the free boson VOA, since it involves the modes conjugate to the zero modes. It lives inside a huge extension, see section \ref{s:heisenberg}. This is a key consequence of the crucial fact that we need to interpret zero modes as derivatives in this particular case.

 Finally, by construction the state (the partition function) is annihilated by the non-negative modes of the strong generators of $\mathcal{W}(\mathfrak{sp}_{2N})$. Therefore, the action of the negative modes generates a ($\hbar$-adically completed) Fock module for $\mathcal{W}(\mathfrak{sp}_{2N})$. Its existence and uniqueness is guaranteed by the theory of Airy ideals -- in fact we can construct it explicitly order by order in $\hbar$.
 
 \subsection{Outline}
 
 In section 2 we develop our new point of view on Airy ideals. We review preliminaries on filtered algebras and modules in Section 2.1. We apply the Rees construction to the Weyl algebra in Section 2.2, and to its polynomial module in Section 2.3. We explore related twisted modules, and introduce modules of ``exponential type'' in Section 2.3. In Section 2.4 we study a special class of automorphisms for the Rees Weyl algebra, which we call transvections. They will play a key role in the following. We introduce the notion of Airy ideals in Section 2.5 (Definition \ref{d:airy}). The key result, which states that Airy ideals are always related to the canonical left ideal generated by derivatives via a stable transvection, is proved in Theorem \ref{t:airy}. Its main corollary, which is the existence and uniqueness of a partition function, is explored in Corollary \ref{c:pf}.
 In Section 2.6 we explain how we can construct Airy ideals concretely (see Lemma \ref{l:aairy}). 
 
 We study two special cases of the construction in Section 2.7 -- when the differential operators either do not depend on all variables, or do not involve all derivatives. These special cases are naturally realized when we connect the Weyl algebra to the Heisenberg algebra, via a choice of interpretation of the zero modes; this is explored in Section 2.8. 
 
 Sections 3 and 4 are devoted to the construction of Airy ideals for the $\mathcal{W}(\mathfrak{sp}_{2N})$-algebras. We review basic properties of $\mathcal{W}(\mathfrak{sp}_{2N})$ in Section 3, and show how we can write down explicit formulae for the strong generators of $\mathcal{W}(\mathfrak{sp}_{2N})$ in a twisted module for the rank $N$ free boson VOA. In Section 4 we show that there exists an automorphism (a dilaton shift) of the algebra of modes such that the image of the non-negative modes of the strong generators of $\mathcal{W}(\mathfrak{sp}_{2N})$ under this automorphism generate an Airy ideal -- this is the key result of that section, which is Theorem \ref{t:3} for $N=3$, and Theorem \ref{t:N} for general $N$.

\noindent {\bf{Acknowledgements}}
 
 We would like to thank Raphaël Belliard and Reinier Kramer for interesting discussions. The authors acknowledge the support of the Natural Sciences and Engineering Research Council of Canada.

\section{Airy ideals}

In this section we propose a different point of view on the definition of Airy structures, as explained in the Introduction.

\subsection{Preliminaries}

\subsubsection{Cyclic modules and twisted modules}

We first review basic concepts in the theory of modules that will be needed later on.  Let $\mathcal{D}$ be an associative algebra over  $\mathbb{C}$,\footnote{We could work over any field $\mathbb{K}$ of characteristic zero instead of $\mathbb{C}$.} and $\mathcal{M}$ a left $\mathcal{D}$-module. We write $r \cdot m \in \mathcal{M}$ for the action of $r \in \mathcal{D}$ on $m \in \mathcal{M}$.

\begin{Def}
The \emph{annihilator} of an element $v \in \mathcal{M}$, which is denoted by $\text{Ann}_{\mathcal{D}}(v)$, is defined as
\begin{equation}
\text{Ann}_{\mathcal{D}}(v) = \{ P \in \mathcal{D}\ |\ P \cdot v = 0 \}.
\end{equation}
It is naturally a left ideal in $\mathcal{D}$.
\end{Def}

\begin{Def}
A left $\mathcal{D}$-module $\mathcal{M}$ is \emph{cyclic} if it is generated by a single element $v \in \mathcal{M}$.
\end{Def}

It is easy to show that a cyclic left $\mathcal{D}$-module $\mathcal{M}$ generated by $v \in \mathcal{M}$ is canonically isomorphic to $\mathcal{D} /\text{Ann}_{\mathcal{D}}(v) $.

We will also need the notion of a twisted module with respect to an automorphism $\phi: \mathcal{D} \to \mathcal{D}$.

\begin{Def}
Let $\phi: \mathcal{D} \to \mathcal{D}$ be an automorphism, and $\mathcal{M}$ a left $\mathcal{D}$-module. 
The \emph{twisted module} $\tensor[^\phi]{\mathcal{M}}{}$ is given by the same vector space as $\mathcal{M}$, but with the new operation
\begin{equation}
r \cdot_\phi m = \phi^{-1}(r) \cdot m.
\end{equation}
\end{Def}

It is easy to show that, if $\mathcal{M}$ is a cyclic left $\mathcal{D}$-module generated by $v \in \mathcal{M}$, then the twisted module $\prescript{\phi}{}{\mathcal{M}}$ is also cyclic and generated by $v$. Furthermore, the annihilator of $v \in  \prescript{\phi}{}{\mathcal{M}}$, which we denote by $\prescript{\phi}{}{\text{Ann}}_{\mathcal{D}}(v)$ to avoid confusion with the annihilator $\text{Ann}_{\mathcal{D}}(v)$ of $v$ in $\mathcal{M}$, is:
\begin{equation}
\prescript{\phi}{}{\text{Ann}}_{\mathcal{D}}(v) = \phi\left( \text{Ann}_{\mathcal{D}}(v) \right).
\end{equation}
It then follows that $\prescript{\phi}{}{\mathcal{M}}$ is canonically isomorphic to $\mathcal{D} /  \phi\left( \text{Ann}_{\mathcal{D}}(v) \right)$.

\subsubsection{Filtrations, Rees algebras and Rees modules}

We now review the Rees construction for filtered algebras and modules. We write $\mathbb{N}$ for the set of non-negative integers, and $\mathbb{N}^*$ for the set of positive integers. 

\begin{Def}An \emph{exhaustive ascending filtration} on $\mathcal{D}$ is an increasing sequence of subspaces $F_i \mathcal{D} \subseteq \mathcal{D}$, for $i \in \mathbb{N}$:
\begin{equation}
\{0\} \subseteq F_0 \mathcal{D} \subseteq F_1 \mathcal{D}\subseteq F_2 \mathcal{D} \subseteq \ldots \subseteq \mathcal{D},
\end{equation}
such that $ \cup_{i \in \mathbb{N}} F_i \mathcal{D} = \mathcal{D}$ and $F_i \mathcal{D} \cdot F_j \mathcal{D} \subseteq F_{i+j} \mathcal{D}$ for all $i,j \in \mathbb{N}$.
 An algebra $\mathcal{D}$ with such a filtration is called a \emph{filtered algebra}. 
 \end{Def}
 
 Filtered modules are defined in a similar way.
 
 \begin{Def}
  Let $\mathcal{M}$ be a left $\mathcal{D}$-module. An \emph{exhaustive ascending filtration} on $\mathcal{M}$ is given by an increasing sequence of subspace $F_i\mathcal{M} \subset \mathcal{M}$ for $i \in \mathbb{N}$:
  \begin{equation}
\{0\} \subseteq F_0 \mathcal{M} \subseteq F_1 \mathcal{M} \subseteq F_2 \mathcal{M} \subseteq \ldots \subseteq \mathcal{M},
\end{equation}
such that $ \cup_{i \in \mathbb{N}} F_i \mathcal{M} = \mathcal{M}$ and $F_i \mathcal{D} \cdot F_j \mathcal{M} \subseteq F_{i+j} \mathcal{M}$ for all $i,j \in \mathbb{N}$.
 A left $\mathcal{D}$-module $\mathcal{M}$ with such a filtration is called a \emph{filtered module}.
 \end{Def}
 
 From a filtered algebra, we can construct a graded algebra in a natural way: this is the Rees construction. Note that this construction is different from the standard associated graded algebra $\text{Gr}(\mathcal{D}) = \bigoplus_{n=1}^\infty \mathcal{G}_n$ with $\mathcal{G}_n = F_n \mathcal{D}/F_{n-1} \mathcal{D}$. 
 
 \begin{Def}
 Given a filtered algebra $\mathcal{D}$, we define the \emph{Rees algebra} $\mathcal{D}^\hbar$ as:
 \begin{equation}
 \mathcal{D}^\hbar = \bigoplus_{n \in \mathbb{N}}   \hbar^n F_n \mathcal{D}.
 \end{equation}
 It is a graded algebra, graded by $\hbar$ (with $\deg(\hbar) = 1$). When needed, we write $\mathcal{D}_n^\hbar := \hbar^n F_n \mathcal{D}$ for the subspace of homomegeneous elements of degree $n$.
 \end{Def}
 
 It will be very important for us to consider not only Rees algebras, but also their completions with respect to the $\hbar$-adic topology.
 
 \begin{Def}\label{d:hadiccompletion}
 We define the \emph{completed Rees algebra} $\widehat{\mathcal{D}}^\hbar$ as
 \begin{equation}
 \widehat{\mathcal{D}}^\hbar = \prod_{n \in \mathbb{N}} \hbar^n F_n \mathcal{D},
 \end{equation}
 which is the completion with respect to the $\hbar$-adic topology. Explicitly, an element $P \in  \widehat{\mathcal{D}}^\hbar$ can be written as a formal power series in $\hbar$:
 \begin{equation}
 P = \sum_{n=0}^\infty \hbar^n P_n,
 \end{equation}
 for some $P_n \in F_n \mathcal{D}$.
 \end{Def}
 
 The same Rees construction can be applied to filtered modules.
 
 \begin{Def}\label{d:modules}
 Given a filtered $\mathcal{D}$-module $\mathcal{M}$, we define the \emph{Rees module} $\mathcal{M}^\hbar$ as
 \begin{equation}
 \mathcal{M}^\hbar = \bigoplus_{n \in \mathbb{N}} \hbar^n F_n \mathcal{M}.
  \end{equation}
  It is a graded left $\mathcal{D}^\hbar$-module,  and we write $\mathcal{M}_n^\hbar = \hbar^n F_n \mathcal{M}$ for the subspace of homogeneous elements.  We define the completed Rees module $\widehat{\mathcal{M}}^\hbar$ as
 \begin{equation}
 \widehat{\mathcal{M}}^\hbar = \prod_{n \in \mathbb{N}} \hbar^n F_n \mathcal{M},
 \end{equation}
 which is the completion with respect to the $\hbar$-adic topology. 
  \end{Def}

 \subsection{Rees Weyl algebra}
 
 We now apply the Rees construction to the Weyl algebra (in either a finite or countably infinite number of variables).
 
 Let $A$ be an index subset, either finite or countably infinite. We write $x_A = \{x_a\}_{a \in A}$ for the set of variables $x_a$ with $a \in A$, and $\partial_A =\{\partial_a\}_{a \in A} $ for the set of partial derivatives $\partial_a$ with respect to the variables $x_a$.
 
 \begin{Def}
 If $A$ is a finite index set, we define the \emph{Weyl algebra} $\mathcal{D}_A = \mathbb{C}[x_A]\langle\partial_A\rangle$ to be the algebra of differential operators over the polynomial ring $\mathbb{C}[x_A]$ in the variables $x_A$. $\mathcal{D}_A$ is the free associative algebra over $\mathbb{C}$ generated by $\{x_A, \partial_A \}$ modulo the commutation relations
 \begin{equation}
[ x_a, x_b] = 0, \qquad [ \partial_a, \partial_b] = 0, \qquad [ \partial_a, x_b] = \delta_{a b}, \qquad \forall a,b \in A.
 \end{equation}
 \end{Def}

In the case where $A$ is countably infinite, we define $\mathcal{D}_A$ to be a particular completion of the Weyl algebra  $ \mathbb{C}[x_A]\langle\partial_A\rangle$.

\begin{Def}
\label{d:infinite}
If $A$ is a countably infinite index set, we define the \emph{completed Weyl algebra} $\mathcal{D}_A$ to be the completion of the Weyl algebra $\mathbb{C}[x_A]\langle\partial_A\rangle$ that contains potentially infinite sums in the derivatives, but with polynomial coefficients.
 (Elements of $\mathcal{D}_A$ remain of finite order as differential operators.)\footnote{This completion should not be confused with the $\hbar$-adic completion of Rees algebras and modules discussed in the previous section.}
\end{Def}

In other words, we  can write an element $P \in \mathcal{D}_A$ uniquely as
\begin{equation}\label{eq:P}
P = \sum_{m=0}^M  \sum_{a_1, \ldots, a_m \in A} p_{a_1 \cdots a_m}(x_A) \partial_{a_1} \cdots \partial_{a_m},
\end{equation}
for some $M \in \mathbb{N}$,
where the $p_{a_1 \cdots a_m}(x_A)$ are polynomials in the variables $x_A$. If $A$ is a countably infinite index set, we see that the sums over the indices $a_i$ can be infinite, but the coefficients are always polynomial (they cannot include infinite sums of monomials). For example, what this means is that an operator like
$\sum_{a \in A} \partial_a$ is in $\mathcal{D}_A$, while $\sum_{a \in A} x_a$ is not.

There is a natural exhaustive ascending filtration on $\mathcal{D}_A$ called the Bernstein filtration. To construct it, we give degree one to the variables $x_a$ and the partial derivatives $\partial_a$, and define the subspaces $F_i \mathcal{D}_A$ as containing all operators in $\mathcal{D}_A$ of degree $\leq i$. More precisely:

\begin{Def}
\label{d:bernstein}
The \emph{Bernstein filtration} on $\mathcal{D}_A$ is defined by
 \begin{equation}
 F_i \mathcal{D}_A = \left \{ \sum_{\substack{m, k \in \mathbb{N} \\ m+k = i}} \sum_{a_1, \ldots, a_m \in A} p_{a_1 \cdots a_m}^{(k)}(x_A) \partial_{a_1} \cdots \partial_{a_m} \right \},
\end{equation}
where the $ p_{a_1 \cdots a_m}^{(k)}(x_A)$ are polynomials of degree $\leq k$.
 Here, $F_0 \mathcal{D}_A = \mathbb{C}$. 
 \end{Def}
 
 From the definition of  $\mathcal{D}_A$ and its Bernstein filtration, it is clear that
 \begin{equation}\label{eq:filtr}
 [ F_m \mathcal{D}_A, F_n \mathcal{D}_A ] \subseteq F_{m+n-2} \mathcal{D}_A.
 \end{equation}
 
 As in the previous section, we construct the Rees algebra associated to the filtered algebra $\mathcal{D}_A$ with the Bernstein filtration.
 \begin{Def}
The \emph{Rees Weyl algebra} $\mathcal{D}_A^{\hbar}$ associated to $\mathcal{D}_A$ with the Bernstein filtration is
 \begin{equation}
 \mathcal{D}_A^{\hbar} = \bigoplus_{n \in \mathbb{N}} \hbar^n F_n \mathcal{D}_A,
 \end{equation}
 which is a graded algebra, graded by $\hbar$ with $\deg(\hbar) = 1$. We write $\mathcal{D}_{A,n}^{\hbar} =  \hbar^n F_n \mathcal{D}_A$ for homogeneous elements of degree $n$.  We define its $\hbar$-adic completion, as in Definition \ref{d:hadiccompletion}:
  \begin{equation}
 \widehat{\mathcal{D}}_A^{\hbar} = \prod_{n \in \mathbb{N}} \hbar^n F_n \mathcal{D}_A,
 \end{equation}
  From \eqref{eq:filtr} we get that
 \begin{equation}\label{eq:hbar2}
 [\widehat{\mathcal{D}}_{A,m}^{\hbar}, \widehat{\mathcal{D}}_{A,n}^{\hbar}] \subseteq \hbar^2 \widehat{\mathcal{D}}_{A,m+n-2}^{\hbar}.
 \end{equation}
 \end{Def}

 \begin{Rem}\label{r:reesfree}
 The Rees Weyl algebra $ \mathcal{D}_A^{\hbar}  \subset  \widehat{\mathcal{D}}_A^{\hbar} $ is the subalgebra consisting of differential operators that are polynomials in $\hbar$.
 Note that we can also think of the Rees Weyl algebra
 $\mathcal{D}_A^\hbar$ as the free associative algebra over $\mathbb{C}$ generated by $\{\hbar, \hbar x_A,  \hbar \partial_A \}$, where $\hbar$ is a central element, and the other generators satisfy the commutation relations
 \begin{equation}
[ \hbar x_a,\hbar x_b] = 0, \qquad [\hbar \partial_a, \hbar \partial_b] = 0, \qquad [ \hbar \partial_a, \hbar x_b] = \hbar^2 \delta_{a b}, \qquad \forall a,b \in A.
 \end{equation}
 \end{Rem}
 
\begin{Ex}
To clarify the notation, an operator $P \in \widehat{\mathcal{D}}_A^{\hbar} $ can be written as
\begin{equation}
P = \sum_{n \in \mathbb{N}} \hbar^n \sum_{\substack{m, k \in \mathbb{N} \\ m+k = n}} \sum_{a_1, \ldots, a_m \in A} p_{a_1 \cdots a_m}^{(n,k)}(x_A) \partial_{a_1} \cdots \partial_{a_m},
\end{equation}
where the $ p_{a_1 \cdots a_m}^{(n,k)}(x_A)$ are polynomials of degree $\leq k$.
\end{Ex}

Because of the subtelties arising due to infinite sums when $A$ is a countably infinite index set, we need to define a particular property for collections of operators in $\widehat{\mathcal{D}}_A^{\hbar} $, which we call boundedness (this condition is called ``filtered family of operators'' in \cite{bks} -- see Section 2.1.2).

\begin{Def}\label{d:bounded}
Let $I$ be a finite or countably infinite index set, and $\{P_i\}_{i \in I}$ be a collection of operators $P_i \in \widehat{\mathcal{D}}_A^{\hbar} $ of the form
\begin{equation}
P_i = \sum_{n \in \mathbb{N}} \hbar^n \sum_{\substack{m, k \in \mathbb{N} \\ m+k = n}} \sum_{a_1, \ldots, a_m \in A} p_{i;a_1 \cdots a_m}^{(n,k)}(x_A) \partial_{a_1} \cdots \partial_{a_m},
\end{equation}
We say that the collection of operators $\{P_i\}_{i \in I}$ is \emph{bounded} if, for all fixed choice of indices $a_1, \ldots, a_m$, $n$, and $k$, the polynomials $p_{i;a_1 \cdots a_m}^{(n,k)}(x_A)$ vanish for all but finitely many indices $i \in I$. We note that the condition is trivially satisfied if $I$ is a finite index set.
\end{Def}

There is a fundamental reason why we consider bounded collection of differential operators. In the following we will study left ideals $\mathcal{I}$ in $\widehat{\mathcal{D}}_A^{\hbar} $ consisting of all $\widehat{\mathcal{D}}_A^{\hbar} $-linear combinations of a collection of operators $\{P_i\}_{i \in I}$; that is, any $Q \in \mathcal{I}$ can be written as
\begin{equation}
Q  = \sum_{i \in I} c_i P_i
\end{equation}
for some $c_i \in \widehat{\mathcal{D}}_A^{\hbar} $ . If $I$ is a finite index set, this is the left ideal generated by the collection of operators $\{P_i\}_{i \in I}$. However, if $I$ is a countably infinite index set, we will want our ideal $\mathcal{I}$ to contain not only finite  $\widehat{\mathcal{D}}_A^{\hbar} $-linear combinations of the $P_i$, but also infinite ones.\footnote{We will often abuse notation and still say that this ideal is ``generated'' by the $P_i$, even though an ideal generated by a set only contains finite linear combination of the elements in the set, regardless of whether the set is finite or countably infinite. For us, we always include both finite and infinite linear combinations when the generating set is countably infinite.} But if $\{P_i\}_{i \in I}$ is an arbitrary collection of operators, infinite $\widehat{\mathcal{D}}_A^{\hbar} $-linear combinations of the $P_i$ may give rise to divergent infinite sums, or to operators whose coefficients are not polynomials in the variables $x_A$ (they may contain infinite sums of monomials). However, if the collection $\{P_i\}_{i \in I}$ is bounded, this cannot happen; in this case, it is straightforward to show that infinite $\widehat{\mathcal{D}}_A^{\hbar} $-linear combinations of the $P_i$ are always well defined operators whose coefficients are polynomials in the variables $x_A$ (since they are finite sums of polynomials), and therefore in $\widehat{\mathcal{D}}_A^{\hbar} $. This is the key reason why we consider bounded collection of operators in $\widehat{\mathcal{D}}_A^{\hbar} $.
 
  \begin{Rem}
 For the readers familiar with the literature on Airy structures, a word of caution is required at this stage. Our $\hbar$ differs from the usual $\hbar$ in the literature on Airy structures. More precisely, as should become clear later, to connect the two approaches, one should start with the traditional definition of Airy structures (for instance in \cite{ks,Airy}), rescale the variables as $x_i \mapsto \hbar^{1/2} x_i$, and then redefine $\hbar \mapsto \hbar^2$. With this transformation, the grading defined in \cite{ks,Airy} becomes the natural $\hbar$-grading on the Rees algebra that we introduce here, with $\deg(\hbar)=1$.
 
 We could also introduce $\hbar$ as in the traditional literature on Airy structures using the Rees construction. What we would need to do then is consider a different filtration on the Weyl algebra, namely the ``order filtration'' instead of the Bernstein filtration, which is defined by 
  \begin{equation}
 F_i \mathcal{D}_A = \left \{ \sum_{m=0}^i \sum_{a_1, \ldots, a_m \in A} p_{a_1 \cdots a_m}(x_A) \partial_{a_1} \cdots \partial_{a_m} \right \},
\end{equation}
where the polynomials $p_{a_1 \cdots a_m}(x_A) $ have arbitrary degree. In other words, $F_i \mathcal{D}_A$ consists of differential operators of order at most $i$, but with polynomial coefficients of arbitrary degree (it corresponds to giving degree one to the partial derivatives $\partial_a$, but degree zero to the variables $x_a$). We could then define the corresponding Rees algebra; the result would be the standard $\hbar$-dependent Weyl algebra considered in the literature on Airy structures.
 
Although the two approaches are ultimately equivalent, we find the introduction of $\hbar$ via the Bernstein filtration instead of the order filtration more natural and transparent, as, among other things, it allows us to work with the Rees polynomial module (i.e. we also introduce $\hbar$ for a $\mathcal{D}_A^\hbar$-module via the Rees construction), and we don't need to deal with formal power series in the variables $x_A$, as will become clear later -- we only need to consider the $\hbar$-adic completions for the Weyl algebra and its polynomial module.
 \end{Rem}

\subsection{Left $\widehat{\mathcal{D}}_A^\hbar$-modules}

\subsubsection{Polynomial $\mathcal{D}_A$-module}

A natural left $\mathcal{D}_A$-module is the polynomial algebra $\mathcal{M}_A = \mathbb{C}[x_A]$, where the action is given by the standard action of differential operators on polynomials. In the case where $A$ is countably infinite, one should be a little bit careful here, since our algebra $\mathcal{D}_A$ is the completion of the Weyl algebra, which includes potentially infinite sums over the derivatives. However, since $\mathcal{M}_A$ is a polynomial algebra, the action of differential operators in $\mathcal{D}_A$ on polynomials always collapses the infinite sums to finite sums, and so the action is well defined.\footnote{For instance, we need to be careful that we don't encounter situations  like the differential operator $\sum_{a \in A} \partial_a$ acting on $\sum_{b \in A} x_b$, since $\sum_{a \in A} \partial_a \left( \sum_{b \in B} x_b \right) = \sum_{a \in A} 1$ which is of course divergent. But, while $\sum_{a \in A} \partial_a$ is in $\mathcal{D}_A$, $\sum_{b \in A} x_b$ is not in $\mathcal{M}_A$ since it is not a polynomial. We never run into this kind of issues because all infinite sums collapse to finite sums, since all elements of $\mathcal{M}_A$ are polynomials, even if we work with infinitely many variables.}

The polynomial module $\mathcal{M}_A$ is a cyclic left $\mathcal{D}_A$-module, generated by $1 \in \mathcal{M}_A$. Moreover, the annihilator of $1$ is
\begin{equation}
\text{Ann}_{\mathcal{D}_A}(1) = \left \{ \sum_{a \in A} c_a \partial_a\ |\ c_a \in \mathcal{D}_A \right \},
\end{equation}
which is the left ideal consisting of $\mathcal{D}_A$-linear combinations of the derivatives. It is clear that $\mathcal{M}_A$ is canonically isomorphic to $\mathcal{D}_A/\text{Ann}_{\mathcal{D}_A}(1)$.

\subsubsection{Rees polynomial $\widehat{\mathcal{D}}_A^\hbar$-module}

Let us now apply the Rees construction to the polynomial $\mathcal{D}_A$-module. We can define many filtrations on the polynomial algebra $\mathcal{M}_A$ that are compatible with the Bernstein filtration on $\mathcal{D}_A$. We will use the following standard filtration.

\begin{Def}
We define the \emph{degree filtration} on $\mathcal{M}_A$ as:
\begin{equation}
F_i \mathcal{M}_A = \{ \text{polynomials of degree $\leq i$} \}.
\end{equation}
%For any $N  \in \mathbb{N}^*$, we defined the \emph{rescaled degree filtration} as:
%\begin{equation}
%F_i \mathcal{M}_A(N) =  \{ \text{polynomials of degree $\leq i N$} \}.
%\end{equation}

\end{Def}

It is easy to check that $\mathcal{M}_A$ with this filtration is a filtered $\mathcal{D}_A$-module. We then apply the Rees construction for filtered modules.

\begin{Def}
We define the Rees polynomial module associated to the degree filtration and its $\hbar$-adic completion:
\begin{equation}\label{eq:reesmodule}
\mathcal{M}_A^{\hbar} = \bigoplus_{n \in \mathbb{N}} \hbar^n F_n \mathcal{M}_A, \qquad \widehat{\mathcal{M}}_A^{\hbar} = \prod_{n \in \mathbb{N}} \hbar^n F_n \mathcal{M}_A .
\end{equation}
Both are graded left $\mathcal{D}_A^\hbar$-modules, and the completed module $\widehat{\mathcal{M}}_A^{\hbar}$ is also a left $\widehat{\mathcal{D}}_A^\hbar$-module.
\end{Def}

\begin{Ex}
To clarify the notation, an element $f \in \widehat{\mathcal{M}}_A^{\hbar}$ is a formal $\hbar$-power series
\begin{equation}
f = \sum_{n=0}^\infty \hbar^n f^{(n)}(x_A)
\end{equation}
where the $f^{(n)}(x_A)$ are polynomials in the variables $x_A$ of degree $\leq n$. 
\end{Ex}

In the following we will mostly be interested in the completed module $\widehat{\mathcal{M}}_A^{\hbar}$, realized as a left $\widehat{\mathcal{D}}_A^\hbar$-module. It is easy to show that it is a cyclic module, generated by $1 \in \widehat{\mathcal{M}}_A^{\hbar}$. The annihilator of $1$ is:
\begin{equation}\label{eq:annih1}
\text{Ann}_{\widehat{\mathcal{D}}^\hbar_A}(1) = \left \{ \sum_{a \in A} c_a \hbar \partial_a\ |\ c_a \in \widehat{\mathcal{D}}^\hbar_A \right \} =: \mathcal{I}_{\text{can}},
\end{equation}
which is the left ideal consisting of $\widehat{\mathcal{D}}^\hbar_A$-linear combinations of the derivatives $\hbar \partial_a$. As we will refer to this canonical left ideal many times in the following, we introduce the shorthand notation $\mathcal{I}_{\text{can}}$. 
 $\widehat{\mathcal{M}}_A^{\hbar}$ is canonically isomorphic to $\widehat{\mathcal{D}}^\hbar_A / \mathcal{I}_{\text{can}}$, which is easy to see.

Since we are working with the Weyl algebra, we can also think of this as solving differential equations. The statement above is that $Z=1 \in \widehat{\mathcal{M}}_A^{\hbar}$ is a solution to the differential equations $\mathcal{I}_{\text{can}} \cdot Z = 0$. That is,  $\hbar \partial_a(1) = 0$ for all $a \in A$, which is obvious. In fact, it is the unique solution to the system $\mathcal{I}_{\text{can}} \cdot Z = 0$ if we impose the initial condition $Z \big|_{x_A=0} = 1$. We call this unique solution $Z=1$ the \emph{partition function} associated to the left ideal $\mathcal{I}_{\text{can}}$.

\subsubsection{$\widehat{\mathcal{D}}^\hbar_A$-modules of exponential type}

As explained above, the completed Rees polynomial $\widehat{\mathcal{D}}^\hbar_A$-module $\widehat{\mathcal{M}}_A^{\hbar}$ is cyclic and generated by $1$. In the following we will consider a more general class of $\widehat{\mathcal{D}}^\hbar_A$-modules, which we call ``of exponential type''.

\begin{Def}\label{d:exp}
Let
\begin{equation}
Z = \exp\left(\sum_{n=0}^\infty \hbar^{n-1} q^{(n+1)}(x_A)  \right)
\end{equation}
for some polynomials $q^{(n+1)}(x_A)$ of degree $\leq n+1$. We define $\widehat{\mathcal{M}}_A^{\hbar}Z$ to be the cyclic left $\widehat{\mathcal{D}}^\hbar_A$-module generated by $Z$, where the action of $\widehat{\mathcal{D}}^\hbar_A$ on $\widehat{\mathcal{M}}_A^{\hbar}Z$ is the standard action of differential operators on polynomials and exponentials of polynomials. It is clear that it is a well defined $\widehat{\mathcal{D}}^\hbar_A$-module, because of the degree condition on the polynomials $q^{(n+1)}$. We call such modules \emph{of exponential type}.\footnote{\label{f:hbar}We note here that the argument of the exponential is not in $\widehat{\mathcal{M}}_A^\hbar$, since the degree of the polynomials is two more than the power of $\hbar$. But this is fine, as after acting with differential operators on $Z$ using the standard action of differential operators, we get polynomials in $\widehat{\mathcal{M}}_A^\hbar$ times $Z$, as stated.}
\end{Def}

Note that given a $\widehat{\mathcal{D}}^\hbar_A$-module of exponential type, we can always uniquely choose the generator $Z$ to satisfy the property that $Z \big|_{x_A=0} = 1$, i.e. $q^{(n+1)}(0) = 0$.

\subsection{Transvections, twisted $\widehat{\mathcal{D}}^\hbar_A$-modules, and $\widehat{\mathcal{D}}^\hbar_A$-modules of exponential type}

\subsubsection{Transvections}

\label{s:transvection}

We now define an important class of automorphisms of $\widehat{\mathcal{D}}^\hbar_A$, which we call ``transvections''.\footnote{The name ``transvection'' comes from Section 4 of \cite{bkk}, suitably generalized to completed Rees Weyl algebras.} Those will play a key role in the story of Airy ideals.

\begin{Def}\label{d:transvection}
Define the map $\phi$ that acts on $\widehat{\mathcal{D}}^\hbar_A$ as $\phi: (\hbar,  \hbar x_a, \hbar \partial_a) \mapsto (\hbar, \hbar x_a, \bar H_a)$, for all $ a\in A$,  with
\begin{equation}\label{eq:formbarH}
\bar H_a = \hbar \partial_a + \sum_{n=0}^\infty \hbar^n \partial_a q^{(n+1)}(x_A)
\end{equation}
for some polynomials $q^{(n+1)}(x_A)$ of degree $\leq n+1$. We call $\phi$ a \emph{transvection}. We say that it is  \emph{stable} if $q^{(1)} = q^{(2)} = 0$.
\end{Def}

\begin{Rem}
We note here that $ \partial_a q^{(n+1)}(x_A)$ in \eqref{eq:formbarH} means the derivative of the polynomial $q^{(n+1)}(x_A)$ with respect to the variable $x_a$, not the product of $\partial_a$ and $q^{(n+1)}(x_A)$ in the Weyl algebra. Equivalently, we could write $\bar H_a$ as
\begin{equation}
\bar H_a = \hbar \partial_a + \sum_{n=0}^\infty \hbar^n [\partial_a, q^{(n+1)}(x_A)],
\end{equation}
where on the right-hand-side we now mean the commutator with respect to the product in the Weyl algebra.
\end{Rem}

\begin{Lem}
$\phi$ is an automorphism of $\widehat{\mathcal{D}}^\hbar_A$.
\end{Lem}

\begin{proof}
For any $P \in \widehat{\mathcal{D}}^\hbar_A$, it is clear that $\phi(P) \in \widehat{\mathcal{D}}^\hbar_A$. Furthermore, the map $\phi$ preserves the commutation relations between the generators of the Rees Weyl algebra, since $[\bar H_a, \hbar x_b] = \hbar^2 \delta_{ab}$ and
\begin{equation}
\begin{split}
[\bar H_a, \bar H_b] =& \left[  \hbar \partial_a + \sum_{n=0}^\infty \hbar^n \partial_a q^{(n+1)}(x_A),  \hbar \partial_b + \sum_{n=0}^\infty \hbar^n \partial_b q^{(n+1)}(x_A) \right]\\
=& \sum_{n=0}^\infty \hbar^{n+1} \left( \partial_a \partial_b q^{(n+1)}(x_A)  - \partial_b \partial_a q^{(n+1)}(x_A)  \right)\\
=& 0.
\end{split}
\end{equation}
\end{proof}

We can think of transvections as conjugations. Indeed, for any $P \in \widehat{\mathcal{D}}^\hbar_A$, we can think of $\phi(P)$ as being given by
\begin{equation}
\phi(P) = \exp\left(- \sum_{n=0}^\infty \hbar^{n-1} q^{(n+1)}(x_A) \right) P \exp\left(\sum_{n=0}^\infty \hbar^{n-1} q^{(n+1)}(x_A)  \right),
\end{equation}
where multiplication here is understood as multiplication in the Rees Weyl algebra (after formally expanding the exponentials). Using standard properties of derivatives of exponentials, it is clear that this is equivalent to the map specified above.\footnote{As in footnote \ref{f:hbar}, we note here that the argument of the exponential is not in $\widehat{\mathcal{M}}_A^\hbar$, but this is not a problem.}

\subsubsection{Twisted polynomial $\widehat{\mathcal{D}}^\hbar_A$-modules}

Now consider the Rees polynomial $\widehat{\mathcal{D}}^\hbar_A$-module $\widehat{\mathcal{M}}_A^\hbar$. Given a transvection $\phi: \widehat{\mathcal{D}}^\hbar_A \to \widehat{\mathcal{D}}^\hbar_A$, we can construct a twisted left $\widehat{\mathcal{D}}^\hbar_A$-module $\prescript{\phi}{}{\widehat{\mathcal{M}}}_A^\hbar$. Thinking of the transvection as a conjugation, the action on the twisted module is given by
\begin{align}
P \cdot_{\phi} v =& \phi^{-1}(P) \cdot v  \nonumber\\
=&   \exp\left( \sum_{n=0}^\infty \hbar^{n-1} q^{(n+1)}(x_A) \right) P \exp\left(- \sum_{n=0}^\infty \hbar^{n-1} q^{(n+1)}(x_A)  \right) \cdot v.
\label{eq:actiontv}
\end{align}
Since $\widehat{\mathcal{M}}_A^\hbar$ is a cyclic $\widehat{\mathcal{D}}^\hbar_A$-module generated by $1$, we know that the twisted module $\prescript{\phi}{}{\widehat{\mathcal{M}}}_A^\hbar$ is also cyclic and generated by $1$.

Furthermore, the annihilator of $1$ in the twisted module is 
\begin{equation}
\prescript{\phi}{}{\text{Ann}}_{\widehat{\mathcal{D}}_A^\hbar}(1) = \phi\left( \text{Ann}_{\widehat{\mathcal{D}}_A^\hbar}(1) \right) = \phi( \mathcal{I}_{\text{can}}),
\end{equation}
where we used  \eqref{eq:annih1} for the annihilator of $1$ in the polynomial module. In other words, it is the image of the canonical left ideal generated by the derivatives $\hbar \partial_a$ under the automorphism $\phi$.
From the definition of transvections (Definition \ref{d:transvection}), we obtain that
\begin{equation}
\prescript{\phi}{}{\text{Ann}}_{\widehat{\mathcal{D}}_A^\hbar}(1) =  \phi( \mathcal{I}_{\text{can}} ) =  \left \{ \sum_{a \in A} c_a \bar H_a\ |\ c_a \in \widehat{\mathcal{D}}^\hbar_A \right \},
\end{equation}
with the $\bar H_a$ defined in \eqref{eq:formbarH}. If we denote this ideal by $\mathcal{I}$, we conclude that the twisted module $\prescript{\phi}{}{\widehat{\mathcal{M}}}_A^\hbar$ is canonically isomorphic to $\widehat{\mathcal{D}}_A^\hbar / \mathcal{I}$. 

We can summarize these statements in the following Lemma.

\begin{Lem}\label{l:twisted}
Let $\mathcal{I} \subset \widehat{\mathcal{D}}_A^\hbar$ be the left ideal
$
\mathcal{I}  = \left \{ \sum_{a \in A} c_a \bar H_a\ |\ c_a \in \widehat{\mathcal{D}}^\hbar_A \right \},
$
where
\begin{equation}
\bar H_a = \hbar \partial_a + \sum_{n=0}^\infty \hbar^n \partial_a q^{(n+1)}(x_A)
\end{equation}
for some polynomials $q^{(n+1)}(x_A)$ of degree $\leq n+1$. Then $\widehat{\mathcal{D}}_A^\hbar / \mathcal{I}$ is a cyclic left module canonically isomorphic to the twisted module $\prescript{\phi}{}{\widehat{\mathcal{M}}}_A^\hbar$, where $\phi: \widehat{\mathcal{D}}_A^\hbar  \to \widehat{\mathcal{D}}_A^\hbar $ is the transvection $\phi: (\hbar, \hbar x_a, \hbar \partial_a) \mapsto (\hbar, \hbar x_a, \bar H_a)$.
\end{Lem}

\subsubsection{$\widehat{\mathcal{D}}^\hbar_A$-modules of exponential type}

As usual, we can think of this result from the point of view of differential equations. The left ideal $\mathcal{I} = \left \{ \sum_{a \in A} c_a \bar H_a\ |\ c_a \in \widehat{\mathcal{D}}^\hbar_A \right \} $ is the annihilator of $1$ in the twisted module $\prescript{\phi}{}{\widehat{\mathcal{M}}}_A^\hbar$. In other words, $Z'=1$ is a solution to the equations $\mathcal{I} \cdot_\phi Z' = 0$. In fact, as before, it is the unique solution if we impose the initial condition $Z' \big|_{x_A = 0} = 1$.

However, from the viewpoint of differential equations, this is not so nice, because the action of $\mathcal{I}$ on $Z'$ here is the twisted action $\cdot_\phi$, not the standard action of differential operators. Fortunately, since $\phi$ is a transvection, we can think of it as conjugation, and the action can be written as in \eqref{eq:actiontv}. This means that, instead of thinking of $ \widehat{\mathcal{D}}_A^\hbar / \mathcal{I}$ as the cyclic twisted module $\prescript{\phi}{}{\widehat{\mathcal{M}}}_A^\hbar$, we can think of it as the unique (untwisted) $\widehat{\mathcal{D}}_A^\hbar$-module of exponential type $\widehat{\mathcal{M}}_A^\hbar Z$ generated by
\begin{equation}\label{eq:ZZ}
Z =  \exp\left(- \sum_{n=0}^\infty \hbar^{n-1} q^{(n+1)}(x_A)  \right) .
\end{equation}
Furthermore, imposing $Z\big|_{x_A = 0} = 1$, we can uniquely choose the generator with $q^{(n+1)}(0) = 0$. $\mathcal{I}$ is of course the annihilator of $Z$.

 In other words, what we have shown is that the $Z$ in \eqref{eq:ZZ} with $q^{(n+1)}(0) = 0$ is the unique exponential solution to the differential equations $\mathcal{I} \cdot Z = 0$ satisfying the initial condition $  Z\big|_{x_A = 0} = 1$.
 This is summarized in the following lemma.

\begin{Lem}\label{l:pf}
Let $\mathcal{I} \subset \widehat{\mathcal{D}}_A^\hbar$ be the left ideal
$
\mathcal{I}  = \left \{ \sum_{a \in A} c_a \bar H_a\ |\ c_a \in \widehat{\mathcal{D}}^\hbar_A \right \},
$
where
\begin{equation}
\bar H_a = \hbar \partial_a + \sum_{n=0}^\infty \hbar^n \partial_a q^{(n+1)}(x_A)
\end{equation}
for some polynomials $q^{(n+1)}(x_A)$ of degree $\leq n+1$. Then $ \widehat{\mathcal{D}}_A^\hbar / \mathcal{I}$ is canonically isomorphic to the module of exponential type $ \widehat{\mathcal{M}}_A^\hbar Z $ with 
\begin{equation}
Z = \exp\left(- \sum_{n=0}^\infty \hbar^{n-1} q^{(n+1)}(x_A)  \right).
\end{equation}
In other words, $Z$ is a solution to the differential equations $\mathcal{I} \cdot Z = 0$, and if we set $q^{(n+1)}(0) = 0$, it is the unique solution satisfying the initial condition $Z\big|_{x_A=0}=1$.
We call $Z$ the \emph{partition function} associated to the left ideal $\mathcal{I}$.
\end{Lem}

This is of course a rather trivial statement here as the differential equations are straightforward to solve; it may look like we rewrote something very easy in a very complicated way (but isn't it part of the fun of doing mathematics? \smiley{} ). In any case, this result will play an important role in the following, which is why we highlight it.

This can also be interpreted from the point of view of integrability. In classical mechanics, we say that a classical system is ``integrable'' if there exists a complete set of Poisson commuting observables; in the quantum world this becomes a complete set of commuting operators. Here, we consider a left ideal $\mathcal{I}$ generated by a complete set of commuting first-order differential operators $\bar H_a$. Integrable systems are interesting because they can in principle be solved; similarly, we found that there always exists a solution to the differential equations $\mathcal{I} \cdot Z = 0$, and it is unique after imposing an initial condition.  The fundamental reason here is because the $\bar H_a$ are related to the $\hbar \partial_a$ by an automorphism of the completed Rees Weyl algebra (a transvection).

\subsection{Airy ideals, Airy modules, and partition functions}

In the previous section we saw that, given a transvection $\phi$ on $\widehat{\mathcal{D}}_A^\hbar $, we can construct a twisted polynomial module $\prescript{\phi}{}{\widehat{\mathcal{M}}}_A^\hbar$, which is canonical isomorphic to $\widehat{\mathcal{D}}_A^\hbar  / \mathcal{I}$ where $\mathcal{I}$ is the ideal generated by a completed set of commuting first-order differential operators of the form \eqref{eq:formbarH}. From the point of view of differential equations, we can instead think of $\widehat{\mathcal{D}}_A^\hbar  / \mathcal{I}$ as a module of exponential type $\widehat{\mathcal{M}}_A^\hbar Z $, with $Z$ the unique exponential solution \eqref{eq:ZZ} to the differential equations $\mathcal{I} \cdot Z =0$ after imposing a suitable initial condition.

This is nice, but in the end, at least from the point of view of differential equations, this is rather trivial; after all, solving the system of equations $\bar H_a \cdot Z = 0$ with $\bar H_a$ of the form \eqref{eq:formbarH} is obvious, and it is clear that \eqref{eq:ZZ} is the unique solution with $Z\big|_{x_A = 0} = 1$. We do not need the fancy ideas of transvections, twisted modules, and modules of exponential type to show this!

The power of the formalism however becomes apparent when we introduce Airy ideals. The idea here is that we introduce a more general class of left ideals $\mathcal{I} \subset \widehat{\mathcal{D}}_A^\hbar $, which we call Airy ideals (traditionally called ``Airy structures'' in the literature). Then, we show that if $\mathcal{I}$ is an Airy ideal, then it is equal to the image of the canonical left ideal $\mathcal{I}_{\text{can}}$ generated by the derivatives $\hbar \partial_a$ for some stable transvection $\phi$. This is rather striking, and far from obvious \emph{a priori}. As a result, Lemmas \ref{l:twisted} and \ref{l:pf} apply; $\widehat{\mathcal{D}}_A^\hbar/ \mathcal{I}$ is isomorphic to a twisted polynomial module or a module of exponential type, depending on the viewpoint. As a result, there exists a unique solution to the differential equations $\mathcal{I} \cdot Z = 0$ after imposing a suitable initial condition. 

We remark that it is absolutely key that we work in the $\hbar$-adic completion of the Rees Weyl algebra here, since the transvection $\phi$ will generally involve operators $\bar H_a$ that are formal power series in $\hbar$. Working within the $\hbar$-adic completion enables us to relate Airy ideals to twisted polynomial modules and modules of exponential types.

\subsubsection{Airy ideals}

Let us now define the concept of an Airy ideal in $\widehat{\mathcal{D}}_A^{\hbar}$.

\begin{Def}\label{d:airy}
Let $\mathcal{I} \subset \widehat{\mathcal{D}}_A^{\hbar}$ be a left ideal. We say that it is an \emph{Airy ideal} (also called \emph{Airy structure} in the literature) if there exists operators $H_a \in \widehat{\mathcal{D}}_A^{\hbar}$, for all $a \in A$, such that:
\begin{enumerate}
\item The collection of operators $\{ H_a \}_{a \in A}$ is bounded (see Definition \ref{d:bounded}).
\item The left ideal $\mathcal{I}$ can be written as
\begin{equation}
\mathcal{I} = \left \{ \sum_{a \in A} c_a  H_a\ | \ c_a \in \widehat{\mathcal{D}}_A^{\hbar} \right \},
\end{equation}
which consists as usual of finite and infinite (if $A$ is countably infinite) $\widehat{\mathcal{D}}_A^{\hbar}$-linear combinations of the $H_a$. 
\item The operators $H_a$ take the form
\begin{equation}\label{eq:form}
H_a = \hbar \partial_a + O(\hbar^2).
\end{equation}
\item The left ideal $\mathcal{I}$ satisfies the property:
\begin{equation}
[ \mathcal{I}, \mathcal{I}] \subseteq \hbar^2 \mathcal{I}.
\end{equation}
\end{enumerate}
\end{Def}

\begin{Rem}
Before we move on, we remark that Condition (4) is non-trivial. First, remark that, trivially, any left ideal $\mathcal{I} \subseteq \widehat{\mathcal{D}}_A^{\hbar}$ satisfies
\begin{equation}
[ \mathcal{I}, \mathcal{I} ] \subseteq \mathcal{I},
\end{equation}
since $P Q - Q P \in \mathcal{I}$ for all $P,Q \in \mathcal{I}$. Furthermore, from \eqref{eq:hbar2}, any left ideal $\mathcal{I} \subseteq \widehat{\mathcal{D}}_A^{\hbar}$ satisfies
\begin{equation}
[\mathcal{I},\mathcal{I}] \subseteq \hbar^2 \widehat{\mathcal{D}}_A^{\hbar}.
\end{equation}
Combining the two statements, we conclude that the commutator of any two elements in an arbitrary left ideal $\mathcal{I}$ is an element of the ideal $\mathcal{I}$ that starts at $O(\hbar^2)$. However, this does not mean that it is equal to $\hbar^2$ times an element of the ideal $\mathcal{I}$. That it must be so is the content of Condition (4).

As an example, let $A = \{1,2\}$, and consider the left ideal $\mathcal{I}$ generated by $H_1 = \hbar \partial_1$ and $H_2 = \hbar \partial_2 + \hbar^2 x_1$. The commutator of $H_1$ and $H_2$ is $[H_1, H_2 ] = \hbar^3$. It is true that $\hbar^3 \in \mathcal{I}$, since
\begin{equation}
\hbar^3 = \hbar \partial_1 (\hbar \partial_2 + \hbar^2 x_1) - (\hbar \partial_2+ \hbar^2 x_1) (\hbar \partial_1).
\end{equation}
If we single out a power of $\hbar^2$ on the right-hand-side of the commutator, it is also true that $\hbar \in \widehat{\mathcal{D}}_A^{\hbar}$. However, $\hbar  \notin \mathcal{I}$, as is easy to check. Thus this ideal does not satisfy Condition (4).
\end{Rem}

\subsubsection{Airy ideals, transvections and twisted modules}

Since $\mathcal{I} \subset \widehat{\mathcal{D}}_A^{\hbar}$ is a left ideal, it is clear that $\widehat{\mathcal{D}}_A^{\hbar} / \mathcal{I}$ is a cyclic left $\widehat{\mathcal{D}}_A^{\hbar}$-module. What is not clear however is how it relates to the Rees polynomial $\widehat{\mathcal{D}}_A^{\hbar}$-module $\widehat{\mathcal{M}}_A^\hbar$. This connection is the fundamental theorem in the theory of Airy ideals.

\begin{Th}\label{t:airy}
Let $\mathcal{I} \subset \widehat{\mathcal{D}}_A^{\hbar}$  be an Airy ideal. There there exists a stable transvection $\phi:  \widehat{\mathcal{D}}_A^{\hbar} \to  \widehat{\mathcal{D}}_A^{\hbar}$ (see Definition \ref{d:transvection}) such that $\mathcal{I} = \phi(\mathcal{I}_{\text{can}})$, where $\mathcal{I}_{\text{can}}$ is the left  ideal generated by the derivatives $\hbar \partial_a$. As a result, $\widehat{\mathcal{D}}_A^{\hbar} / \mathcal{I}$ is a cyclic left module canonically isomorphic to the twisted polynomial module $\prescript{\phi}{}{\widehat{\mathcal{M}}}_A^\hbar$.
\end{Th}

This is a powerful theorem. What it means is that we can find a complete set of commuting first-order differential operators $\bar H_a$ of the form \begin{equation}
\bar H_a = \hbar \partial_a + \sum_{n=2}^\infty \hbar^n \partial_a q^{(n+1)}(x_A),
\end{equation}
for some polynomials $q^{(n+1)}(x_A)$ of degree $\leq n+1$, such that the Airy ideal $\mathcal{I}$ can be rewritten as
\begin{equation}
\mathcal{I}  = \left \{ \sum_{a \in A} c_a \bar H_a\ |\ c_a \in \widehat{\mathcal{D}}^\hbar_A \right \}.
\end{equation}
This is highly non-trivial, as the original Airy ideal $\mathcal{I}$ will not usually be presented in this form.

To prove Theorem \ref{t:airy} we will first prove a series of lemmas. The first three lemmas do not require Condition (4) in the definition of Airy ideals Definition \ref{d:airy}. The fourth lemma highlights the crucial role played by Condition (4).

\begin{Lem}\label{l:poly}
Let $\mathcal{I} \subset \widehat{\mathcal{D}}_A^{\hbar}$  be a left ideal satisfying conditions (1)--(3) of Definition \ref{d:airy}. Then for any $P \in  \widehat{\mathcal{D}}_A^{\hbar}$, we can write
\begin{equation}
P = \sum_{n=0}^\infty \hbar^n p^{(n)}(x_A) + Q
\end{equation}
for some polynomials $p^{(n)}(x_A)$ of degree $\leq n$ and some $Q \in \mathcal{I}$.
\end{Lem}

\begin{proof}
Let $P \in  \widehat{\mathcal{D}}_A^{\hbar}$. We can write
\begin{equation}
P = p^{(0,0)} + \hbar \left( p^{(1,1)} + \sum_{b \in A} p^{(1,0)}_b \partial_b \right) + \hbar^2 \left(   p^{(2,2)} + \sum_{b \in A} p^{(2,1)}_{b} \partial_b + \sum_{b,c \in A} p^{(2,0)}_{bc} \partial_b \partial_c \right) + O(\hbar^3),
\end{equation}
where the $p^{(m,k)}_{\cdots}$ are polynomials of degree $\leq k$ (we removed the dependence in $x_A$ for clarity). 

The idea is simple. Since $\mathcal{I}$ is an Airy ideal, it is generated by a bounded collection of operators $\{H_a \}_{a \in A}$ of the form $H_a = \hbar \partial_a + O(\hbar^2)$. So for each term in $P$ that is not polynomial, we can replace the right-most derivative $\hbar \partial_a$ by $H_a$, up to higher order terms in $\hbar$. Applying this procedure recursively order by order in $\hbar$, we will end up rewriting $P$ as a polynomial plus an operator in the ideal $\mathcal{I}$. We see here that it is key that we are working in the $\hbar$-adic completion of the Rees Weyl algebra, otherwise we would not be allowed to keep going order by order in $\hbar$ forever.

More precisely, we start at $O(\hbar)$. We use $\hbar \partial_b = H_b + O(\hbar^2)$ to rewrite
\begin{equation}
 \hbar \left( p^{(1,1)} + \sum_{b \in A} p^{(1,0)}_b \partial_b \right)  =  \hbar p^{(1,1)} + \sum_{b \in A} p^{(1,0)}_b H_b + O(\hbar^2).
\end{equation}
The first term is a polynomial term, and the second term is in $\mathcal{I}$. The procedure however created new terms at the next order, $O(\hbar^2)$, which we must study further. If we write $H_a$ as
\begin{equation}\label{eq:beginning}
H_a = \hbar \partial_a + \hbar^2 \left( g^{(2,2)}_a + \sum_{b \in A} g^{(2,1)}_{a;b} \partial_b + \sum_{b,c \in A} g^{(2,0)}_{a;bc} \partial_b \partial_c \right) + O(\hbar^3),
\end{equation}
where the $g^{(2,i)}_{\cdots}(x_A)$ are polynomials of degree $\leq i$, then the terms of $O(\hbar^2)$ created by the procedure above take the form
\begin{equation}\label{eq:hbar3}
- \hbar^2 \left( \sum_{b \in A} p^{(1,0)}_b g^{(2,2)}_b + \sum_{b,c \in A}  p^{(1,0)}_b g^{(2,1)}_{b;c} \partial_c + \sum_{b,c,d \in A} p^{(1,0)}_b  g^{(2,0)}_{b;cd} \partial_c \partial_d \right).
\end{equation}
When the index set $A$ is countably infinite, we need to make sure that the terms in brackets do not involve infinite divergent sums, and are all in $F_2 \mathcal{D}_A$ (with respect to the Bernstein filtration, see Definition \ref{d:bernstein}). We note that because the collection of operators $\{H_a\}_{a \in A}$ is bounded, all polynomials $g^{(2,k)}_{a;\cdots}$ vanish for all but finitely many $a \in A$. Therefore, the sums over $b \in A$ collapse to finite sums and are well defined. The remaining sums over $c,d \in A$ are potentially infinite, but they come with derivatives in these indices, and hence all terms in brackets in \eqref{eq:hbar3} are in $\mathcal{D}_A$. Furthermore, looking at the degree of the terms with respect to the Bernstein filtration, we see that they are all in $F_2 \mathcal{D}_A$. 

Now repeat the procedure for all terms at $O(\hbar^2)$, including the newly obtained terms, keeping the polynomial terms and replacing the right-most derivatives in the other terms by $H$'s up to terms of higher order in $\hbar$. The result will be a polynomial term at $O(\hbar^2)$, plus a term that is in the ideal $\mathcal{I}$, plus corrections at higher order. The argument above shows that the corrections are well defined. Then keep applying this procedure recursively, order by order in $\hbar$. In the end, all that remains are polynomial terms plus terms in the ideal $\mathcal{I}$. That is, we conclude that we can write
\begin{equation}
P= \sum_{n=0}^\infty \hbar^n p^{(n)}(x_A) + Q
\end{equation}
for some polynomials $p^{(n)}(x_A)$ of degree $\leq n$ and some $Q \in \mathcal{I}$.
For instance, 
\begin{equation}
p^{(0)} = p^{(0,0)}, \qquad p^{(1)} = p^{(1,1)}, \qquad p^{(2)} = p^{(2,2)} - \sum_{b \in A} p^{(1,0)}_b g^{(2,2)}_b,
\end{equation}
 and so on and so forth. 
\end{proof}

\begin{Lem}\label{l:barh}
Let $\mathcal{I} \subset \widehat{\mathcal{D}}_A^{\hbar}$  be a left ideal satisfying conditions (1)--(3) of Definition \ref{d:airy}. Then there exist operators $\bar H_a \in \mathcal{I}$, for all $a \in A$, of the form
\begin{equation}\label{eq:formhbar}
\bar H_a = \hbar \partial_a + \sum_{n=2}^\infty \hbar^n p^{(n)}_a(x_A)
\end{equation}
for polynomials $p^{(n)}_a(x_A)$ of degree $\leq n$. Furthermore, the collection of operators $\{ \bar H_a\}_{a \in A}$ is bounded.
\end{Lem}

\begin{proof}
We know that $H_a = \hbar \partial_a + P_a$ for some $P_a \in  \widehat{\mathcal{D}}_A^{\hbar}$ of $O(\hbar^2)$. From Lemma \ref{l:poly}, we can write
\begin{equation}
H_a = \hbar \partial_a + \sum_{n=2}^\infty \hbar^n p_a^{(n)}(x_A) + Q_a
\end{equation}
for some polynomials $p_a^{(n)}(x_A)$ of degree $\leq n$ and some $Q_a \in \mathcal{I}$. We define
\begin{equation}
\bar H_a = H_a - Q_a =  \hbar \partial_a + \sum_{n=2}^\infty \hbar^n p_a^{(n)}(x_A).
\end{equation}
Clearly, $\bar H_a \in \mathcal{I}$. Furthermore, since the collection $\{H_a\}_{a \in A}$ is bounded, the polynomials $p_a^{(n)}(x_A) $ must vanish for all but finitely many $a \in A$, and hence the collection $\{ \bar H_a \}_{a \in A}$ is also bounded. 
\end{proof}

\begin{Lem}\label{l:Ibar}
Let $\mathcal{I} \subset \widehat{\mathcal{D}}_A^{\hbar}$  be a left ideal satisfying conditions (1)--(3) of Definition \ref{d:airy}, and $\bar I \subseteq I$ be the left ideal generated by the $\bar H_a$ of Lemma \ref{l:barh}. Then $\bar I = I $. 

In other words, we can think of $I$ as being generated by the $\bar H_a$ instead of the $H_a$:  $\mathcal{I} =  \{ \sum_{a \in A} c_a \bar H_a\ |\ c_a \in  \widehat{\mathcal{D}}_A^{\hbar} \}$.
\end{Lem}

\begin{proof}
By definition, $\bar H_a = H_a - Q_a$ for some $Q_a \in \mathcal{I}$. We can write $Q_a = \sum_{b \in A} p_{a b} H_b$
for some $p_{a b} \in \widehat{\mathcal{D}}_A^{\hbar}$ . In fact, from the proof of Lemma \ref{l:barh}, we know that $p_{ab} = O(\hbar)$.

We can write
\begin{equation}
\bar H_a = H_a  - \sum_{b \in A} p_{a b} H_b
\end{equation}
 But then, $H_b = \bar H_b + Q_b = \bar H_b + \sum_{c \in A} p_{b c} H_c$, and thus we get
\begin{equation}
\bar H_a = H_a - \sum_{b \in A} p_{a b} \bar H_b -  \sum_{b,c \in A} p_{a b} p_{ bc} H_c.
\end{equation}
We note here that since both $\{H_a\}_{a \in A}$ and $\{ \bar H_a\}_{a \in A}$ are bounded, the collection $\{ Q_a\}_{a \in A}$ is also bounded, and hence the sum over $b \in A$ in the third term on the right-hand-side is well defined. Furthermore, since $p_{ab} = O(\hbar)$, $p_{ab} p_{bc} = O(\hbar^2)$.

Continuing this process recursively, we end up with the statement that
\begin{equation}
\bar H_a = H_a - \bar Q_a
\end{equation}
where $\bar Q_a$ is an infinite sum of terms that are linear combinations of the $\bar H_b$ with coefficients starting at higher and higher order in $\hbar$. Thus, for a finite power of $\hbar$, only a finite number of terms contribute, and the result is that $\bar Q_a \in \bar{\mathcal{I}}$.  It follows that $H_a \in \bar{\mathcal{I}}$, and hence $\mathcal{I} \subseteq \bar{\mathcal{I}}$. We conclude that $\mathcal{I} = \bar{\mathcal{I}}$.
\end{proof}

So far we have not used at all Condition (4) in the definition of Airy ideals Definition \ref{d:airy}. This condition is crucial; imposing Condition (4) implies that there are no non-zero polynomials in an Airy ideal $\mathcal{I}$. In particular, it implies that the operators $\bar H_a$ commute with each other, which in turn implies the existence of a stable transvection that relates $\mathcal{I}$ to the canonical left ideal $\mathcal{I}_{\text{can}}$, as we will see.

\begin{Lem}\label{l:nononzero}
Let  $\mathcal{I} \subset \widehat{\mathcal{D}}_A^{\hbar}$ be an Airy ideal. Then there are no non-zero polynomials in $\mathcal{I}$. 
\end{Lem}

\begin{proof}
By Lemma \ref{l:Ibar}, we think of $I$ as $\mathcal{I} =  \{ \sum_{a \in A} c_a \bar H_a\ |\ c_a \in  \widehat{\mathcal{D}}_A^{\hbar} \}$, with $\bar H_a$ of the form 
\begin{equation}\label{eq:hbarrr}
\bar H_a = \hbar \partial_a + \sum_{n=2}^\infty \hbar^n p^{(n)}_a(x_A).
\end{equation}

We prove the Lemma by induction on the power of $\hbar$.  Let $N \geq 2$ be a positive integer. The induction hypothesis is that all polynomials in $\mathcal{I}$ start at least at $O(\hbar^N)$. Then we show that it implies that they must start at least at order $O(\hbar^{N+1})$. By induction on $N$, this means that all polynomials  in $\mathcal{I}$ must vanish. 

The base case for the induction is obvious. We need to show that all polynomials in $\mathcal{I}$ must start at least at $O(\hbar^2)$. But since $\mathcal{I}$ is generated by $\bar H_a$ of the form \eqref{eq:hbarrr}, it is clearly impossible to get a polynomial with a $\hbar^0$ constant term or a $\hbar^1$ linear term as a linear combination of $\bar H_a$'s. 

Now assume that all polynomials in $\mathcal{I}$ start at least at $O(\hbar^N)$. The commutator of the $\bar H_a$ is:
\begin{align}
[\bar H_a, \bar H_b ] =& \left[ \hbar \partial_a + \sum_{n=2}^\infty \hbar^n p_{a}^{(n)}(x_A), \hbar \partial_b + \sum_{n=2}^\infty \hbar^n p_{b}^{(n)}(x_A) \right] \nonumber\\
=&\hbar^2  \sum_{n=2}^\infty \hbar^{n-1} \left( \partial_a p_{b}^{(n)}(x_A) - \partial_b p_{a}^{(n)}(x_A) \right).
\end{align}
By Condition (4) of Definition \ref{d:airy}, we know that
\begin{equation}
 \sum_{n=2}^\infty \hbar^{n-1} \left( \partial_a p_{b}^{(n)}(x_A) - \partial_b p_{a}^{(n)}(x_A) \right) \in \mathcal{I}.
\end{equation}
By assumption, this must start at least at $O(\hbar^N)$, so we must have that
\begin{equation}\label{eq:conditions}
\partial_a p_{b}^{(n)}(x_A) = \partial_b p_{a}^{(n)}(x_A), \qquad \text{ for all $a,b \in A$ and $n = 2,\ldots,N$.}
\end{equation}
Assuming first that $A$ is a finite index set, by Poincare's lemma we conclude that there exists polynomials $q^{(n+1)}(x_A)$ of degree $\leq n+1$ such that
\begin{equation}
p_{a}^{(n)}(x_A) = \partial_a q^{(n+1)}(x_A),  \qquad \text{ for all $a\in A$ and $n = 2,\ldots,N$.}
\end{equation}
Moreover, if we require that $q^{(n+1)}(0) = 0$, then the polynomials are uniquely fixed.

If $A$ is a countably infinite index set, then Poincare's lemma still holds, but the $q^{(n+1)}(x_A)$, which will still be of degree $\leq n+1$, could \emph{a priori} involve infinite linear combinations of monomials of the same degree. However, since the collection of operators $\{ \bar H_a \}_{a \in A}$ is bounded, we know that for a fixed $n$, the polynomials $p_a^{(n)}(x_A)$ vanish for all but finitely many $a \in A$. This means that for each $n$ the conditions \eqref{eq:conditions} become a finite system in a finite numbers of variables, and thus we conclude by Poincare's lemma that the $q^{(n+1)}(x_A)$ will be polynomials.

As result, we conclude that we can write
\begin{equation}
\bar H_a = \hbar \partial_a + \sum_{n=2}^N \hbar^n \partial_a q^{(n+1)}(x_A) + \sum_{n = N+1}^\infty \hbar^n p_n(x_A).
\end{equation}
Now let
\begin{equation}
\phi_N: (\hbar, \hbar x_a, \hbar \partial_a) \mapsto (\hbar, \hbar x_a, \hbar \partial_a - \sum_{n=2}^N \hbar^n \partial_a q^{(n+1)}(x_A) ).
\end{equation}
It is a stable transvection, and 
\begin{equation}\label{eq:phiH}
\phi_N(\bar H_a) = \hbar \partial_a +  \sum_{n = N+1}^\infty \hbar^n p_n(x_A).
\end{equation}
Under this automorphism, the ideal $\mathcal{I}$ is mapped to the ideal $\phi_N(\mathcal{I})$ generated by the $\phi_N(\bar H_a)$ above. Now suppose that $P$ is a polynomial in $\mathcal{I}$. By definition of the transvection, $\phi_N(P) = P$, and thus $P$ must also be in $\phi_N(\mathcal{I})$. But looking at the form of the generators $\phi_N(\bar H_a)$ in \eqref{eq:phiH}, it is clear that any polynomial in $\phi_N(\mathcal{I})$ must start at least at $O(\hbar^{N+1})$. Therefore, $P$ must start at least at $O(\hbar^{N+1})$, which completes the induction.

\end{proof}

With these four lemmas under our belt, the proof of Theorem \ref{t:airy} is straightforward.

\begin{proof}[Proof of Theorem \ref{t:airy}]

By Lemma \ref{l:Ibar}, we can write $\mathcal{I}$ as $\mathcal{I} =  \{ \sum_{a \in A} c_a \bar H_a\ |\ c_a \in  \widehat{\mathcal{D}}_A^{\hbar} \}$ for some operators $\bar H_a \in \mathcal{I}$ of the form
\begin{equation}
\bar H_a = \hbar \partial_a + \sum_{n=2}^\infty \hbar^n p^{(n)}_a(x_A),
\end{equation}
with the $p^{(n)}_a(x_A)$ polynomials of degree $\leq n$. Since there are no non-zero polynomials in $\mathcal{I}$ (Lemma \ref{l:nononzero}), and that the commutator $[\bar H_a, \bar H_b] \in \mathcal{I}$ is a polynomial, we must have
\begin{equation}
[ \bar H_a, \bar H_b] = 0 \qquad \text{for all $a,b \in A$.}
\end{equation}
Using Poincare's Lemma as in the proof of Lemma \ref{l:nononzero}, we conclude that we can rewrite the differential operators as
\begin{equation}
\bar H_a = \hbar \partial_a + \sum_{n = 2}^\infty \hbar^n \partial_a q^{(n+1)}(x_A) ,
\end{equation}
for polynomials $q^{(n+1)}(x_A)$ of degree $\leq n+1$ with $q^{(n+1)}(0) = 0$.
We have thus shown that $\mathcal{I}$ is the image of $\mathcal{I}_{\text{can}} = \left \{ \sum_{a \in A} c_a \hbar \partial_a\ |\ c_a \in \widehat{\mathcal{D}}^\hbar_A \right \}$ under the stable transvection $\phi: (\hbar, \hbar x_a, \hbar \partial_a) \mapsto (\hbar, \hbar x_a, \bar H_a)$. By Lemma \ref{l:twisted}, it follows that  $\widehat{\mathcal{D}}_A^{\hbar} / \mathcal{I}$ is a cyclic left module canonically isomorphic to the twisted polynomial module $\prescript{\phi}{}{\widehat{\mathcal{M}}}_A^\hbar$.
\end{proof}

\subsubsection{Partition function}

As usual, we can reformulate Theorem \ref{t:airy} in the language of differential equations.  Lemma \ref{l:pf} directly implies the following corollary, which is the existence and uniqueness theorem at the foundation of the theory of Airy ideals (or Airy structures), first proved by Kontsevich and Soibelman in \cite{ks}.

\begin{Cor}\label{c:pf}
Let $\mathcal{I} \subset \widehat{\mathcal{D}}_A^\hbar$ be an Airy ideal. Then there exists a unique module of exponential type $\widehat{\mathcal{M}}_A^\hbar Z$ that is canonically isomorphic to $\widehat{\mathcal{D}}_A^\hbar / \mathcal{I}$. Furthermore, if we impose the initial condition $Z\big|_{x_A=0}=1$ on the generator, then it is uniquely fixed and takes the form
\begin{equation}\label{eq:pf}
Z = \exp \left( \sum_{\substack{g \in \frac{1}{2} \mathbb{N}, n \in \mathbb{N}^* \\ 2g-2+n>0}} \hbar^{2g-2+n} F_{g,n}(x_A) \right)
\end{equation}
for some polynomials $F_{g,n}(x_A)$ homogeneous of degree $n$ with $F_{g,n}(0) = 0$.\footnote{The use of homogeneous polynomials $F_{g,n}(x_A)$ of degree $n$ instead of the polynomials $q^{(k+1)}(x_A)$ of degree $\leq k+1$ previously used in Lemma \ref{l:pf} is simply to connect with the existing literature on the topic, but  it is straightforward to show that it is equivalent: the polynomials $q^{(k+1)}(x_A)$ of degree $\leq k+1$ are reconstructed as $q^{(k+1)}(x_A) = \sum_{\substack{g \in \frac{1}{2} \mathbb{N}, n \in \mathbb{N}^* \\ 2g-1+n=k}} F_{g,n}(x_A) $. } In other words, $Z$ is the unique solution to the differential equations $\mathcal{I} \cdot Z = 0$ satisfying the initial condition $Z\big|_{x_A=0}=1$.
We call $Z$ the \emph{partition function} of the Airy ideal $\mathcal{I}$.
\end{Cor}

Equivalently, $Z$ is the unique solution to the differential equations $H_a Z = 0$, for all $a\in A$, where the $H_a$ generated the Airy ideal $\mathcal{I}$.

\begin{Rem}
We note here that we could have proven existence and uniqueness of the partition function directly from the definition of Airy ideals; the proof is fairly straightforward, and goes by induction on $\hbar$. This is what is done in \cite{ks} (see also \cite{bo}). However, our approach of relating Airy ideals to transvections, twisted modules, and modules of exponential type may shed light on what an Airy ideal really is, and why there always exists a unique partition function annihilated by an Airy ideal.
\end{Rem}

\subsubsection{Airy modules}

Given an Airy ideal $\mathcal{I}\subset \widehat{\mathcal{D}}_A^{\hbar}$, we have shown that $\widehat{\mathcal{D}}_A^{\hbar} / \mathcal{I}$ is a cyclic left module canonically isomorphic to the twisted polynomial module $\prescript{\phi}{}{\widehat{\mathcal{M}}}_A^\hbar$, or equivalently to the module of exponential type generated by the partition function $Z$. Turning this around, we can define the notion of an Airy left $\widehat{\mathcal{D}}_A^\hbar$-module.

\begin{Def}\label{d:airymodule}
We say that a cyclic left $\widehat{\mathcal{D}}_A^\hbar$-module is \emph{Airy} if it is generated by an element $v$ whose annihilator $\text{Ann}_{\widehat{\mathcal{D}}_A^\hbar}(v)$ is an Airy ideal.
\end{Def}

It is easy to show that all modules of exponential type (see Definition \ref{d:exp}) are Airy modules.

\begin{Lem}\label{l:expairy}
Let $\widehat{\mathcal{M}}_A^\hbar Z$ be a left $\widehat{\mathcal{D}}_A^\hbar$-module of exponential type, with generator
\begin{equation}
Z = \exp\left(\sum_{n=0}^\infty \hbar^{n-1} q^{(n+1)}(x_A)  \right)
\end{equation}
for some polynomials $q^{(n+1)}(x_A)$ of degree $\leq n+1$. Then $\widehat{\mathcal{M}}_A^\hbar Z$ is an Airy module.
\end{Lem} 

\begin{proof}
This is clear. As we have seen in Lemma \ref{l:pf}, the module of exponential type $\widehat{\mathcal{M}}_A^\hbar Z$ is canonically isomorphic to $\widehat{\mathcal{D}}_A^\hbar / \mathcal{I}$, where 
\begin{equation}
\mathcal{I}  = \left \{ \sum_{a \in A} c_a \bar H_a\ |\ c_a \in \widehat{\mathcal{D}}^\hbar_A \right \},
\end{equation}
with
\begin{equation}
\bar H_a = \hbar \partial_a - \sum_{n=0}^\infty \hbar^n \partial_a q^{(n+1)}(x_A) .
\end{equation}
Since $[\bar H_a, \bar H_b]=0$ for all $a,b \in A$, it is easy to show that $\mathcal{I}$ satisfies the four conditions in the definition of Airy ideals (Definition \ref{d:airy}).

\end{proof}

Similarly, all twisted modules $\prescript{\phi}{}{\widehat{\mathcal{M}}}_A^\hbar$ obtained from a transvection $\phi$ of the completed Rees Weyl algebra are Airy modules.

\begin{Lem}\label{l:twistedairy}
Let $\phi: \widehat{\mathcal{D}}_A^\hbar \to \widehat{\mathcal{D}}_A^\hbar$ be a transvection, and $\prescript{\phi}{}{\widehat{\mathcal{M}}}_A^\hbar$ the $\phi$-twisted polynomial module. Then  $\prescript{\phi}{}{\widehat{\mathcal{M}}}_A^\hbar$ is an Airy module.
\end{Lem}

\begin{proof}
Same argument as in the proof of the previous Lemma.
\end{proof}

\subsubsection{$\hbar$-polynomial and $\hbar$-finite Airy ideals}

The existence of a partition function becomes particularly interesting when the operators $H_a$  live in the subalgebra $\mathcal{D}_A^\hbar$; that is, the $H_a$ are polynomials (instead of formal series) in $\hbar$. Even more interesting is when all $H_a$ are polynomials of degree less than a certain fixed positive integer $N$. We thus formulate the following definition.

\begin{Def}\label{d:hbarfinite}
Let $\mathcal{I} \subset \widehat{\mathcal{D}}_A^{\hbar}$ be an Airy ideal. We say that it is \emph{$\hbar$-polynomial} if there exists $H_a \in \mathcal{D}_A^\hbar$ for all $a \in A$ (i.e., they are polynomials in $\hbar$) satisfying the conditions of Definition \ref{d:airy}.

We say that $\mathcal{I}$ is \emph{$\hbar$-finite} if it is $\hbar$-polynomial, and there exists a positive integer $N$ such that all $H_a$ are polynomials in $\hbar$ of degree $\leq N$. We call the smallest such $N$ the \emph{$\hbar$-degree} of $\mathcal{I}$.

We also say that an Airy left $\widehat{\mathcal{D}}_A^\hbar$-module is $\hbar$-polynomial (resp. $\hbar$-finite) if it is generated by an element $v$ whose annihilator is $\hbar$-polynomial (resp. $\hbar$-finite).
\end{Def}

For example, an Airy ideal generated by a collection of operators $\{H_a\}_{a \in \mathbb{N}^*}$ that are polynomials of degree $2$ in $\hbar$ is $\hbar$-finite, and has $\hbar$-degree $2$. However, an Airy ideal generated by a collection of operators $\{H_a\}_{a \in \mathbb{N}^*}$ that are polynomials of degree $a$ in $\hbar$ is $\hbar$-polynomial, but not $\hbar$-finite (as the $\hbar$-degree of the $H_a$ keeps increasing as $a$ increases).
%
%\begin{Rem}
%We note that while all $\hbar$-adic integrable ideals are Airy, it is probably not true that all $\hbar$-adic integrable ideals are $\hbar$-polynomial Airy ideals. The $\hbar$-polynomial, and particularly the $\hbar$-finite, conditions appear to be rather strong conditions. 
%
%It is also worth remarking that it is not necessarily easy to determine whether an Airy ideal is $\hbar$-polynomial or not; it is not because one first presents an Airy ideal as $\widehat{\mathcal{D}}_A^{\hbar}$-linear combinations of operators $H_a$ that are not polynomial in $\hbar$ that the ideal cannot be $\hbar$-polynomial. There could be another collection of operators, this time polynomial in $\hbar$, such that all elements of the ideals can be written as $\widehat{\mathcal{D}}_A^{\hbar}$-linear combinations of these new operators.  It would be good to have a better and more intrinsic characterization of $\hbar$-polynomial and $\hbar$-finite Airy ideals.
%\end{Rem}

Why are $\hbar$-finite Airy ideals particularly interesting? We saw in Lemma \ref{l:expairy} that all left $\widehat{\mathcal{D}}_A^\hbar$-modules of exponential type $\widehat{\mathcal{M}}_A^\hbar Z$ are Airy modules. In other words, any exponential of the form
\begin{equation}\label{eq:Zform}
Z = \exp \left( \sum_{\substack{g \in \frac{1}{2} \mathbb{N}, n \in \mathbb{N}^* \\ 2g-2+n>0}} \hbar^{2g-2+n} F_{g,n}(x_A) \right)
\end{equation}
is the partition function for some Airy ideal. Indeed, the Airy ideal is given by
\begin{equation}
\mathcal{I}  = \left \{ \sum_{a \in A} c_a \bar H_a\ |\ c_a \in \widehat{\mathcal{D}}^\hbar_A \right \},
\end{equation}
with
\begin{equation}
\bar H_a = \hbar \partial_a - \sum_{\substack{g \in \frac{1}{2} \mathbb{N}, n \in \mathbb{N}^* \\ 2g-2+n>0}} \hbar^{2g-1+n} \partial_a F_{g,n}(x_A) .
\end{equation}
This is a rather trivial statement.

But are all $Z$  of the form \eqref{eq:Zform} partition functions for Airy ideals that are $\hbar$-finite? In other words, can we always rewrite the Airy ideal $\mathcal{I}$ above as an ideal generated by operators $H_a$ that are $\hbar$-polynomials of degree $\leq N$ for some $N$? Equivalently, which partition functions of the form \eqref{eq:Zform} satisfy differential equations $H_a Z = 0$ for $\hbar$-polynomial operators $H_a$ of degree $\leq N$ for some $N$ that generate an Airy ideal?

This is a very interesting question. The existence of such operators $H_a$ implies a recursive structure for the polynomials $F_{g,n}(x_A)$, which is the foundation of the story of topological recursion reformulated in the language of Airy structures. See \cite{ks, ABCD,Airy} for more on this recursive structure. Traditionally, the usual topological recursion formula is obtained directly by applying the differential operators $H_a$ on the partition function $Z$, and setting the result to $0$. This gives a recursion for the polynomials $F_{g,n}(x_A)$ (or their coefficients) if the operators $H_a$ are $\hbar$-polynomials of degree $\leq N$ for some $N$. From our point of view, this recursive structure is encapsulated in the recursive construction of the commuting first-order operators $\bar H_a$ from the $H_a$. When the $H_a$ are polynomials in $\hbar$, the infinite set of polynomials (order by order in $\hbar$) in the operators $\bar H_a$ will be constructed out of a finite set (in $\hbar$) of polynomials in the $H_a$, and hence will satisfy some recursive structure determined by initial conditions. The two recursions are of course equivalent, as is easy to show.

\subsection{Constructing Airy ideals}

We are interested in constructing Airy ideals $\mathcal{I} \subset \widehat{\mathcal{D}}_A^\hbar$, because as we saw in Corollary \ref{c:pf} there exists a unique partition function $Z$ such that $\mathcal{I} \cdot Z = 0$. More often than not, what we will be interested in is the partition function $Z$, which may be a generating series for some enumerative invariants. Conversely, if the generating series for a particular set of enumerative invariants takes the form of a partition function for a ($\hbar$-finite) Airy ideal, then what this means is that it satisfies a set of differential constraints that uniquely fix the generating series, such as Virasoro constraints, $\mathcal{W}$-constraints, etc.

How can we construct Airy ideals? We need to construct  a left ideal $\mathcal{I} \subset \widehat{\mathcal{D}}_A^{\hbar}$ that satisfies the properties of Definition \ref{d:airy}. So what we do is construct a collection of differential operators $\{ H_a\}_{a \in A}$ with $H_a \in \widehat{\mathcal{D}}_A^{\hbar}$. To show that the left ideal generated by the $H_a$ is an Airy ideal, what we need to show is:
\begin{enumerate}
\item The collection $\{H_a\}_{a \in A}$ is bounded;
\item $H_a = \hbar \partial_a + O(\hbar^2)$;
\item $[H_a, H_b] = \hbar^2 \sum_{c \in A} f_{abc} H_c$ for some $f_{abc} \in \widehat{\mathcal{D}}_A^{\hbar}$.
\end{enumerate}
It is easy to show that the third condition is equivalent to requiring that $[\mathcal{I},\mathcal{I}] \subseteq \hbar^2 \mathcal{I}$, Condition (4) in Definition \ref{d:airy}.

Concretely, sometimes we may construct  a collection of operators $H_a \in \widehat{\mathcal{D}}_A^\hbar$ that do not satisfy the condition $H_a  = \hbar \partial_a + O(\hbar^2)$, but that satisfy the other two conditions. Does that mean that they do not generate an Airy ideal? Not necessarily. For instance, if $I$ and $A$ are two finite index sets of the same length, then any two collections of differential operators $\{G_i\}_{i \in I}$ and $\{H_a\}_{a \in A}$ related by an invertible $\mathbb{C}$-linear transformation generate the same left ideal in $\widehat{\mathcal{D}}_A^\hbar$. This can be generalized to the infinite context as follows (see Definition 2.3 in \cite{bks}).

\begin{Lem}
\label{l:aairy}
Let $I,A$ be either finite or countably infinite index sets. Let $\{H_i\}_{i \in I}$ be a collection of operators $H_i \in \widehat{\mathcal{D}}_A^\hbar$ such that:
\begin{enumerate}
\item The collection $\{H_i\}_{i \in I}$ is bounded;
\item $[H_i, H_j] = \hbar^2 \sum_{k \in I} f_{ijk} H_k$ for some $f_{ijk} \in \widehat{\mathcal{D}}_A^\hbar$;
\item 
\begin{equation}
\label{eq:degree}
H_i = \sum_{a \in A} M_{i a} \hbar \partial_{a} + O(\hbar^2),
\end{equation}
where the $M_{ia}$ are complex numbers such that for all fixed $a \in A$, the $M_{ia}$ vanish for all but finitely many $i \in I$.
\item There exists complex numbers $N_{bj}$, with $b \in A$ and $j \in I$, such that for all fixed $j \in I$, they vanish for all but finitely many $b \in A$, and
\begin{equation}
\label{eq:sums}
\sum_{i \in I} N_{b i} M_{i a} = \delta_{a b} \qquad \forall a,b \in A \qquad \text{and} \qquad \sum_{a \in A} M_{i a} N_{a j} = \delta_{i j} \qquad \forall i,j \in I.
\end{equation}
\end{enumerate}
Then the left ideal $\mathcal{I} = \{ \sum_{i \in I} c_i H_i\ |\ c_i \in \widehat{\mathcal{D}}_A^\hbar\} \subset \widehat{\mathcal{D}}_A^\hbar$  is an Airy ideal.

Moreover, $\mathcal{I}$ can also be written as $\mathcal{I} = \{ \sum_{a \in A} c_a \tilde H_a\ | \ c_a \in  \widehat{\mathcal{D}}_A^\hbar\}$, where the $\tilde H_a$ are defined by
\begin{equation}
\label{eq:normal}
\tilde H_a = \sum_{i \in I} N_{ai} H_i \in \widehat{\mathcal{D}}^\hbar_A,
\end{equation}
and satisfy $\tilde H_a = \hbar \partial_a + O(\hbar^2)$.
\end{Lem}

\begin{proof}
First, we note that the conditions in \eqref{eq:sums} are well defined, since because of the constraints on the coefficients $N_{bj}$ and $M_{ia}$ both sums collapse to finite sums. 

Now, by condition (2) it is clear that $\mathcal{I}$ satisfies condition (2) of Definition \ref{d:airy}. Moreover, because of condition (1), it is clear that the $\tilde H_a$ are well defined operators in $\widehat{\mathcal{D}}^\hbar_A$. It is also clear from the properties of the complex numbers $M_{ia}$ and $N_{bj}$ that
\begin{align}
\tilde H_a =& \sum_{i \in I} N_{ai} H_i \\
=& \sum_{i \in I}N_{a i} \left( \sum_{b \in A}  M_{i b} \hbar \partial_b \right) +O(\hbar^2)\\
=&\sum_{b \in A} \left(\sum_{i \in I} N_{ai} M_{ib} \right)  \hbar \partial_b +O(\hbar^2)\\
= & \hbar \partial_a +O(\hbar^2),
\end{align}
where in the third line we could exchange the order of the sums since the $M_{ib}$ are such that for all fixed $b \in A$, they vanish for all but finitely many $i \in I$.
Therefore, to show that $\mathcal{I}$ is an Airy ideal, all that we have to show is that the ideals $\mathcal{I} = \{ \sum_{i \in I} c_i H_i\ | \ c_i \in \widehat{\mathcal{D}}^\hbar_A \}$ and $\mathcal{J} = \{ \sum_{a \in A} c_a \tilde H_a\ |\  c_a \in \widehat{\mathcal{D}}^\hbar_A \}$ are the same.

First, we note that we can invert the relation between the $\tilde H_a$ and the $H_i$. As $\tilde H_a = \sum_{i \in I} N_{ai}H_i$, we have
\begin{equation}
\sum_{a \in A} M_{i a} \tilde H_a = \sum_{a \in A} M_{i a} \left(\sum_{j \in I} N_{a j} H_j \right)=  \sum_{j \in I} \left( \sum_{a \in A} M_{ia} N_{aj} \right) H_j = H_i,
\end{equation}
where we could exchange the order of the sums since the $N_{aj}$ are such that for all fixed $j \in I$, they vanish for all but finitely many $a \in A$.

Now suppose that $P \in \mathcal{I}$. Then we can write
\begin{equation}
P = \sum_{i \in I} f_i H_i
\end{equation}
for some $f_i \in \widehat{\mathcal{D}}^\hbar_A$. But then
\begin{equation}
P = \sum_{i \in I} f_i H_i = \sum_{i \in I} f_i \left( \sum_{a \in A} M_{i a} \tilde H_a \right) = \sum_{a \in A} \left( \sum_{i \in I} f_i M_{i a} \right) \tilde H_a.
\end{equation}
Since the $M_{ia}$ are such that for all fixed $a \in A$, they vanish for all but finitely many $i \in I$, the sums $\sum_{i \in I} f_i M_{ia}$ are finite, and hence certainly produce well defined operators in $\widehat{\mathcal{D}}^\hbar_A$. Therefore $P \in \mathcal{J}$, and hence $\mathcal{I} \subseteq \mathcal{J}$.

Conversely, suppose that $Q \in \mathcal{J}$. Then we can write
\begin{equation}
Q = \sum_{a \in A} g_a \tilde H_a = \sum_{a \in A} g_a \left( \sum_{i \in I} N_{a i} H_i \right) = \sum_{i \in I} \left( \sum_{a \in A} g_a N_{a i} \right) H_i,
\end{equation}
for some $g_a \in \widehat{\mathcal{D}}_A^\hbar$. Again, since the $N_{ai} $ are such that for all fixed $i \in I$, they vanish for all but finitely many $a \in A$, the sums $\sum_{a \in A} g_a N_{a i} $ are finite and hence produce well defined operators in $\widehat{\mathcal{D}}^\hbar_A$. Therefore $Q \in \mathcal{I}$, and hence $\mathcal{J} \subseteq \mathcal{I}$. We conclude that $\mathcal{I} = \mathcal{J}$.

\end{proof}

\begin{Rem}
In \cite{bks}, the authors define ``Airy structures in normal form'' as being given by collections of differential operators $(H_a)_{a \in A}$ satisfying the condition $H_a = \hbar \partial_a + O(\hbar^2)$, and ``Airy structures'' as being given by collections of differential operators related to Airy structures in normal forms by transformations as in Lemma \ref{l:aairy}. From our point of view, since we define an Airy structure (or rather an Airy ideal) as being the left ideal itself, we do not need to make such a distinction.
\end{Rem}

\subsection{Two special cases}

We now look at two special cases of the construction, which will be necessary to connect with the Heisenberg algebra in the next section.

Consider the Weyl algebra $\mathcal{D}_A = \mathcal{C}[x_A] \langle \partial_A \rangle$. Let $I \subset A$ be a subset of the index set $A$, and let $J = A \setminus I \subset A$. There are two subalgebras that can be constructed easily:
\begin{enumerate}
\item $\mathcal{D}(x_J,\partial_A) := \mathbb{C}[x_J]\langle \partial_A \rangle \subset \mathcal{D}_A $, which is the subalgebra of differential operators whose coefficients do not depend on the variables $x_I$;
\item $\mathcal{D}(x_A,\partial_J) := \mathbb{C}[x_A]\langle \partial_J \rangle \subset \mathcal{D}_A $, which is the subalgebra of differential operators in which the derivatives $\partial_I$ do not appear.
\end{enumerate}
It is clear that both of those are subalgebras of the Weyl algebra $\mathcal{D}_A$. In fact, in the second case,  we can think of the subalgebra $\mathcal{D}(x_A,\partial_J)$ as being the Weyl algebra $\mathcal{D}_J$ but over the polynomial ring $\mathbb{C}[x_A]$. We will however not take this point of view here, since we want the extra variables $x_I$ to be included in the Bernstein filtration, so it makes perhaps more sense to view $\mathcal{D}(x_A,\partial_J)$ as a subalgebra of $\mathcal{D}_A$.

The Bernstein filtration can be defined on the subalgebras as well, and we can construct the Rees Weyl subalgebras $ \mathcal{D}^\hbar(x_J,\partial_A)$ and $\mathcal{D}^\hbar(x_A,\partial_J)$ and their $\hbar$-adic completions $\widehat{\mathcal{D}}^\hbar(x_J,\partial_A)$ and $\widehat{\mathcal{D}}^\hbar(x_A,\partial_J)$, which are subalgebras of $\widehat{\mathcal{D}}_A^\hbar$. The construction of transvections, twisted modules, modules of exponential types, and Airy ideals also goes through the subalgebras. Let us be a bit more precise.

\subsubsection{The subalgebra  $\widehat{\mathcal{D}}^\hbar(x_A,\partial_J)$}
\label{s:subdelJ}

We consider differential operators in $\widehat{\mathcal{D}}_A^\hbar$ whose coefficients are polynomials in all the variables $x_A$, but that only involve derivatives in $\partial_J$. Then the Rees polynomial algebra $\widehat{\mathcal{M}}_A^\hbar$ is a left $\widehat{\mathcal{D}}^\hbar(x_A,\partial_J)$-module.

We highlight a few simple statements here, which follow from the construction of the previous sections:
\begin{enumerate}
\item Let 
\begin{equation}
\mathcal{I}_{\text{can}}(x_A,\partial_J):=   \left \{ \sum_{j \in J} c_j \hbar \partial_j\ |\ c_j \in \widehat{\mathcal{D}}^\hbar(x_A,\partial_J). \right \}.
\end{equation}
Then $\mathcal{I}_{\text{can}}(x_A,\partial_J)$ is the annihilator of $1 \in \widehat{\mathcal{M}}_A^\hbar$, and the Rees polynomial module $\widehat{\mathcal{M}}_A^\hbar$ is canonically isomorphic to $\widehat{\mathcal{D}}^\hbar(x_A,\partial_J)/\mathcal{I}_{\text{can}}(x_A,\partial_J)$.
\item We can define transvections as in Definition \ref{d:transvection}, with non-trivial action only on the derivatives $\hbar \partial_J$, but with polynomials $p_a^{(n+1)}=\partial_a q^{(n+1)}$ that depend on all variables $x_A$. Those transvections are automorphisms of the subalgebra $\widehat{\mathcal{D}}^\hbar(x_A,\partial_J)$. Then, if we define the left ideal
\begin{equation}
\mathcal{I} = \left \{ \sum_{j \in J} c_j \bar H_j\ |\ c_j \in \widehat{\mathcal{D}}^\hbar(x_A,\partial_J) \right \}
\end{equation}
with the $\bar H_j$ defined as in Definition \ref{d:transvection}, it is clear that $\widehat{\mathcal{D}}^\hbar(x_A,\partial_J)/ \mathcal{I}$ is a cyclic left $\widehat{\mathcal{D}}^\hbar(x_A,\partial_J)$-module canonically isomorphic to the twisted module $\prescript{\phi}{}{\widehat{\mathcal{M}}}_A^\hbar$.
\item We can also rephrase the statement in terms of modules of exponential type. We get that $\widehat{\mathcal{D}}^\hbar(x_A,\partial_J)/ \mathcal{I}$ is canonically isomorphic to the module of exponential type $\widehat{\mathcal{M}}_A^\hbar Z$ with
\begin{equation}\label{eq:pfsc2}
Z = \exp \left( - \sum_{n=0}^\infty \hbar^{n-1} q^{(n+1)}(x_A) \right), 
\end{equation}
where the $q^{(n+1)}(x_J)$ are degree $\leq n+1$ in the variables $x_A$ as usual. In terms of differential equations, what this says is that $Z$ is a solution to the differential equations $\mathcal{I} \cdot Z = 0$. A subtelty arises however in the uniqueness statement. Indeed, there are many choices of polynomials $q^{(n+1)}(x_A)$ that give rise to the same transvection, since the transvection only involves the derivatives $\partial_j q^{(n+1)}(x_A)$ with respect to the variables $\partial_j \in \partial_J$. In other words, any two $q^{(n+1)}(x_A)$ that differ by a polynomial $p^{(n+1)}(x_I)$ in the variables $x_I$ give rise to the same transvection in the subalgebra $\widehat{\mathcal{D}}^\hbar(x_A,\partial_J)$. As a result, we can state the uniqueness result as saying that there is a unique solution to the differential equation $\mathcal{I} \cdot Z = 0$ with $Z\big|_{x_J=0}$, which basically amounts to not only requiring that the $q^{(n+1)}(0) = 0$ but also that they do not depend at all on the variables $x_I$. Or, we could give up on uniqueness, and state that we can construct families of solutions of the form \eqref{eq:pfsc2} parametrized by the variables $x_I$. The unique solution above would then pick the origin of this family.
\end{enumerate}

As usual, we can define Airy ideals as in Definition \ref{d:airy}, but requiring that the ideal $\mathcal{I} \subset \widehat{\mathcal{D}}^\hbar(x_A,\partial_J)$. That is, we only have differential operators $H_j = \hbar \partial_j + O(\hbar^2)$ for $j \in J \subset A$. Then everything goes through, and Theorem \ref{t:airy} holds. Corollary \ref{c:pf} also holds, with the caveat about uniqueness mentioned in point (3) above. 

More precisely, if $\mathcal{I} \subset \widehat{\mathcal{D}}^\hbar(x_A,\partial_J)$ is an Airy ideal, then $\widehat{ \mathcal{D}}^\hbar(x_A,\partial_J)/\mathcal{I}$ is canonically isomorphic to the twisted module $\prescript{\phi}{}{\widehat{\mathcal{M}}}_A^\hbar$ for some stable transvection $\phi$ of $\widehat{\mathcal{D}}^\hbar(x_A,\partial_J)$, and also canonically isomorphic to a module of exponential type $\widehat{\mathcal{M}}_J^\hbar Z$ with $Z$ of the form of \eqref{eq:pfsc2} with $q^{(1)}(x_A) = q^{(2)}(x_A) = 0$. 

From the point of view of differential equations, we conclude that we can construct families of solutions to the differential equations $\mathcal{I} \cdot Z$, parametrized by the extra variables $x_I$, of the form
\begin{equation}\label{eq:pfss}
Z = \exp \left( \sum_{\substack{g \in \frac{1}{2} \mathbb{N}, n \in \mathbb{N}^* \\ 2g-2+n>0}} \hbar^{2g-2+n} F_{g,n}(x_A)\right),
\end{equation}
for some polynomials $F_{g,n}(x_A)$ homogeneous of degree $n$ with $F_{g,n}(0) = 0$. We can also state that there is a unique such solution satisfying the initial condition $Z\big|_{x_J=0}=1$, which amounts to imposing that the polynomials $F_{g,n}$ do not depend on the variables $x_I$ (this is the origin in the family of solutions parametrized by the $x_I$).

\subsubsection{The subalgebra $\widehat{\mathcal{D}}^\hbar(x_J,\partial_A)$}
\label{s:subxJ}

In this case, we consider differential operators in $\widehat{\mathcal{D}}_A^\hbar$ whose coefficients are polynomials but only in the variables $x_J \subset x_A$. Let $\widehat{\mathcal{M}}_J^\hbar$ be the $\hbar$-adic completion of the Rees polynomial module in the variables $x_J$. It is clearly a left $\widehat{\mathcal{D}}^\hbar(x_J,\partial_A)$-module, where the derivatives $\partial_I$ act trivially.

\begin{enumerate}
\item Let 
\begin{equation}
\mathcal{I}_{\text{can}}(x_J,\partial_A):= \left \{ \sum_{a \in A} c_a \hbar \partial_a\ |\ c_a \in \widehat{\mathcal{D}}^\hbar(x_J,\partial_A). \right \}.
\end{equation}
Then $\mathcal{I}_{\text{can}}(x_J,\partial_A)$ is the annihilator of $1 \in \widehat{\mathcal{M}}_J^\hbar$, and the Rees polynomial module $\widehat{\mathcal{M}}_J^\hbar$ is canonically isomorphic to $ \widehat{\mathcal{D}}^\hbar(x_J,\partial_A)/\mathcal{I}_{\text{can}}(x_J,\partial_A) $.
\item We can define transvections as in Definition \ref{d:transvection}, but with the polynomials $p_a^{(n+1)}=\partial_a q^{(n+1)}$ depending only on the variables $x_J$. Those transvections are automorphisms of the subalgebra $\widehat{\mathcal{D}}^\hbar(x_J,\partial_A)$. Then, if we define the left ideal
\begin{equation}
\mathcal{I} = \left \{ \sum_{a \in A} c_a \bar H_a\ |\ c_a \in \widehat{\mathcal{D}}^\hbar(x_J,\partial_A) \right \}
\end{equation}
with the $\bar H_a$ defined as in Definition \ref{d:transvection}, it is clear that $\widehat{\mathcal{D}}^\hbar(x_J,\partial_A)/ \mathcal{I}$ is a cyclic left $\widehat{\mathcal{D}}^\hbar(x_J,\partial_A)$-module canonically isomorphic to the twisted module $\prescript{\phi}{}{\widehat{\mathcal{M}}}_J^\hbar$.
\item We can also rephrase the statement in terms of modules of exponential type as usual, but we need to be a little bit careful here. While the polynomials $p_a^{(n+1)}=\partial_a q^{(n+1)}$ do not depend on the variables $x_I$, this does not mean that the polynomials $q^{(n+1)}$ do not depend on the $x_I$. What it means is that they  can be written as 
\begin{equation}
q^{(n+1)}(x_J) + s^{(n+1)}(x_I)
\end{equation}
for linear polynomials $s^{(n+1)}(x_I)$ in the variables $x_I$. (Here we extend the action of the differential operators on the module by the natural action on the variables $x_I$.) In other words, we get that $\widehat{\mathcal{D}}^\hbar(x_J,\partial_A)/ \mathcal{I}$ is canonically isomorphic to the module of exponential type $\widehat{\mathcal{M}}_J^\hbar Z$ with
\begin{equation}\label{eq:pfsc}
Z = \exp \left( - \sum_{n=0}^\infty \hbar^{n-1} q^{(n+1)}(x_J) - \sum_{n=0}^\infty \hbar^{n-1} s^{(n+1)}(x_I) \right), 
\end{equation}
where the $q^{(n+1)}(x_J)$ are degree $\leq n+1$ in the variables $x_J$ as usual, but the $s^{(n+1)}(x_I)$ are linear polynomials in the variables $x_I$. In terms of differential equations, what this says is that $Z$ is a solution to the differential equations $\mathcal{I} \cdot Z = 0$, and if we set $q^{(n+1)}(0) = 0$ and $s^{(n+1)}(0)=0$, it is the unique solution satisfying $Z\big|_{x_A = 0} = 1$.

\end{enumerate}

Then, we can define Airy ideals as in Definition \ref{d:airy}, but requiring that the ideal $\mathcal{I} \subset \widehat{\mathcal{D}}^\hbar(x_J,\partial_A)$. That is, the $H_a$ are differential operators with coefficients that do not depend on the $x_I$. Then everything goes through, and Theorem \ref{t:airy} holds, with the transvection $\phi$ being of the form above (i.e. not involving the variables $x_I$). Corollary \ref{c:pf} also holds, with the caveat that the partition function takes the form \eqref{eq:pfsc}. 

More precisely, if $\mathcal{I} \subset \widehat{\mathcal{D}}^\hbar(x_J,\partial_A)$ is an Airy ideal, then $ \widehat{\mathcal{D}}^\hbar(x_J,\partial_A)/\mathcal{I}$ is canonically isomorphic to the twisted module $\prescript{\phi}{}{\widehat{\mathcal{M}}}_J^\hbar$ for some stable transvection $\phi$ of the form above, and also canonically isomorphic to a module of exponential type $\widehat{\mathcal{M}}_J^\hbar Z$ with $Z$ of the form of \eqref{eq:pfsc} with $q^{(1)}(x_J) = q^{(2)}(x_J) = s^{(1)}(x_I) = s^{(2)}(x_I) = 0$. Using the standard notation in the literature on Airy structures, we conclude that the unique exponential solution to the differential equations $\mathcal{I} \cdot Z$ with initial condition $Z\big|_{x_A = 0}$ takes the form
\begin{equation}\label{eq:pfsss}
Z = \exp \left( \sum_{\substack{g \in \frac{1}{2} \mathbb{N}, n \in \mathbb{N}^* \\ 2g-2+n>0}} \hbar^{2g-2+n} F_{g,n}(x_J) +  \sum_{g \in \frac{1}{2} \mathbb{N}^*} \hbar^{2g-1} F_{g,1} (x_I) \right),
\end{equation}
for some polynomials $F_{g,n}$ homogeneous of degree $n$ in the respective variables, with $F_{g,n}(0) = 0$.

To connect further with standard notation in the literature, it is customary to write down the following expansions for the homogeneous piolynomials $F_{g,n}$:
\begin{equation}
F_{g,n}(x_J) = \frac{1}{n!} \sum_{j_1, \ldots, j_n \in J} F_{g,n}[j_1,\ldots,j_n] x_{j_1} \cdots x_{j_n},
\end{equation}
where the $F_{g,n}[j_1,\ldots,j_n] \in \mathbb{C}$ are coefficients symmetric under permutations of the entries. What \eqref{eq:pfsss} says is that, if we were to define the coefficients above for all entries $a_j \in A$ and write down a general partition function, then the coefficients $
F_{g,n}[a_1,\ldots,a_n] $ would vanish whenever $n \geq 2$ and at least one of the entries is in $I$.

\subsection{Airy ideals and the Heisenberg algebra}

\label{s:AiryHeisenberg}

The Weyl algebra is obviously closely connected to the universal enveloping algebra of the Heisenberg algebra. Thus, not surprisingly, many Airy ideals can be constructed from this vantage point. We now review this construction.

\subsubsection{The Heisenberg algebra}

\label{s:heisenberg}

Let $\mathfrak{h}$ be the Heisenberg Lie algebra with basis $\{J_n\}_{n \in \mathbb{Z}} \cup \{ c\}$ and Lie bracket
\begin{equation}
[J_m, J_n] = m \delta_{m,-n} c, \qquad [J_m, c]=0, \qquad \forall m,n \in \mathbb{Z}.
\end{equation}
Abusing notation a little bit, we will write $U(\mathfrak{h})$ for the quotient of its universal enveloping algebra by the ideal $c=1$. It is the free associative algebra over $\mathbb{C}$ generated by $\{J_m\}_{m \in \mathbb{Z}}$ modulo the commutation relations
\begin{equation}
[J_m, J_{n}] = m \delta_{m,-n}, \qquad \forall m,n \in\mathbb{Z}.
\end{equation}
It is the algebra of modes of the rank one free boson vertex operator algebra (VOA), often denoted by $\pi$.
 For future use, we define $U_+(\mathfrak{h})$ as being the subalgebra generated by the positive modes, $U_{\geq 0}(\mathfrak{h})$ as being the subalgebra generated by the non-negative modes, $U_-(\mathfrak{h})$ as being the subalgebra generated by the negative modes, and $U_0(\mathfrak{h})$ as being the subalgebra generated by the zero mode.

Simple modules of the free boson VOA (or Heisenberg VOA) are Fock modules $\pi_\lambda$, parameterized by a complex weight label $\lambda$. They are generated by a highest-weight state $| \lambda \rangle$ satisfying
\begin{equation}
J_n | \lambda \rangle = 0 \quad \text{for} \ n>0, \qquad J_0 | \lambda \rangle = \lambda | \lambda\rangle
\end{equation}
and the negative modes act freely on the highest-weight state. In particular as a vector space $\pi_\lambda$ coincides with the polynomial ring in the negative modes. 
Not every module of the Heisenberg VOA is completely reducible and in fact there are infinite length indecomposable modules constructed as follows. 
Consider the polynomial ring $\mathbb C[y]$ in one variable. It becomes a module for the abelian Lie algebra $\mathbb C J_0$ generated by $J_0$ under
\begin{equation}
\rho_\lambda :  \mathbb C J_0 \rightarrow \mathbb C[y], \qquad J_0 \mapsto \frac{d}{dy} +\lambda .
\end{equation}
The module will be denoted by $\rho_\lambda$  and it induces firstly  to a module of the non-negative modes by demanding that the positive modes act as zero and then to a module $\pi_{\rho_\lambda}$ of the free boson VOA by letting all negative modes act freely. In formulas, the induced module is
\begin{equation}
\pi_{\rho_\lambda} = \text{Ind}_{U_{\geq 0}(\mathfrak{h})}^{U(\mathfrak h)} \rho_\lambda.
\end{equation}
$\pi_{\rho_\lambda} $ has the Fock module $\pi_\lambda$ as submodule while the quotient is isomorphic to $\pi_{\rho_\lambda} $ itself, that is it satisfies the non-split exact sequence
\begin{equation}
 0 \rightarrow \pi_\lambda \rightarrow \pi_{\rho_\lambda}  \rightarrow \pi_{\rho_\lambda}  \rightarrow 0.
 \end{equation}
As a vector space $\pi_{\rho_\lambda}$ is isomorphic to the polynomial ring in the negative modes together with the extra variable $y$ on which the zero-mode acts as $ \frac{d}{dy} +\lambda$.

We also want to think of $J_0$ as a variable. For this consider the representation of $\mathbb C J_0$ 
\begin{equation}
\kappa_\lambda :  \mathbb C J_0 \rightarrow \mathbb C[y], \qquad J_0 \mapsto  y +\lambda
\end{equation}
and the induced module 
\begin{equation}
\pi_{\kappa_\lambda} = \text{Ind}_{U_{\geq 0}(\mathfrak{h})}^{U(\mathfrak h)} \kappa_\lambda.
\end{equation}
$\pi_{\kappa_\lambda} $ has the Fock module $\pi_\lambda$ as homomorphic image (mapping $y$ to zero)  while the kernel is isomorphic to $\pi_{\kappa_\lambda} $ itself, that is it satisfies the non-split exact sequence
\begin{equation}
 0  \rightarrow \pi_{\kappa_\lambda}  \rightarrow \pi_{\kappa_\lambda} \rightarrow \pi_\lambda \rightarrow 0.
 \end{equation}
As a vector space $\pi_{\kappa_\lambda}$ is isomorphic to the polynomial ring in the negative modes together with the extra variable $y$ on which the zero-mode acts via multiplication by $ y +\lambda$.  Both modules  $\pi_{\rho_\lambda}$ and $ \pi_{\kappa_\lambda}$ are naturally modules of a slightly larger algebra that we will discuss now. 

\subsubsection{Adding a conjugate zero-mode}

$U(\mathfrak{h})$ is almost the Weyl algebra $\mathcal{D}_{\mathbb{N}^*}$ under the identification
\begin{equation}\label{eq:Heisenbergidentification}
J_m = \partial_m, \qquad J_{-m} = m x_m, \qquad m \in \mathbb{N}^*,
\end{equation}
but not quite: the zero mode $J_0$ is missing. The zero mode is central, as it commutes with all other modes $J_m$. To take into account the zero mode $J_0$, we introduce a conjugate zero mode $\tilde J_0$ that satisfies the commutation relationsz
\begin{equation}\label{eq:zeroconjugate}
[J_m, \tilde J_0] = \pm \delta_{m,0}.
\end{equation}
The choice of sign here will be dictated by our interpretation of the zero mode $J_0$ and its conjugate $\tilde J_0$. We then consider the free associative algebra over $\mathbb{C}$ generated by the $\{J_m\}_{m \in \mathbb{Z}}$ and $\tilde J_0$, which is isomorphic to $U(\mathfrak{h})\otimes_{\mathbb{C}} \mathbb{C}[\tilde J_0]$.
Now we have two choices to map the resulting algebra to the Weyl algebra. We set the parameters  of the modules of the previous section to $y = x_0$ and $\lambda =0$.
\begin{enumerate}
\item[($\pi_{\kappa_0}$)] The case $\pi_{\kappa_0}$ corresponds to a minus sign in \eqref{eq:zeroconjugate}, and, as operators on $\pi_{\kappa_0}$,  $J_0 = x_0$ is a variable, and $\tilde J_0 = \partial_0$ is a derivative. Via the identification \eqref{eq:Heisenbergidentification} we see that $\pi_{\kappa_0}$ is
$\mathbb C[x_0, x_{1}, \dots]$ the polynomial ring in infinitely many variables including $x_0$. Modes of fields of the Heisenberg vertex algebra are infinite sums of monomials in the derivatives, excluding $\partial_0$,  whose coefficients are polynomials, that is elements of $\mathbb C[x_0, x_{1}, \dots]$. This means that $\pi_{\kappa_0}$ carries an action of $\mathcal{D}_\mathbb{N}$ and the algebra of modes of the Heisenberg VOA is contained in the subalgebra $\mathcal{D}(x_{\mathbb{N}}, \partial_{\mathbb{N}^*}) \subset \mathcal{D}_{\mathbb{N}}$ (a special case of the type of subalgebras considered in Section \ref{s:subdelJ}).
\begin{comment}
With this choice, the algebra $U(\mathfrak{h})\otimes_{\mathbb{C}} \mathbb{C}[\tilde J_0]$ is identified with the Weyl algebra $\mathcal{D}_{\mathbb{N}}$ (which includes the variable $x_0$ and its derivative $\partial_0$) via the map
\begin{equation}
J_m = \partial_m, \qquad J_{-m} = m x_m, \qquad J_0 = x_0, \qquad \tilde J_0 = \partial_0, \qquad m \in \mathbb{N}^*.
\end{equation}
Under this identification, the original algebra $U(\mathfrak{h})$, which does not include the conjugate zero mode $\tilde J_0$, and is what we are ultimately interested in, is identified with the subalgebra $\mathcal{D}(x_{\mathbb{N}}, \partial_{\mathbb{N}^*})$, which is a special case of the type of subalgebras considered in Section \ref{s:subdelJ}, with $A = \mathbb{N}$, $J = \mathbb{N}^*$, and $I = \{0 \}$; that is, we consider differential operators whose coefficients are polynomials in the variables $x_a$ with $a \in \mathbb{N}$, but with derivatives only in the variables $\partial_j$ with $j \in \mathbb{N}^*$, i.e. no $\partial_0$. In fact, to make the identification precise we need to consider a particular completion of $U(\mathfrak{h}) $ where we allow infinite linear combinations in the positive modes (the derivatives) but with polynomial coefficients in the non-positive modes. 
\end{comment}
\item[($\pi_{\rho_0}$)] The case $\pi_{\rho_0}$ corresponds to a plus sign in \eqref{eq:zeroconjugate}, and $J_0 = \partial_0$ acts as a derivative, and $\tilde J_0=x_0$ as a variable. The module $\pi_{\rho_0}$ then is also intepreted as $\mathbb C[x_0, x_{1}, \dots]$ and so it again carries an action of $\mathcal{D}_\mathbb{N}$, but this time the modes of fields of the Heisenberg vertex algebra are infinite sums of monomials in the derivatives, including $\partial_0$, but with polynomial coefficients in $\mathbb C[x_{1}, \dots]$. The conjugate zero-mode $\tilde J_0 = x_0$ doesn't appear. This means the algebra of modes of the Heisenberg VOA is now contained in the subalgebra $\mathcal{D}(x_{\mathbb{N}^*}, \partial_{\mathbb{N}}) \subset \mathcal{D}_\mathbb{N}$ (a special case of the type of subalgebras considered in Section \ref{s:subxJ}).
\begin{comment}
In this case,  the algebra $U(\mathfrak{h})\otimes_{\mathbb{C}} \mathbb{C}[\tilde J_0]$ is still identified with the Weyl algebra $\mathcal{D}_{\mathbb{N}}$, but via the map
\begin{equation}
J_m = \partial_m, \qquad J_{-m} = m x_m, \qquad J_0 = \partial_0, \qquad \tilde J_0 = x_0, \qquad m \in \mathbb{N}^*.
\end{equation}
Under this identification, the original algebra $U(\mathfrak{h})$ is identified with the subalgebra $\mathcal{D}(x_{\mathbb{N}^*}, \partial_{\mathbb{N}})$, which is an example of the type considered in Section \ref{s:subxJ} with $A = \mathbb{N}$, $J = \mathbb{N}^*$, and $I = \{0 \}$; that is, we consider differential operators with derivatives in all variables $\partial_a$ with $a \in \mathbb{N}$, but with polynomial coefficients that only depend on the variables $x_j$ with $j \in \mathbb{N}^*$. Again, to make the identification precise, we need to consider a completion of $U(\mathfrak{h}) $ , but this time we allow infinite linear combinations in the non-negative modes (the derivatives), with polynomial coefficients in the negative modes.
\end{comment}
\end{enumerate}

In both cases, we identify a completion of the universal enveloping algebra $U(\mathfrak{h})$ with a subalgebra of the Weyl algebra $\mathcal{D}_{\mathbb{N}}$. Note that as modules for $\mathcal{D}_{\mathbb{N}}$ our two modules $\pi_{\kappa_0}$ and $\pi_{\rho_0}$ are isomorphic and are highest-weight modules generated by a highest-weight vector $| x_0 \rangle$ on which all positive modes act as zero and all non-negative ones act freely. Let us call this  $\mathcal{D}_{\mathbb{N}}$ module $M$.

Finally, we introduce $\hbar$. In both cases, we implement the Rees construction with respect to the filtration on $U(\mathfrak{h})\otimes_{\mathbb{C}} \mathbb{C}[\tilde J_0]$ defined by 
\begin{equation}
F_i  \left(U(\mathfrak{h})\otimes_{\mathbb{C}} \mathbb{C}[\tilde J_0] \right)= \{ \text{polynomials of degree $\leq i$ in the modes $J_m, \tilde J_0$} \}.
\end{equation}
This is of course mapped to the Bernstein filtration as required. We then introduce $\hbar$ via the Rees construction. Ultimately, the result is the free associative algebra over $\mathbb{C}$ generated by $\hbar$, $\hbar \tilde J_0$ and the $\{\hbar J_m\}_{m \in \mathbb{Z}}$ modulo their commutation relations, which is isomorphic to the Rees Weyl algebra $\mathcal{D}_{\mathbb{N}}^\hbar$ (see Remark \ref{r:reesfree}). Finally, we consider the $\hbar$-adic completion, which is mapped to $\widehat{\mathcal{D}}_{\mathbb{N}}^{\hbar}$.

As a result of all this, we have identified the $\hbar$-adic completion of the Rees universal enveloping algebra, which we denote by $\widehat{U}^\hbar(\mathfrak{h})$, with either the subalgebra  $\widehat{\mathcal{D}}^\hbar(x_{\mathbb{N}}, \partial_{\mathbb{N}^*})$, or the subalgebra $\widehat{\mathcal{D}}^\hbar(x_{\mathbb{N}^*}, \partial_{\mathbb{N}})$, depending on the choice of intepretation for the central zero mode $J_0$, i.e. the choice of module $\pi_{\kappa_0}$ or $\pi_{\rho_0}$.

Because of this identification, all the results summarized in Section \ref{s:subdelJ} and \ref{s:subxJ} apply. Let us summarize their meaning in the context of the Heisenberg algebra.
\begin{comment}

\subsubsection{Identifying $\widehat{U}^\hbar(\mathfrak{h})$ with $\widehat{\mathcal{D}}^\hbar(x_{\mathbb{N}}, \partial_{\mathbb{N}^*})$}

\label{s:first}

In this case we consider the zero mode $J_0$ as an extra variable $x_0$. (This is the case that has been considered in the literature so far, for instance in \cite{Airy, Whittaker,bks}.) The polynomial $\mathcal{D}(x_{\mathbb{N}}, \partial_{\mathbb{N}^*})$-module $\mathcal{M}_{\mathbb{N}}$ is naturally identified with the Fock $U(\mathfrak{h})$-module generated by the vector $|x_0\rangle$, such that
\begin{equation}
J_m | x_0 \rangle =0 \qquad \text{for all $m \in \mathbb{N}^*$, and} \qquad J_0 | x_0 \rangle = x_0 | x_0 \rangle.
\end{equation}
The Fock module can then be identified with the subalgebra $U_-(\mathfrak{h}) \otimes_\mathbb{C} U_0(\mathfrak{h})$ (as it is generated by the non-positive modes acting on $|x_0 \rangle$), which in turn is identified with the polynomial algebra in the non-positive modes $J_{-m}$, $m \in \mathbb{N}$. All this construction then goes through the Rees construction and the $\hbar$-adic completion as usual.
\end{comment} 

\subsubsection{Identifying $\widehat{U}^\hbar(\mathfrak{h})$ with $\widehat{\mathcal{D}}^\hbar(x_{\mathbb{N}}, \partial_{\mathbb{N}^*})$}

\label{s:first}

The three results highlighted in Section \ref{s:subdelJ} then have the following interpretation in the case of $\pi_{\kappa_0}$, i.e. $J_0 = x_0$ is a variable:
\begin{enumerate}
\item Let $\mathcal{I}_{\text{can}}$ be the left ideal in $\widehat{U}^\hbar(\mathfrak{h})$ generated by the positive modes $J_m$, $m \in \mathbb{N}^*$. Then $\widehat{U}^\hbar(\mathfrak{h}) / \mathcal{I}_{\text{can}}$ is canonically isomorphic to the ($\hbar$-adic completion of the Rees)  module $M$ generated by $|x_0 \rangle$, which is clear.
\item We define transvections as usual, they take the form $\phi: (\hbar, \hbar J_{-m}, \hbar J_0, \hbar J_m) \mapsto (\hbar, \hbar J_{-m}, \hbar J_0, \bar H_m)$, with
\begin{equation}
\bar H_m = \hbar J_m + \sum_{n=0}^\infty \hbar^n [J_m, q^{(n+1)}(J_0,J_{-1}, J_{-2}, \ldots) ], \qquad m \in \mathbb{N}^*,
\end{equation}
where the $q^{(n+1)}$ are polynomials of degree $\leq n+1$ in the non-positive modes. Then, if we define $\mathcal{I}$ to be the left ideal generated by the $\bar H_m$, $\widehat{U}^\hbar(\mathfrak{h}) / \mathcal{I}$ is canonically isomorphic to the ($\hbar$-adic completion of the Rees) module $M$ twisted by the automorphism $\phi$.
\item Rephrasing in terms of modules of exponential type, we get that $\widehat{U}^\hbar(\mathfrak{h}) / \mathcal{I}$ is canonically isomorphic to the module generated by the state 
\begin{equation}\label{eq:state}
Z |x_0 \rangle = \exp \left( -  \sum_{n=0}^\infty \hbar^{n-1} q^{(n+1)}(J_0,J_{-1}, J_{-2}, \ldots) \right) | x_0 \rangle.
\end{equation}
This is of course a family of Fock modules $\pi_{x_0}$, parametrized by the choice of highest weight $x_0$. There is a unique choice if we impose that the highest weight is $x_0=0$ (in other words, we set the zero mode $J_0$ to zero) and the polynomials satisfy $q^{(n+1)}(0) = 0$.

\end{enumerate}

The interesting statements however are for Airy ideals in $\widehat{U}^\hbar(\mathfrak{h})$. To construct an Airy ideal, we need to construct a collection of operators $\{ H_j \}_{j \in \mathbb{N}^*}$ in $\widehat{U}^\hbar(\mathfrak{h})$ of the form
\begin{equation}
H_j = \hbar J_j + O(\hbar^2),
\end{equation}
and satisfying the properties of Definition \ref{d:airy}. Such collections will naturally arise, for instance, from the modes of the strong generators of some algebras, such as $\mathcal{W}$-algebras, which can be constructed as sub-VOAs of the Heisenberg VOA. If we are given such an Airy ideal $\mathcal{I}$, then we know that $\widehat{U}^\hbar(\mathfrak{h})/\mathcal{I}$ is canonically isomorphic to a twisted module as above for some stable transvection $\phi$, and also canonically isomorphic to a module of exponential type generated by a state as in \eqref{eq:state} with $q^{(1)} = q^{(2)} = 0$.

In particular, since $q^{(1)} = q^{(2)} = 0$ the argument of the exponential starts at $O(\hbar)$, and  expanding the exponential we can think of the state $Z|x_0 \rangle$ as living in the $\hbar$-adic completion of the Rees module $M$ generated by $|x_0 \rangle$.\footnote{To be precise, we would need to define the Rees Fock module slightly more generally here, since if we expand the exponential the polynomial coefficients will be $3$ times the $\hbar$-degree, but this can be done easily without complication.} To summarize, given any Airy ideal, we showed that $\widehat{U}^\hbar(\mathfrak{h})/\mathcal{I}$ is canonically isomorphic to a cyclic left module generated by a state in the $\hbar$-adic completion of the Rees Fock module.

This is particularly interesting if the modes $H_j$ that generate the Airy ideal actually are a subset of the modes (such as the positive modes) of the strong generators of a sub-VOA, such as a $\mathcal{W}$-algebra. In this case, what we have a constructed is a state $Z|x_0 \rangle$ in the $\hbar$-adic completion of the Rees Fock module for $\widehat{U}^\hbar(\mathfrak{h})$ that is annihilated by all the modes $H_j$ in this subset. Considering the action of the other modes of the sub-VOA on this state, we obtain a cyclic module for this sub-VOA, which is generated by the state $Z|x_0 \rangle$. Depending on the subset of modes considered, this may be a highest weight module, or a Whittaker module, for the sub-VOA \cite{Airy,Whittaker}.

\subsubsection{Identifying $\widehat{U}^\hbar(\mathfrak{h})$ with $\widehat{\mathcal{D}}^\hbar(x_{\mathbb{N}^*}, \partial_{\mathbb{N}})$}

\label{s:second}
\begin{comment}
The story is similar in this case, but we identify the zero mode $J_0$ as a derivative $\partial_0$. This case has not been considered in the literature yet, but it will be needed in the following sections for the $\mathcal{W}(\mathfrak{sp}_{2N})$-algebras.

Here we start with the Fock module generated by the vacuum state $|0 \rangle$, such that
\begin{equation}
J_m | 0 \rangle =0 \qquad \text{for all $m \in \mathbb{N}$}.
\end{equation}
Note that it is also killed by the zero mode $J_0$. This Fock module can then be identified with the subalgebra $U_-(\mathfrak{h})$ (as it is generated by the negative modes acting on $|0\rangle$), which is in turn identified with the polynomial algebra $\mathcal{M}_{\mathbb{N}^*}$ in the negative modes $J_{-m}$, $m \in \mathbb{N}^*$. We go through the Rees construction and $\hbar$-adic completion as usual.
\end{comment}
The three highlighted results have the following interpretation in the case of $\pi_{\rho_0}$, i.e. $J_0 = \partial_0$ is a derivative:
\begin{enumerate}
\item Let $\mathcal{I}_{\text{can}}$ be the left ideal in $\widehat{U}^\hbar(\mathfrak{h})$ generated by the non-negative modes $J_m$, $m \in \mathbb{N}$. Then $\widehat{U}^\hbar(\mathfrak{h}) / \mathcal{I}_{\text{can}}$ is canonically isomorphic to the ($\hbar$-adic completion of the Rees)  $\widehat{\mathcal{D}}^\hbar(x_{\mathbb{N}^*}, \partial_{\mathbb{N}})$ submodule of $M$ generated by $| x_0 \rangle$, which is clear.
\item We define transvections as usual, they take the form $\phi: (\hbar, \hbar J_{-m}, \hbar J_0, \hbar J_m) \mapsto (\hbar, \hbar J_{-m},\bar H_0, \bar H_m)$, with
\begin{equation}
\bar H_m = \hbar J_m + \sum_{n=0}^\infty \hbar^n [J_m, q^{(n+1)}(J_{-1}, J_{-2}, \ldots) ], \qquad m \in \mathbb{N},
\end{equation}
where the $q^{(n+1)}$ are polynomials of degree $\leq n+1$ in the negative modes. Then, if we define $\mathcal{I}$ to be the left ideal generated by the $\bar H_m$, $\widehat{U}^\hbar(\mathfrak{h}) / \mathcal{I}$ is canonically isomorphic to the ($\hbar$-adic completion of the Rees) submodule of $M$ generated by $| x_0 \rangle$ and twisted by the automorphism $\phi$.
\item Rephrasing in terms of modules of exponential type, we get that $\widehat{U}^\hbar(\mathfrak{h}) / \mathcal{I}$ is canonically isomorphic to the submodule of $M$ generated by the state 
\begin{equation}\label{eq:state2}
Z |x_0 \rangle = \exp \left( -  \sum_{n=0}^\infty \hbar^{n-1} q^{(n+1)}(J_{-1}, J_{-2}, \ldots)  - \sum_{n=0}^\infty \hbar^{n-1}s^{(n+1)}(\tilde J_0)\right) | x_0 \rangle,
\end{equation}
where the $s^{(n+1)}(\tilde J_0)$ are linear polynomials in the conjugate zero mode $\tilde J_0$, which acts on $|x_0\rangle$ as $\tilde J_0|x_0\rangle = x_0|x_0 \rangle$.
There is a unique choice of generator if we impose that the polynomials satisfy $q^{(n+1)}(0) = s^{(n+1)}(0)=0$. 
\end{enumerate}

We can proceed and study Airy ideals in $\widehat{U}^\hbar(\mathfrak{h})$ as usual. To construct an Airy ideal, we need to construct a collection of operators $\{ H_a \}_{j \in \mathbb{N}}$ in $\widehat{U}^\hbar(\mathfrak{h})$ of the form
\begin{equation}
H_a = \hbar J_a + O(\hbar^2),
\end{equation}
and satisfying the properties of Definition \ref{d:airy}. Note that there is a $ H_0$ associated to the zero mode $J_0$ here. Such collections again naturally arise form sub-VOAs such as $\mathcal{W}$-algebras. If we are given an Airy ideal $\mathcal{I}$, then we know that $\widehat{U}^\hbar(\mathfrak{h})/\mathcal{I}$ is canonically isomorphic to the submodule of $M$  generated by $|x_0 \rangle$ as above but twisted by some stable transvection $\phi$. It is also canonically isomorphic to a module of exponential type generated by a state as in \eqref{eq:state2} with $q^{(1)} = q^{(2)} =s^{(1)}=s^{(2)}= 0$.

What is really interesting here is the appearance of the conjugate zero modes $\tilde J_0$ in the exponential in \eqref{eq:state2}. Again, for an Airy ideal we must have $q^{(1)} = q^{(2)} = s^{(1)}=s^{(2)}=0$, so that the argument of the exponential starts at $O(\hbar)$. We can expand the exponential, but the resulting state \emph{does not} live in the $\hbar$-adic completion of the Rees Fock module generated by $|x_0\rangle$ over $\widehat{U}^\hbar(\mathfrak{h})$, because of the appearance of the conjugate modes $\tilde J_0$. It instead lives in the submodule of $M$ generated by $|x_0 \rangle$ over the $\hbar$-adic completion of the Rees algebra associated to $U(\mathfrak{h}) \otimes_{\mathbb{C}} \mathbb{C}[\tilde J_0]$. This is a key distinction between this scenario and the previous one.

In particular, if the $H_a$ form a subset of modes of the strong generators of a sub-VOA, such as a $\mathcal{W}$-algebra, we once again constructed a state $Z |x_0 \rangle$ that is annihilated by all the modes $H_a$ in this subset, and this states generates a cyclic module for the sub-VOA, which could be a highest weight module, or a Whittaker module, depending on the subset of modes. However, the state $Z | x_0 \rangle$ is not anymore in a Fock module over $U(\mathfrak{h})$ (suitably $\hbar$-adically completed), but rather in the larger module $M =  \pi_{\rho_0}$.

\begin{Rem}
We note that in some cases, $Z | x_0 \rangle$ may still live in the $\hbar$-adic completion of the Rees Fock module generated by $|x_0\rangle$ over $\widehat{U}^\hbar(\mathfrak{h})$. This will happen if all the linear polynomials $s^{(n+1)}(\tilde J_0)$ vanish. In turn, this will happen if the transvection $\phi$ does not act on $\hbar J_0$, that is, $\bar H_0 = \hbar J_0$. From the point of view of Airy ideals, this means that the operator $H_0$ is simply equal to $H_0 = \hbar J_0$. This was the case for instance in some of the constructions in \cite{Airy}.

In this particular case, it does not matter what scenario we use to interpret the zero mode. On the one hand, if we think of $J_0$ as a variable $x_0$, then since $H_0 = \hbar J_0$ must kill $Z | x_0 \rangle$, we must have $x_0=0$, i.e. we set the zero mode to zero. On the other hand, if we think of $J_0$ as a derivative, then $Z | x_0 \rangle$ does not include the conjugate modes $\tilde J_0$ because the $s^{(n+1)}(\tilde J_0)$ vanish, and hence $H_0 = \hbar J_0$ naturally kills $Z | x_0 \rangle$, i.e. we can simply set $J_0$ to zero as before. In both cases the state $Z | x_0 \rangle$ is the same, lives in the Fock module, and we can simply set the zero mode to zero, which is what was done in \cite{Airy}.

However, this is a very particular case; for general Airy ideals, there is no reason why the operator $H_0 = \hbar J_0 + O(\hbar^2)$ that starts with the zero mode should not have terms of $O(\hbar^2)$ or higher. We will see an example of that in the next sections when considering $W(\mathfrak{sp}_{2N})$-algebras.
\end{Rem}

\begin{Rem}
We note that even if the operator $H_0 = \hbar J_0 + O(\hbar^2)$ has higher order terms, it does not mean that the linear polynomials $s^{(n+1)}(\tilde J_0)$  will be non-zero. $H_0$ could still be special enough such that all $s^{(n+1)}(\tilde J_0)=0$, in which case $Z|x_0\rangle$ would live in the $\hbar$-adic completion of the Rees module generated by $|x_0\rangle$ over $\widehat{U}^\hbar(\mathfrak{h})$. It appears to be not so easy to determine whether a given Airy ideal will be such that all $s^{(n+1)}(\tilde J_0) = 0$; however, what is easy to show is that, if $H_0$ has polynomial terms at $O(\hbar^2)$, then the linear polynomials $s^{(n+1)}(\tilde J_0) $ do not all vanish. So this gives a simple criteria to determine when $Z$ involves the conjugate zero modes $\tilde J_0$.
\end{Rem}

\subsubsection{The rank $N$ free boson $\pi^N$}

In this section we simply note that the construction of the previous section continue to hold if we consider direct sums of Heisenberg algebra $\mathfrak{h} := \bigoplus_{i=1}^N \mathfrak{h}^{(i)}$, with basis $\{J^{i}_n\}_{i \in \{1,\ldots,N\}, n \in \mathbb{Z}} \cup \{ c\}$ and Lie bracket
\begin{equation}
[J^{i}_m, J^{j}_n] = m \delta_{m,-n} \delta_{i,j} c, \qquad [J^{i}_m, c]=0, \qquad \forall m,n \in \mathbb{Z}, i,j \in \{1,\ldots,N\}.
\end{equation}
The universal enveloping algebra is constructed as usual, quotienting by the ideal $c=1$. It is the free associative algebra generated by the modes $J^{(i)}_m$ modulo their commutation relations. It is the algebra of modes for the rank $N$ free boson VOA $\pi^N$.

The only difference with the previous section is that we now have $N$ zero modes $J_0^{i}$. We thus introduce $N$ conjugate zero modes $\tilde J_0^{i}$, and proceed as before with the identification with the Weyl algebra $\mathcal{D}_A$, where we now consider the multi-index set $A = \{(i,n)\ |\ i \in \{1,\ldots,N\},  n \in \mathbb{N} \}$.

In principle, for each zero mode we can make a choice between the two scenarios of the previous section, i.e. whether we consider $J_0^{i}$ as a variable or a derivative. It is usually more meaningful to make the same choice for all zero modes. Then we proceed as before, and the results are very similar, so there is no need to re-state them here.

\subsubsection{The VOA viewpoint}

\label{s:VOA}

The last few sections can also be reformulated from the viewpoint of VOAs, which is how the construction of Airy ideals naturally arises. Recall that a VOA is given by the data of a vector space of states $V$, and a state-field correspondence $Y: V \to \text{End}(V)[[z^{\pm1} ]]$, which satisfies a number of defining axioms. Given a vector $v \in V$, we call $Y(v,z) = \sum_{n \in \mathbb{Z}} v_n z^{-n-1}$ the corresponding field, and the endomorphisms $v_n$ its modes.

A VOA module is another space $M$, with a maps $Y_M: V \to \text{End}(M)[[z^{\pm1} ]]$. It realizes the modes of $Y_M(v,z)$ as endomorphisms of the space $M$. 

The rank one free boson VOA is generated by a single state $\chi \in V$, with corresponding field $Y(\chi, z) = \sum_{n \in \mathbb{Z}} J_n z^{-n-1}$. Its modes $J_n$ satisfy the commutation relations of the Heisenberg algebra $\mathfrak{h}$ with $c=1$:
\begin{equation}
[ J_m, J_n] = m \delta_{m,-n}.
\end{equation}
The associative algebra of modes is the universal enveloping algebra $U(\mathfrak{h})$.

To map to $\widehat{\mathcal{D}}^{\hbar}_{\mathbb{N}}$ and $\widehat{\mathcal{M}}^\hbar_{\mathbb{N}}$ as in the previous sections, we think of the Rees polynomial algebra $\widehat{\mathcal{M}}^\hbar_{\mathbb{N}}$ as a VOA module, with map $Y^\hbar: V \to \text{End}(M)[[z^{\pm1} ]]$ acting on the generating state $\chi \in V$ as
\begin{equation}
Y^\hbar(\chi,z) = \sum_{n \in \mathbb{Z}} \hbar J_n z^{-n-1}.
\end{equation}
We also impose that the module satisfies the property
\begin{equation}\label{eq:translation}
Y^\hbar( T v, z) = \hbar \partial_z Y^\hbar (v,z)
\end{equation}
for all $v \in V$, where $T$ is the translation endomorphism on $V$. This turns the algebra of modes into the Rees graded algebra with respect to Li's filtration by conformal weight on the algebra of modes of the free boson \cite{HLi}. It allows us to identify the algebra of modes with a subalgebra of the Rees Weyl algebra $\widehat{\mathcal{D}}^{\hbar}_{\mathbb{N}}$ (as in Sections \ref{s:first} or \ref{s:second}, depending on the interpretation of the zero mode $J_0$), which acts on the module $M= \widehat{\mathcal{M}}^\hbar_{\mathbb{N}}$ (which is also, of course, a left module for the algebra of modes).

Introducing $\hbar$ in this way is in fact very simple. Since it turns the algebra of modes into the Rees graded algebra associated to the filtration by conformal weight, we can simply introduce $\hbar$ at the end of a calculation. Indeed, if $v \in V$ is a state of conformal weight $m$, then we know that its field in the $\hbar$-deformed module will be given by
\begin{equation}
Y^\hbar(v,z) = \hbar^m Y(v,z).
\end{equation}

The rank $N$ free boson VOA is constructed similarly by taking an $N$-fold tensor product of the rank one free boson VOA. It is generated by $n$ states $\chi^{i} \in V$, $i=0,\ldots,N-1$, with corresponding fields $Y(\chi^{i}, z) = \sum_{n \in \mathbb{Z}} J_n^{i} z^{-n-1}$. The modes satisfy the commutation relations
\begin{equation}
[J_m^{i}, J_n^{j}] =m \delta_{m,-n} \delta_{i,j},
\end{equation}
as in the previous section with $c=1$. The algebra of modes is the corresponding universal enveloping algebra, and everything goes through as before.

\subsubsection{Twisted modules for the rank $N$ free boson VOA}

\label{s:twisted}

In many application, the starting point is not quite the algebra of modes of the rank $N$ free boson VOA as in the previous section, but rather the algebra of modes of a twisted module for the rank $N$ free boson VOA.

Roughly speaking, if $\sigma$ is an automorphism of a VOA $V$ of finite order $r$, then a $\sigma$-twisted VOA module is another space $M$ and a map $Y_\sigma: V \to \text{End}(M)[[z^{\pm 1/r}]]$. The difference of  course is that fractional powers of $z$ appear.

In this paper we will only consider the case of the rank $N$ free boson VOA, which is generated by states $\chi^{i} \in V$, $i=0,\ldots,N-1$, with the automorphism $\sigma$ that cyclically permutes the $N$ states:
\begin{equation}
\sigma: \chi^{0} \to \chi^{1} \to \ldots \to \chi^{N-1} \to \chi^{0}.
\end{equation}
In this case, we can define a diagonal basis $v^{a} \in V$, $a=0, \ldots, N-1$, with
\begin{equation}\label{eq:diagonal}
v^{a} = \sum_{j=0}^{N-1} \theta^{-aj} \chi^{j},
\end{equation}
with $\theta = e^{2 \pi i / N}$. The inverse relation is
\begin{equation}\label{eq:inverse}
\chi^{i} = \frac{1}{N} \sum_{a=0}^{N-1} \theta^{i a} v^{a}.
\end{equation}
In this diagonal basis, the automorphism $\sigma$ acts by multiplication by roots of unity:
\begin{equation}
\sigma: v^{a} \mapsto \theta^{a} v^{a}, \qquad a=0,\ldots,N-1.
\end{equation}
The map $Y_\sigma: V \to \text{End}(M)[[z^{\pm 1/r}]]$ takes the simpler form
\begin{equation}\label{eq:twistfields}
Y_\sigma(v^{a},z) = \sum_{k \in \frac{a}{N} + \mathbb{Z}} J_{k N} z^{-k-1},
\end{equation}
with the modes satisfying the commutation relations
\begin{equation}
[ J_{m N}, J_{n N}] = N m \delta_{m,-n}.
\end{equation}

In the end, what we found is that, after redefining indices $m N \mapsto k$, the algebra of modes of the $\sigma$-twisted module is nothing but the universal enveloping algebra of the Heisenberg algebra $\mathfrak{h}$ with $c=1$ already studied in Section \ref{s:heisenberg}, which is the algebra of modes for the rank one free boson. It has only one zero mode $J_0$, not $N$ zero modes as in the untwisted rank $N$ case studied in the previous section.

As the algebra of modes of the twisted modules is identified with the universal enveloping algebra of the Heisenberg algebra $\mathfrak{h}$, all the results of Sections \ref{s:first} and \ref{s:second} apply, depending on a choice of interpretation for the zero mode.

While we only need to consider the  fully cyclic automorphism $\sigma$ in the rest of the paper,  we note that more general automorphisms can certainly be considered, see for instance \cite{Airy,bks,BM}. For an automorphism $\sigma$ in the symmetric group $S_N$ that corresponds to a permutation of the $N$ free bosons, the construction works pretty much the same as explained here, applied independently to each cycle in the permutation (see \cite{bks,BM}). More precisely, in the end one obtains a set of bosonic modes for each cycle of the permutation $\sigma$. The resulting algebra of modes is then naturally identified with the algebra of modes of the untwisted rank $M$ free boson as in the previous section, with $M$ being the number of cycles in $\sigma$. It has $M$ distinct zero modes, one for each cycle in the permutation $\sigma$.

\subsubsection{Boundedness in the VOA setting}

Condition (1) in the definition of Airy ideals, see Definition \ref{d:airy}, states that the collection of operators $\{H_a\}_{a \in A}$ must be bounded. If the Airy ideal is generated by a subset of modes of the strong and free generators of a VOA realized as a sub-VOA of the Heisenberg VOA, then the boundedness condition is automatically satisfied. This is what we prove in this section.

Consider a VOA $W$, that is freely and strongly generated by $N$-fields $W^1, \dots, W^N$ and that allows for an embedding in the rank $N$ Heisenberg algebra $\pi^N$.
Let $H$ be the Virasoro zero-mode of the usual Virasoro field of the Heisenberg algebra.
Let 
\begin{equation}
W^m(z) = \sum_{k \in \mathbb Z} W^m_k z^{-k-1}
\end{equation}
be the mode expansion of the field $W^m(z)$ and we require that $W^m$ has weight $\Delta_m \in \mathbb Z_{>0}$ in the sense that
\begin{equation}
[H, W^m_k] = (\Delta_m - k - 1) W^m_k
\end{equation}
for all $k$. 
Let 
\begin{equation}
W^m_k = \sum_{\substack{ 0 \leq a_1 \leq  \dots \leq a_t\\ i_1, \dots, i_t \in \{0, \dots, N\} }} A^{i_1, \dots, i_N}_{a_1, \dots, a_t}(m, k)  J^{i_1}_{a_1} \dots J^{i_t}_{a_t}
\end{equation}
 where the $A^{i_1, \dots, i_N}_{a_1, \dots, a_t}(m, k)$ are polynomials in the negative modes.
Then the boundedness condition in the VOA setting for the non-negative modes
$W^m_k$  is that for any set $0 \leq a_1 \leq  \dots \leq a_t$ one has $A^{i_1, \dots, i_N}_{a_1, \dots, a_t} (m, k) = 0$ for all but finitely non-negative modes
$W^m_k$. Interpreting the modes of the Heisenberg algebra as variables and derivatives as before immediately translates to the boundedness condition in the Weyl algebra setting. 
\begin{Lem}
The boundedness condtion for non-negative modes
$W^m_k$ holds on modules of $\pi^N$.
\end{Lem}
\begin{proof}
 Note that $[ H, J^i_k] = -k J^i_k$ for the Heisenberg modes. 
Let $I = \{ (i_1, k_1), \dots, (i_r, k_r)\}$ an ordered index set of length $r$, that is $i_1, \dots, i_r \in \{ 1, \dots, N\}$ and $k_1, \dots, k_r \in \mathbb Z$ with $k_a \geq k_{a-1}$ and if $k_a = k_{a-1}$ then $i_a \geq _{a-1}$. 
Set $p_I= J^{i_1}_{k_1} \dots J^{i_r}_{k_r}$  and $k_I = k_1 + \dots + k_r$ so that  $[H, p_I] = - k_I p_I$. 
Let $\mathcal I$ be the set of all such index sets of any length. Then $W^m_k$ is of the form
\begin{equation}
W^m_k = \sum_{\substack{ I \in \mathcal I \\ k_I = k+1 - \Delta_m} }a_I p_I 
\end{equation}
for certain coefficients $a_I$. We are interested in the boundedness conditions in the VOA setting. This means for a given ordered monomial $J^{i_s}_{k_s} \dots J^{i_r}_{k_r}$ 
with $k_s \geq 0$ (and hence all $k_i \geq 0$) we wonder if there exists $\{ (i_1, k_1), \dots, (i_{s-1}, k_{s-1})\}$, such that $a_I \neq 0$ for $I = \{ (i_1, k_1), \dots, (i_r, k_r)\}$.
We have that $k_I \leq k_s + \dots + k_r$ and so we necessarily have $a_I = 0$ if  $ k+1 - \Delta_m > k_s + \dots + k_r$.  In particular there are only finitely many pairs $(m, k)$ such that $a_I$ can be non-zero, i.e. the  boundedness condition holds. 
\end{proof}
Let $\sigma$ be a finite order automorphism of the Heisenberg algebra that leaves $W$ and $H$ invariant. Then the $\sigma$-twisted module is still a module for $W$ and still graded by $H$. Note that with the set-up of the previous setting the twisted modes  $J_{kN}$ have $H$ eigenvalue $-k$. 
As an operator on a $\sigma$-twisted module
\begin{equation}
W^m_k = \sum_{ 0 \leq a_1 \leq  \dots \leq a_t} A^\sigma_{a_1, \dots, a_t}(m, k)  J_{a_1} \dots J_{a_t}
\end{equation}
 with $A^\sigma_{a_1, \dots, a_t}(m, k)$ polynomials in the negative modes.
 \begin{Lem}\label{lemma:VOAboundedness}
The boundedness condtion for non-negative modes
$W^m_k$ holds  on $\sigma$-twisted modules after any possible shift of negative modes.
\end{Lem}
 \begin{proof}
 The argument is the same as the previous Lemma: $J_{a_1} \dots J_{a_t}$ has $H$-eigenvalue $- (a_1 + \dots +a_t)/N$ and so $A_{a_1, \dots, a_t}(m, k) = 0$ if
 $ k+1 - \Delta_m >  a_1 + \dots +a_t)/N$, that is for all but finitely many pairs $(m, k)$.  Any possible shift of negative modes is nothing but a homomorphism on polynomials in the negative modes and so if $A_{a_1, \dots, a_t}(m, k) = 0$ then the same remains true after any possible shift of negative modes. 
 \end{proof} 
 
 \begin{Rem}
The boundedness condition of course holds for subsets of non-negative modes as well. In some cases, see \cite{Airy}, one also wants to include some negative modes $W^m_k$. The argument for boundedness is still exacly the same. Therefore, the collections of modes considered in \cite{Airy,Whittaker} are all bounded, and the Airy ideals constructed in these papers are indeed well defined. 
\end{Rem}

\section{$\mathcal{W}(\mathfrak{sp}_{2N})$ and twisted modules}

We now switch gears, and construct examples of Airy ideals within the algebra of modes of a $\sigma$-twisted module for the rank $N$ free boson VOA. More precisely, we will consider Airy ideals that are generated by the non-negative modes of the strong generators of the principal $\mathcal{W}$-algebra of $\mathfrak{sp}_{2N}$ at level $-N-1/2$, which we denote by
 $\mathcal{W}^{-N-1/2}(\mathfrak{sp}_{2N})$. To do so, we need to realize the $\mathcal{W}^{-N-1/2}(\mathfrak{sp}_{2N})$-algebras as sub-VOAs of the rank $N$ free boson VOA. In this section we review background notions on the $\mathcal{W}^{-N-1/2}(\mathfrak{sp}_{2N})$-algebras, how they can be realized within the rank $N$ free boson VOA, and how we can construct modules for them from twisted modules for the rank $N$ free boson VOA.

\subsection{Generators of $\mathcal{W}(\mathfrak{sp}_{2N})$}

We consider the universal principal $\mathcal{W}$-algebra of type $C_N$ at level $-N-1/2$. 
This algebra is isomorphic to the orbifold of $N$-pairs of symplectic fermions;
the reason is that the coset $\text{Com}\left(V^k(\mathfrak{sp}_{2N}), V^k(\mathfrak{osp}_{1|2N})\right)$ is isomorphic to 
$\mathcal{W}_{\ell}(\mathfrak{sp}_{2N})$ for generic $\ell$ with $\ell$ and $k$ related via $(\ell +N  +1)^{-1} + (k + N +1)^{-1}=2$,
 by \cite[Thm. 4.1]{cl2} as well as \cite[Thm. 3.2]{CGL}. The limit $k \rightarrow \infty$ makes sense; in this limit, the coset becomes an orbifold of a free field algebra \cite[Thm. 6.10]{cl3},
 which in this case is the $Sp(2N)$-orbifold  $\mathcal A(N)^{Sp(2N)}$ of $N$-pairs of symplectic fermions $\mathcal A(N)$. 
 For clarity, we write:
 \begin{equation}
 \mathcal{W}(\mathfrak{sp}_{2N}) := \mathcal A(N)^{Sp(2N)} \cong \mathcal{W}^{-N-1/2}(\mathfrak{sp}_{2N}).
 \end{equation}

Let us be a little more precise. The symplectic fermion algebra of rank $N$, $\mathcal{A}(N)$, is strongly and freely generated by $N$ pairs of symplectic fermions $\{e^{i}(z), f^{i}(z)\}_{i=1,2, \ldots, N}$ by . Their OPEs are given by
\begin{equation}
e^{i}(z) f^{j}(w) \sim \frac{\delta_{ij}}{(z-w)^2}.
\end{equation}

\begin{Prop}[\cite{cl}]
	Let $\{e^{i}, f^{i}\}_{i=1, \ldots, N}$ be symplectic fermions. Then $\mathcal{W}(\mathfrak{sp}_{2N})$ is freely generated by fields $W^2, W^4, \ldots, W^{2N}$ of conformal weights $2, 4, \ldots, 2N$ respectively that have the following free field description:
	\begin{align}
	W^{m}(z)= \frac{1}{(m-2)!}\sum_{i=1}^N \left(:e^i(z) \partial_z^{m-2} f^i(z): + :\partial_z^{m-2} e^i(z)f^i(z): \right),  \qquad m=2,4, \ldots, 2N. \label{eq:wgen}
	\end{align}
\end{Prop}

We can express the above result in terms of free bosonic fields after making use of the boson-fermion correspondence \cite{fms}. Let $Y(\cdot, z)$ denote the state-operator map for the integral lattice VOA $V_{\mathbb{Z}^N}$ generated by an orthonormal basis $\{\chi^0, \chi^1, \ldots,  \chi^{N-1}\}$, and
\begin{equation}
\chi^i(z) = \sum_{ n \in \mathbb{Z}} \chi^i_n z^{-n-1}
\end{equation}
be the free bosonic fields, which satisfy the OPE:
\begin{equation}\label{eq:opeb}
\chi^i(z)\chi^j(w) \sim \frac{\delta_{ij}}{(z-w)^2}.
\end{equation}

Recall that the free fermion OPE is generated by a pair of odd fields $\psi(z), \psi^*(w)$ with OPE
\begin{equation}
\psi(z)\psi^*(w) \sim \frac{1}{z-w}.
\end{equation}
The boson-fermion correspondence gives a pair of free fermions:
\begin{equation}\label{eq:bf}
\psi_i(z) := Y(\e^{\chi_i},z), \quad \psi^*_i(z) := Y(\e^{-\chi_i},z),
\end{equation}
where
\begin{equation}
Y(\e^{\chi_i},z) =  z^{\chi^i_0} U_{\chi^i}
\exp\Bigl( \sum_{n\in \mathbb{Z}_{<0}}
\chi^i_{n} \frac{z^{-n}}{n} \Bigr)
\exp\Bigl( \sum_{n\in \mathbb{Z}_{>0}}
\chi^i_{n} \frac{z^{-n}}{n} \Bigr),
\end{equation} 
and the shift operators $U_{\chi^i}$ satisfy
\begin{align}
[\chi^i_{m},U_{\chi^i}] = \delta_{m,0} \, U_{\chi^i},
\qquad  m\in  \mathbb{Z}.
\end{align}
The fields
\begin{equation}\label{eq:sfbosons}
e^i(z):= \psi_i(z), \quad f^i(z):=\partial_z \psi^*_i(z) 
\end{equation}
generate a VOA isomorphic to $\mathcal{A}(N)$.
\begin{Prop}\label{p:gen}
	Let $\{W^{m}(z)\}_{m=2,\ldots, 2N}$ be the fields defined in \eqref{eq:wgen}, and define the states:
	\begin{equation}
	\nu^m := [\e^{\chi^i}_{-1}\e^{-\chi^i}_{-m} +\e^{\chi^i}_{-m+1}\e^{-\chi^i}_{-2}] \textbf{1}. \label{eq:genvec}
	\end{equation}
	Then 
	\begin{align}
	W^{m}(z)	=  \sum_{i=0}^{N-1} Y( \nu^m ,z). \ \label{eq:gen}
	\end{align} 
\end{Prop}
\begin{proof}
	Translation covariance implies the formula
	\begin{equation}
	f^i(z)=Y(\e^{-\chi^i}_{-2}\textbf{1},z).
	\end{equation}
	The result then follows directly from \eqref{eq:bf} and application of the reconstruction theorem (See \cite[Theorem 4.4.1]{fren}) to the VOA $\mathcal{A}(N)$.
\end{proof}

This Proposition gives us an expression for the strong generating fields of $\mathcal{W}(\mathfrak{sp}_{2N})$ within the rank $N$ free boson VOA, which is the starting point to study whether subsets of modes (such as non-negative modes) of the generators of $\mathcal{W}(\mathfrak{sp}_{2N})$ generate an Airy ideal.

\subsection{Twisted module}

In fact, to construct an Airy ideal generated by the modes of the generators of $\mathcal{W}(\mathfrak{sp}_{2N})$, we will need to start with a $\sigma$-twisted module for the rank $N$ free boson VOA (see Section \ref{s:twisted}). Upon reduction to the $\mathcal{W}(\mathfrak{sp}_{2N})$ sub-VOA, it will become a normal, untwisted, module for $\mathcal{W}(\mathfrak{sp}_{2N})$.

Let us first review basic properties of twisted modules (studied in detail in \cite{bk}) and prove some important formulas. Let $Q$ be an integral lattice with bilinear form $(\cdot, \cdot)$, $\sigma$ an automorphism of $Q$, and $V_{Q}$ be the lattice VOA of $Q$. The bilinear form can be linearly extended to $\mathbb{C} \otimes_{\mathbb{Z}} Q$. Let $M$ be a $\sigma$-twisted $V_{Q}$ module.

We use the same notation as in Section \ref{s:twisted}. We consider the rank $N$ free boson VOA. Let $\chi^0, \chi^1, \ldots \chi^{N-1}$ be an orthonormal basis for $Q$. We consider the cyclic automorphism $\sigma: \chi^0 \mapsto \chi^1 \mapsto  \dots \mapsto \chi^{N-1} \mapsto \chi^0$, and the corresponding $\sigma$-twisted module with map $Y_\sigma$. Let $v^0, v^1, \ldots, v^{N-1}$ be the diagonal basis defined in \eqref{eq:diagonal}, with inverse relation \eqref{eq:inverse}. The twisted fields are as in \eqref{eq:twistfields}, which we reproduce here for convenience:
\begin{equation}
Y_\sigma(v^{a},z) = \sum_{k \in \frac{a}{N} + \mathbb{Z}} J_{k N} z^{-k-1}.
\end{equation}
Using the inverse relation \eqref{eq:twistfields}, we get the twisted fields associated to the original generators $\chi^i$:
\begin{align}
Y_\sigma(\chi^i,z) &= \frac{1}{N} \sum_{m \in \frac{1}{N}\mathbb{Z}} \theta^{imN} J_{mN} z^{-m-1} \\
&= \frac{1}{N} \sum_{m \in \mathbb{Z}} \sum_{a=0}^{N-1} \theta^{ia} 
 J_{a+Nm} z^{-\frac{a}{N}-m-1},
\end{align}
where $\theta = e^{2 \pi i/N}$.

Our goal is to construct the twisted fields $W^m_\sigma(z)$ associated to the $\mathcal{W}(\mathfrak{sp}_{2N})$ generators $W^m(z)$. From Proposition \ref{p:gen}, we know that this involves calculating the twisted fields for products of elements of the form $\e^{\chi^i}_{-a} \e^{-\chi^i}_{-b}$ for $a,b > 0$. The following two formulas do exactly that.  
%
%In addition we modify the formulas slightly to make the $\mathbb{N}$-grading of \eqref{eq:zgrad} compatible with the grading defined by the conformal weight. In other words we wish to keep track of the conformal weight of a field by the power of its $\hbar$-coefficient. This can be achieved by the rescaling
%\begin{align}
%	z \mapsto \hbar^{-1/2}z.
%\end{align}
%Note that this is equivalent to rescaling the modes of the free fields $\{\chi_i\}_{i=0,1,2, \ldots, N}$ by
%\begin{align}
%	(\chi_i)_n \mapsto \hbar^{1/2}(\chi_i)_n.
%\end{align}
%Similarly the derivative rescales as
%\begin{align}\label{eq:rder}
%	\partial_z \mapsto \hbar^{1/2} \partial_z.
%\end{align}
\begin{Lem}\label{l:bruno} Let $d \ge 1$, $\epsilon \in \{1,-1\}$ and $\{\chi^i\}_{i=0,1,\ldots,N-1}$ be an orthonormal basis of $\mathbb{C}^N$. Then 
	\begin{equation}\label{eq:ysig}
	-Y_{\sigma}(\e^{ \epsilon \chi^i}_{-d}\e^{- \epsilon \chi^i}_{-1},z) = \sum_{k=0}^{d} \frac{ c_k }{z^k} S_{d-k}(\epsilon \chi^i,z)
	\end{equation}
	where $c_k$ is the $k$-th coefficient in the Taylor expansion of the function
	\begin{equation}
	g(x)=\frac{1}{N} x^{\frac{1-N}{2N}} \prod_{k=1}^{N-1} (x^{1/N}-\theta^k)
	\end{equation}
	at $x=1$. In particular, $c_k$ is independent of the basis vectors $\chi_i$ and $\epsilon$,  and $c_0=1$, $c_1=0$. Moreover, the $S_k^\hbar$ are the Fa\`a di Bruno polynomials defined by
	\begin{align}
		S_n(\chi^i, z) := \frac{1}{n!} \left(\partial_z + Y_\sigma(\chi^i,z)\right)^n \cdot 1.
	\end{align}
	
\end{Lem}
\begin{proof}
	The proof is exactly analogous to the proof of Lemma 3.7 in \cite{bm}. One can check that, $g'(1)=0$ and thus $c_1=0$. 
\end{proof}

\begin{Rem}
	From \eqref{eq:sfbosons} and the anticommutativity of free fermions it follows that
	\begin{equation}\label{eq:anti}
	Y_{\sigma}(\e^{\chi^i}_{-1}\e^{-\chi^i}_{-d})= -Y_{\sigma}(\e^{-\chi^i}_{-d}\e^{\chi^i}_{-1}).
	\end{equation}
\end{Rem}

\begin{Lem}\label{l:derivation}
	Let $a, b \in V$ be elements of a VOA and $T$ be the translation operator. Then
	\begin{equation}\label{eq:derive}
	\partial_z Y_{\sigma}(a_{-m}b_{-n}\textbf{1},z) = Y_{\sigma}(a_{-m-1}b_{-n}\textbf{1},z) + Y_{\sigma}(a_{-m}b_{-n-1}\textbf{1},z).
	\end{equation}
\end{Lem}
\begin{proof}
	From \cite[equation 3.14]{bk} (see also \eqref{eq:translation} in Section \ref{s:VOA}), 
	\begin{equation} \label{eq:t}
	 \partial_z Y_{\sigma}(v,z) = Y_{\sigma}(Tv,z).
	\end{equation}
	In addition we know that
	\begin{equation}
	Tv = v_{-2} \textbf{1}.
	\end{equation}
	Applying the above two formulas to the vector $v=a_{-m}b_{-n}\textbf{1}$ yields \eqref{eq:derive}. 
\end{proof}
%
%\begin{Rem}
%	In \eqref{eq:t} we have rescaled $\partial_z$ by a factor of $\hbar^{1/2}$ due to \eqref{eq:rder}.
%\end{Rem}
Lemmas \eqref{eq:ysig} and \eqref{l:derivation}, in conjunction with Proposition \ref{p:gen}, allow us to write down explicit expressions for the twisted fields $W^m_\sigma(z)$ associated  to the strong generators of $\mathcal{W}(\mathfrak{sp}_{2N})$. At this stage, we now introduce $\hbar$ in our twisted module as in Section \ref{s:VOA}. It turns the algebra of modes into the graded Rees algebra associated to Li's filtration by conformal weight. Put simply, since the strong generators $W^m_\sigma(z)$ have conformal weights $m$, it simply rescales them by $\hbar^m$.

We summarize the result in the following lemma. 

\begin{Lem}
\label{l:finalform}
 Let  $\{\chi^i\}_{i=0,1, \ldots, N-1}$ be states generating the rank $N$ free boson VOA and $\sigma$ the fully cyclic automorphism. Then the twisted fields $W_{\sigma,\hbar}^m(z)$, $m=2, 4,\ldots, 2N$ corresponding to the strong generators of $\mathcal{W}(\mathfrak{sp}_{2N})$ in the $\hbar$-deformed $\sigma$-twisted module  read:
	 
	\begin{multline}\label{eq:total}
	\begin{split}
		W^m_{\sigma,\hbar}(z)&=\hbar^m  \sum_{i=0}^{N-1} \Bigg[\sum_{k=0}^{m} \frac{ c_k}{z^k} \left(S_{m-k}(-\chi^i,z) + S_{m-k}(\chi^i,z) \right) \\
	&+ \sum_{k=1}^{m-1} \frac{ k c_k}{z^{k+1}} S_{m-1-k}(\chi^i,z) - \sum_{k=0}^{m-1} \frac{c_k}{z^k} \partial_z S_{m-1-k}(\chi^i,z)\Bigg] \qquad m=2,4, \ldots, 2N.
	\end{split}
	\end{multline}
	\end{Lem}
\begin{proof}
	Using Proposition \ref{p:gen} we can write, 
	\begin{align}\label{eq:compy}
	 W^m_{\sigma}(z) = \sum_{i=0}^{N-1}\left( Y_{\sigma}(\e^{\chi_i}_{-1} \e^{-\chi_i}_{-m}\textbf{1})+ Y_{\sigma}(\e^{\chi_i}_{-m+1} \e^{-\chi_i}_{-2}\textbf{1}) \right).
	\end{align}
	The first term in the sum above can be commuted directly from lemma \eqref{l:bruno} with $\epsilon = -1$ and is given by, 
	\begin{align}
		Y_{\sigma}(\e^{\chi_i}_{-1} \e^{-\chi_i}_{-m}\textbf{1}) = \sum_{k=0}^{m} \frac{ c_k}{z^k} S_{m-k}(-\chi_i,z)
	\end{align}
	where the $c_k$s are as in Lemma \ref{l:bruno}. The second term in \eqref{eq:compy} is given by an application of Lemma \eqref{l:derivation} followed by Lemma \eqref{l:bruno}:
	\begin{align}
		Y_{\sigma}(\e^{\chi_i}_{-m+1} \e^{-\chi_i}_{-2}\textbf{1}) &=  \partial_z Y_{\sigma}(\e^{\chi_i}_{-m+1} \e^{-\chi_i}_{-1} \textbf{1},z)- Y_{\sigma}(\e^{\chi_i}_{-m}\e^{-\chi_i}_{-1},z)  \nonumber\\
		& = \sum_{k=0}^{m-1} \frac{kc_k}{z^{k+1}} S_{m-1-k}(\chi_i,z) - \sum_{k=0}^{m-1} \frac{c_k}{z^k} \partial_z S_{m-1-k}(\chi_i,z) + \sum_{k=0}^{m} \frac{  c_k}{z^k} S_{m-k}(\chi_i,z).
	\end{align}
	Adding the above two formulas gives $W^m_\sigma(z)$. Finally, since $W^m_\sigma(z)$ has conformal weight $m$, we get $W^m_{\sigma,\hbar}(z) = \hbar^m W^m_\sigma(z)$, which yields \eqref{eq:total}. 
	\end{proof}
\begin{Rem}
Note that only terms invariant under the automorphism $\sigma$ will survive after summing over the index $i$ in the above formulas, and hence the expressions will simplify considerably. 
\end{Rem}

\section{Constructing Airy ideals for $\mathcal{W}(\mathfrak{sp}_{2N})$}

In the previous section we obtained twisted fields $W^m_{\sigma,\hbar}(z)$ for the strong generators of $\mathcal{W}(\mathfrak{sp}_{2N})$ in a $\hbar$-deformed $\sigma$-twisted module for the rank $N$ free boson VOA. Our goal is to show that the non-negative modes of these fields generate an Airy ideal $\mathcal{I}$ in the $\hbar$-completed Rees algebra of bosonic modes $\widehat{U}^\hbar(\mathfrak{h})$. More precisely, this will not be true for the twisted fields $W^m_{\sigma,\hbar}(z)$ directly. What we will do is construct an automorphism of $\widehat{U}^\hbar(\mathfrak{h})$, which is usually called ``dilaton shift'' (and is very similar to the transvections studied in Section \ref{s:transvection}), such that he image of the fields $W^m_{\sigma,\hbar}(z)$ under this automorphism are such that their non-negative modes generate an Airy ideal.

As explained in Section \ref{s:AiryHeisenberg}, to construct an Airy ideal we need to choose a scenario to map $\widehat{U}^\hbar(\mathfrak{h})$ into the Rees Weyl algebra: we need to make a choice for the action of the zero mode. In the case of $\mathcal{W}(\mathfrak{sp}_{2N})$, only one choice works: we need to make the zero mode acts as a derivative, and hence we find ourselves in the scenario explored in Section \ref{s:second}. As far as we are aware, this is the first example of a $\mathcal{W}$-algebra explored in the literature that involves this scenario.
%
%\vb{I haven't changed this yet.}
%
%In this section we construct an Airy structure as a representation of $\mathcal{W}(\mathfrak{sp}_{2N})$ based on a twisted representations of a Heisenberg algebra. 
%Our strategy goes as follows:
%\begin{enumerate}
%\item We consider a twisted representation for the rank 3 Heisenberg algebra. The algebra of modes of the twisted bosonic fields will form the Weyl algebra $\mathcal{D}^\hbar_A$ from Subsection \ref{ss:weyl}.
%\item As outlined in the previous section, we can obtain a free field bosonic representation of the generators of $\mathcal{W}(\mathfrak{sp_6})$ in terms of the twisted representation of the Heisenberg algebra. We consider the subset of non-negative modes of the strong generators; via the free field bosonic representation, those are realized as differential operators in $\mathcal{D}^\hbar_A$.
%\item We show that the non-negative modes satisfy the properties of Lemma \eqref{l:aairy}, and hence that they generate a left ideal $\mathcal{I}$ for $\mathcal{D}^\hbar_A$ that is an Airy structure.
%\end{enumerate}
%As a result, we conclude that there exists a unique partition function associated to the Airy structure, which generates a $\mathcal{D}^\hbar_A$-module isomorphic to $\mathcal{D}^\hbar_A/ \mathcal{I}$.

 We present first the calculations for the special case of $\mathcal{W}(\mathfrak{sp}_6)$, since the calculations are more explicit. We then move on to the general case.
%We first present explicit formulas for the generators in section \ref{s:sp61}. In sections \ref{s:sp62} \ref{s:sp63}, we compute the effect of certain dilaton shifts on terms that have coefficients $\hbar^0$ and $\hbar^{1/2}$ respectively. However it is clear that no combination of dilaton shifts leads to the the right degree one terms required for an Airy structure. In section \ref{s:sp64} we introduce an operator on the completion of the mode algebra, such that acting on the dilaton shfited generators yields an Airy structure. We generalize the construction to the general $N$ case in \ref{s:genN}. 

\subsection{An Airy ideal for $\mathcal{W}(\mathfrak{sp}_6)$}\label{s:sp61}

\subsubsection{The strong generators}

We use the notation in Section \ref{s:VOA} and Section \ref{s:twisted}. We consider the rank 3 free boson VOA, generated by states $\{ \chi^0, \chi^1, \chi^2 \}$. We consider the fully cyclic automorphism $\sigma$, and the corresponding $\sigma$-twisted module with map $Y_\sigma$. Let $v^0, v^1, v^2$ be the diagonal basis. For clarity, we write
\begin{equation}
\chi^i(z) := Y_\sigma(\chi^i,z), \qquad v^i(z) = Y_\sigma(v^i,z), \qquad i=0,1,2.
\end{equation}
We index the modes of the twisted fields in the diagonal basis as usual:
\begin{align}
v^0(z) &= \sum_{n \in \mathbb{Z}} J_{3n} z^{-n-1}, \nonumber\\
v^1(z) & = \sum_{n \in \mathbb{Z}} J_{3n+1} z^{-n- 4/3}, \nonumber\\
v^2(z) & = \sum_{n \in \mathbb{Z}} J_{3n+2} z^{-n- 5/3}.
\label{eq:bb}
\end{align}
%We want the completed algebra of the modes $J_k$, with the commutator $[J_m, J_n ]= \hbar m \delta_{m,-n}$, to become our Weyl algebra $\mathcal{D}^\hbar_A$. But there is a subtlety, because of the zero mode $J_0$. We represent the positive modes as $J_m = \hbar \partial_{x_m}$, and the negative modes as $J_{-m} = m x_m$. However, we will also want to think of the zero mode $J_0$, which commutes with all $J_k$, $k \neq 0$, as a derivative: $J_0 = \hbar \partial_{x_0}$. With this definition, we can think of the completed algebra of the modes $J_k$ as a subalgebra of the Weyl algebra $\mathcal{D}^\hbar_A$ with $A \in \mathbb{N}$ (i.e. in the variables $(x_0, x_1, x_2, \ldots)$), consisting of differential operators whose coefficients are bounded formal power series that do not depend on $x_0$.
%
%The grading on $\mathcal{D}^\hbar_A$ amounts to assigning the following grading to the bosonic modes:
%$$
%\deg J_k = 1, \qquad 
%\deg \hbar^{1/2} = 1.
%$$
%
%\subsubsection{The representation of the $\mathcal{W}(\mathfrak{sp_6})$-algebra}
%
We introduce $\hbar$ in our module as in Section \ref{s:VOA}, which turns the algebra of modes into the graded Rees algebra associated to Li's filtration by conformal weight.

We obtained the twisted fields $W^m_{\sigma,\hbar}(z)$ for the strong generators of $\mathcal{W}(\mathfrak{sp}_6)$ in our $\hbar$-deformed $\sigma$-twisted module in Lemma \ref{l:finalform}. For clarity, we now drop the subscripts $\sigma,\hbar$ on the generators, and denote them simply by $W^m(z)$. In the case of $\mathfrak{sp}_6$, we have three such generators, $W^2(z), W^4(z), W^6(z)$. We define the modes of the strong generators as
\begin{equation}
W^m(z) = \sum_{k \in \mathbb{Z}}W^m_{k} z^{-k-1}.
\end{equation}

From Lemma \ref{l:finalform} it is easy to see that each of these modes takes the form
\begin{equation}
W^m_k = \hbar^m p^m_k(J_a),
\end{equation}
where $p^m_k(J_a)$ is a polynomial (a sum of normal ordered monomials) in the bosonic modes $J_a$ of degree $ \leq m$. As such, the non-negative modes $W^m_k$ with $k \geq 0$ certainly do not generate an Airy ideal (see Definition \ref{d:airy} and Lemma \ref{l:aairy}), as they are homogeneous of degree $m$ in $\hbar$, and $m>1$. This is not surprising, since $\hbar$ was introduced via Li's filtration by conformal weight, and the fields $W^m(z)$ have conformal weight $m$.

To obtain an Airy ideal we somehow need to break the $\hbar$ homogeneity of the strong generators of $\mathcal{W}(\mathfrak{sp}_6)$. In order to do this, we introduce an automorphism of the Rees algebra of bosonic modes $\widehat{U}^\hbar(\mathfrak{h})$ called ``dilaton shift''. 

\begin{Def}\label{d:dilatonshift}
Let $\phi: \widehat{U}^\hbar(\mathfrak{h}) \to \widehat{U}^\hbar(\mathfrak{h})$ be the automorphism given by
\begin{equation}
\phi: (\hbar, \hbar J_m) \mapsto (\hbar, \hbar J_m + \delta_{m,-3} + \delta_{m,-4}).
\end{equation}
In other words, it simply shifts the modes $\hbar J_{-3} \mapsto \hbar J_{-3} + 1$ and $\hbar J_{-4} \mapsto \hbar J_{-4} + 1$. It can be understood as acting by conjugation; for any $P \in \widehat{U}^\hbar(\mathfrak{h})$, we can think of $\phi(P)$ as being given by
\begin{equation}
\phi(P) = \exp \left( \frac{J_3}{3 \hbar} +\frac{J_4}{4 \hbar} \right) P  \exp \left( - \frac{J_3}{3 \hbar} -\frac{J_4}{4 \hbar} \right). 
\end{equation}
This is very similar to the transvections studied in Section \ref{s:transvection}, except that now the non-trivial action is on the coordinates  (on the negative modes $J_{-3}$ and $J_{-4}$) instead of the derivatives.
\end{Def}

This is an automorphism of the Rees algebra of bosonic modes. The strong generators $W^m(z)$ of $\mathcal{W}(\mathfrak{sp}_6)$ are mapped to new fields $\phi(W^m(z))$ under this automorphism. We introduce the following notation for the image fields and their modes:
\begin{equation}
H^m(z) := \frac{3^{m-1} m!}{2} \phi(W^m(z)), \qquad H^m_k := \frac{3^{m-1} m!}{2} \phi(W^m_k).
\end{equation}
The rescaling of the generators here is simply for convenience.

Clearly, the action of $\phi$ breaks homogeneity in $\hbar$, which is what we want. Our goal is to show that the non-negative modes $H^m_k$, $k \geq 0$, generate an Airy ideal in $\widehat{U}^\hbar(\mathfrak{h})$. To do so, we need to prove Conditions (1)--(4) in Lemma \ref{l:aairy}. We first focus on Condition (3), which amounts to studying the $O(\hbar^0)$ and $O(\hbar^1)$ terms in the $H^m_k$.
%
% are given by Proposition \ref{p:gen} and can be computed explicitly in terms of the bosonic free fields using \eqref{eq:ysig}. We are especially interested in the terms with coefficients $\hbar^0$ and $\hbar^{1/2}$ as these are responsible for degree $0$ and $1$ terms after dilaton shifts. The generators are given by, up to $\mathcal{O}(\hbar)$:
%
%We define the modes of the strong generators as
%
%Via the representation in terms of the twisted bosonic fields, we can think of the modes $W^n_k \in \mathcal{D}^\hbar_A$.
%
%We now introduce a new representation for the strong generators via conjugation, also called ``dilaton shift''.
%\begin{Def}\label{d:dshift}
%	Let $\widehat{T}:=\exp(\frac{J_3}{3\hbar}+\frac{J_4}{4\hbar})$. We denote by $H^i_k \in \mathcal{D}^\hbar_A$ the operators obtained via conjugation by $\widehat{T}$ (also called ``dilaton shift''):
%	\begin{equation}
%	H^i_k =  \frac{3^{i-1} i!}{2}\widehat{T} W^i_k \widehat{T}^{-1}.
%	\end{equation}
%	In terms of the fields $W^n(z)$, the $H^i_k$ are the modes of fields $H^n(z)$ obtained from $W^n(z)$ via the shift
%	\begin{align}
%	\label{eq:shiftv}
%	v_0(z) &\mapsto v_0(z) + 1, \\
%	v_2(z) & \mapsto v_2(z) + z^{1/3}.
%\end{align}
%\end{Def}

\begin{Lem}\label{l:linearterms}
	The modes $H^m_k$ satisfy, for $k \geq 0$:
	\begin{align}
	H^2_k=& \hbar( 2J_{3k} + 2J_{3k+1}) + O(\hbar^2),   \\
	H^4_k  =& \hbar( 4J_{3k} + 12J_{3k+1} + 12 J_{3k+2}+4J_{3k+3}) + O(\hbar^2),\\
	H^6_k=&  \hbar(6J_{3k} + 30J_{3k+1}+60J_{3k+2}+60J_{3k+3}+30J_{3k+4}+6J_{3k+5}) + O(\hbar^2).
	\end{align}
\end{Lem}
\begin{proof}
We start with Lemma \ref{l:finalform} for the strong generators $W^m(z)$, $m=2,4,6$. We mentioned before that the modes take the form
\begin{equation}
W^m_k = \hbar^m p^m_k(J_a),
\end{equation}
where $p^m_k(J_a)$ is a polynomial in the bosonic modes of degree $\leq m$. We note that the  automorphism $\phi$ acts as $\hbar J_{-3} \mapsto \hbar J_{-3} + 1$ and $\hbar J_{-4} \mapsto \hbar J_{-4} + 1$. As such, it can decrease the order in $\hbar$. We are interested in resulting terms of order $O(\hbar^0)$ and $O(\hbar^1)$. Clearly, only the monomials of degree $m$ and $m-1$ in the polynomials $p^m_k(J_a)$ can give rise to terms of order $O(\hbar^0)$ and $O(\hbar^1)$ following the action of $\phi$. So we are only interested in these higher degree terms.

From Lemma \ref{l:finalform}, we get:
\begin{align}
	W^2(z) &= \hbar^2 \sum_{i=0}^2 \left( \chi^i(z)^2  - \partial_z \chi^i(z) + \frac{1}{27 z^2} \right), \\
	W^4(z) & = \hbar^4 \sum_{i=0}^2  \left(\frac{2}{4!}\chi^i(z)^4 - \frac{1}{3!} \partial_z \chi^i(z)^3 +  \text{lower degree} \right), \\
	W^6(z) &= \hbar^6 \sum_{i=0}^2 \left( \frac{2}{6!} \chi^i(z)^6 - \frac{1}{5!}\partial_z \chi^i(z)^5 +\text{lower degree} \right),
\end{align}
where ``lower degree'' means polynomial terms of degree $\leq m-2$ in the bosonic modes.

These expressions are in terms of the bosonic fields $\chi^i(z)$. We need to rewrite them in terms of the twisted fields $v^i(z)$ in the diagonal basis, since the bosonic modes $J_a$ are defined for the twisted fields $v^i(z)$ (see \eqref{eq:bb}). Recall that
\begin{equation}
\chi^i(z) = \frac{1}{3} \sum_{a=0}^2 \theta^{i a} v^a(z),
\end{equation}
where $\theta = e^{2 \pi i/3}$.

We consider first the highest degree terms in the fields $W^m(z)$, of degree $m$ in the bosonic modes. In terms of the fields $v^i(z)$, the highest degree terms read:
\begin{align}
\frac{3}{\hbar^2} W^2(z)  &=\ (v^0(z))^2 + 2 v^1(z) v^2(z) + \ldots \\
\frac{4! 3^3}{2 \hbar^4} W^4(z) &= (v^0(z))^4 + 4 v^0(z) (v^1(z))^3 + 4 v^0(z) (v^2(z))^3 + 12 (v^0(z))^2 v^1(z) v^2(z)\nonumber\\&\quad + 6 (v^1(z))^2 (v^2(z))^2 + \ldots \\
\frac{6! 3^5}{2 \hbar^6} W^6(z) &= (\vi{0})^6+ (\vi{1})^6+ (\vi{2})^6 + 20 (\vi{0})^3 (\vi{2})^3 + 20 (\vi{0})^3 (\vi{1})^3  \nonumber\\& \quad+ 20 (\vi{1})^3 (\vi{2})^3 
+ 30\vi{0} \vi{1} (\vi{2})^4 + 30(\vi{0})^4 \vi{1} \vi{2} \nonumber\\&\quad+ 30\vi{0} (\vi{1})^4 \vi{2} 
 + 90 (\vi{0})^2 (\vi{1})^2 (\vi{2})^2 + \ldots  
\end{align}
We can write a general formula as:
\begin{equation}
\frac{m!3^{m-1}}{2\hbar^m }W^m(z) = \sum_{\substack{\alpha_1+2\alpha_2 |3  \\ \alpha_0+\alpha_1+\alpha_2=m}} \frac{m!}{\alpha_0! \alpha_1! \alpha_2!} (\vi{0})^{\alpha_0} (\vi{1})^{\alpha_1} (\vi{2})^{\alpha_2}	
+ \text{lower degree}.
\end{equation}
In terms of the modes, we get:
\begin{equation}\label{eq:wmodes}
\frac{m! 3^{m-1}}{2\hbar^m}W^m_k = \sum_{\substack{\alpha_1+2\alpha_2 |3,  \alpha_0+\alpha_1+\alpha_2=m  \\ \sum_{p,q,r} \beta_p^0 +\beta_q^1 + \beta_r^2 = 3k+3 -3m  }} \frac{m!}{\alpha_0! \alpha_1! \alpha_2!}  \prod_{i=1}^{\alpha_0} :J_{{\beta_i^0}}: \prod_{i=1}^{\alpha_1} :J_{{\beta_i^1}}: \prod_{i=1}^{\alpha_2} :J_{{\beta_i^2}}: + \text{lower degree},
\end{equation}
where $\beta_i^j \equiv j \text{ (mod $3$)}$. 

With these formulae, we can implement the automorphism $\phi$ from Definition \ref{d:dilatonshift} (the dilaton shift) on the highest degree terms. We see that for the non-negative modes, $k \geq 0$, we obtain precisely the $O(\hbar^1)$ terms in the statement of the Lemma, and no terms of $O(\hbar^0)$.

Next we look at the terms of degree $m-1$ in the bosonic modes in the fields $W^m(z)$. Those could potentially contribute terms of $O(\hbar^1)$ after applying the automorphism $\phi$. In terms of the fields $v^i(z)$, the degree $m-1$ terms read:
\begin{align}
-\frac{1}{\hbar^2}W^2(z) &=  \partial_z v^0(z)+ \ldots \\
-\frac{3! 3^2}{\hbar^4}W^4(z) &=\partial_z [(v^0(z))^3+(v^1(z))^3+(v^2(z))^3 + 6v^0(z)v^1(z)v^2(z)] + \ldots \\
-\frac{5! 3^4 }{\hbar^6}W^6(z) & = \partial_z\left[ \right. (v^0(z))^5+ 10 (v^0(z))^2 (v^2(z))^3 + 20 (v^0(z))^3 v^1(z) v^2(z) + 5(v^1(z))^4 v^2(z)  \nonumber\\
& + 30 v^0(z) (\vi{1})^2 (\vi{2})^2 + 10 (\vi{0})^2 (\vi{2})^5 + 5v^1(z) (\vi{2})^4 \left. \right] +\ldots
\end{align}
The action of the automorphism $\phi$ from Definition \ref{d:dilatonshift} on these degree $m-1$ terms does give rise to $O(\hbar^1)$ terms in the image fields $H^m(z)$.  Those terms take the form:
\begin{align}
	H^2(z) & = 0+ \ldots \\
	H^4(z) & = - 6 \hbar + \ldots \\
	H^6(z) &= - 90 \hbar + \ldots,
	\end{align}
	where we singled out the $O(\hbar)$ terms that arise from applying $\phi$ to the degree $m-1$ terms in the $W^m(z)$. What is key is that these terms are constants, i.e. do not come with powers of $z$.  As a result, they only appear in the modes $H^4_{-1}$ and $H^6_{-1}$, and hence do not contribute to the non-negative modes $H^m_k$ with $k \geq 0$. This concludes the proof of the Lemma.
\end{proof}

\subsubsection{The Airy ideal}
\label{s:sp64}

We now prove that the left ideal generated by the modes $\{H^2_k, H^4_k, H^6_k \}$ with $k \geq 0$ in $\widehat{U}^\hbar(\mathfrak{h})$ is an Airy ideal. We find ourselves in the scenario of Section \ref{s:second}, where the zero mode $J_0$ of the field $v^0(z)$ (which is the only zero mode, see \eqref{eq:bb}) acts as a derivative $\partial_0$.

\begin{Th}\label{t:3}
Let $\mathcal{I}$ be the left ideal in $\widehat{U}^\hbar(\mathfrak{h})$ generated by the $\{H^2_k, H^4_k, H^6_k \}$ with $k \geq 0$. Then $\mathcal{I}$ is an Airy ideal. 
\end{Th}

\begin{proof}
To prove that $\mathcal{I}$ is an Airy ideal, we need to check that Conditions (1)--(4) in Lemma \ref{l:aairy} are satisfied.

\emph{Condition (1)}. The boundedness condition is automatically satisfied for the modes of the fields of a VOA (see Lemma \ref{lemma:VOAboundedness}).

\emph{Condition (2)}. It is always satisfied for the subset of non-negative modes of the strong generators of a VOA  (see Proposition 3.14 in \cite{Airy}). 

\emph{Condition (3)}. For simplicity, let us re-index our operators $H^m_k$, $i=2,4,6$, as
\begin{equation}
H^m_k =: L_{3 k +\frac{m}{2} - 1}.
\end{equation}
Then the operators are indexed by $\{L_i\}_{i \in I}$ with $I=A= \mathbb{N}$. We want to determine whether
\begin{equation}\label{eq:pi1}
L_i =   \sum_{a \in \mathbb{N}} M_{i a} \hbar  J_a + O(\hbar^2)
\end{equation}
for some coefficients $M_{ia}$ such that for all fixed $a \in \mathbb{N}$, they vanish for all but finitely many $i \in \mathbb{N}$. But we have shown in Lemma \ref{l:linearterms} that
	\begin{align}
	\label{eq:pi1h}
	H^2_k=& L_{3 k} = \hbar(2J_{3k} + 2J_{3k+1}) + O(\hbar^2),   \\
	H^4_k =&L_{3k+1}= \hbar(4J_{3k} + 12J_{3k+1} + 12 J_{3k+2}+4J_{3k+3})+O(\hbar^2),\\
	H^6_k=&  L_{3k+2} = \hbar(6J_{3k} + 30J_{3k+1}+60J_{3k+2}+60J_{3k+3}+30J_{3k+4}+6J_{3k+5})+O(\hbar^2).
	\label{eq:pi1h3}
	\end{align}
	As a result, we see that for a fixed $a = 3 k + b$ with $b \in \{0,1,2\}$, the only non-vanishing coefficients $M_{i a}$ are for $i \leq 3k+2$. In particular, for all $a \in \mathbb{N}$ they vanish for all but finitely many $i \in \mathbb{N}$, as required.
	
	\emph{Condition (4).} We need to show that there exists coefficients $N_{bj}$ such that
	\begin{equation}
	\sum_{i \in \mathbb{N}} N_{b i} M_{i a} = \delta_{a b}, \qquad \sum_{a \in \mathbb{N}} M_{i a} N_{a j} = \delta_{ij},
	\end{equation}
	and such that for all fixed $j \in \mathbb{N}$, the coefficients $N_{bj}$ vanish for all but finitely many $ b \in \mathbb{N}$. Equivalently, we need to show that we can invert the relations \eqref{eq:pi1h}--\eqref{eq:pi1h3} to get
	\begin{equation}
	\sum_{ i \in \mathbb{N}} N_{b i}L_i=\hbar J_b  + O(\hbar^2),
	\end{equation}
	with the coefficients such that for all fixed $i \in \mathbb{N}$ they vanish for all but finitely many $b \in \mathbb{N}$.
	
	Let $\textbf{J}_k:=\hbar \begin{pmatrix}
	J_{3k} \\J_{3k+1} \\ J_{3k+2}
	\end{pmatrix}$, $\textbf{K}_k:=\begin{pmatrix}
	L_{3k} \\ L_{3k+1}- 4\hbar J_{3k+3}\\ L_{3k+2} - 60 \hbar J_{3k+3}- 30\hbar J_{3k+4} - 6\hbar J_{3k+5}
	\end{pmatrix}$ and $\textbf{M}= \begin{pmatrix}
	2 & 2 & 0 \\ 4 & 12 & 12 \\ 6 & 30 & 60
	\end{pmatrix}$ for all $k \geq 0$. From \eqref{eq:pi1h} -- \eqref{eq:pi1h3}, we have the matrix equations
		\begin{align}
	\textbf{M} \textbf{J}_k +O(\hbar^2)= \textbf{K}_k, \qquad \forall k \geq 0. 
	\end{align}
	As $\textbf{M}$ is invertible, we can invert the above relation to get 
	\begin{align}
	\hbar J_{3k} + O(\hbar^2) &= \frac{15}{16}L_{3k} -\frac{5}{16} L_{3k+1} +\frac{1}{16}L_{3k+2} - \frac{5}{2}\hbar J_{3k+3} - \frac{15}{8} \hbar J_{3k+4}  - \frac{3}{8}\hbar J_{3k+5}  \label{eq:jinvert3}\\
	\hbar J_{3k+1} + O(\hbar^2) &= -\frac{7}{16}L_{3k} + \frac{5}{16}L_{3k+1} - \frac{1}{16}L_{3k+2} +\frac{5}{2} \hbar J_{3k+3} + \frac{15}{8}\hbar J_{3k+4}  + \frac{3}{8}\hbar J_{3k+5} \\
	\hbar J_{3k+2}+O(\hbar^2) &= \frac{1}{8}L_{3k} - \frac{1}{8}L_{3k+1} + \frac{1}{24}L_{3k+2} -2\hbar J_{3k+3} - \frac{5}{4}\hbar J_{3k+4}  - \frac{1}{24}\hbar J_{3k+5} \label{eq:jinvert4}
	\end{align}
	Substituting back the formulas for $\hbar J_{3k+3}$, $\hbar J_{3k+4}$, and $\hbar J_{3k+5}$ recursively we can write
	\begin{equation}
	\hbar J_{3k+i} + O(\hbar^2) = \sum_{m \in \mathbb{N}} N_{3k+i, m}  L_{m}
	\end{equation}
	for $k \geq 0$ and $i \in \{0,1,2\}$ and where all $N_{3k+i,m}=0$ for $m < 3k$. Equivalently, for any fixed $m \in \mathbb{N}$, the only non-vanishing coefficients are $N_{3k+i,m}$ with $3k \leq m$. In particular, for all fixed $m \in \mathbb{N}$ the coefficients vanish for all but finitely many $3k+i \in \mathbb{N}$, and the condition is satisfied.
	
	As all conditions of Lemma \ref{l:aairy} are satisfied, we conclude that the left ideal $\mathcal{I}$ generated by the $\{H^2_k ,H^4_k, H^6_k \}$, $k \geq 0$, is an Airy ideal. Furthermore, from the calculation above we see that we can also think left ideal $\mathcal{I}$ as being generated by the differential operators ($k \geq 0$, $i \in \{0,1,2\}$):
	\begin{equation}
	\tilde{L}_{3k+i} =  \sum_{m \in \mathbb{N}} N_{3k+i, m}  L_m = \hbar J_{3k+i} + O(\hbar^2).
	\end{equation}
In particular, we notice that we have an operator $\tilde L_0 = \hbar J_0 + O(\hbar^2)$ (and in fact, one can show that the $O(\hbar^2)$ contributions are non-vanishing for this operator). Therefore, we find ourselves in the scenario of Section \ref{s:second}, where the map from the algebra of modes $\widehat{U}^\hbar(\mathfrak{h})$ to the Rees Weyl algebra goes to the subalgebra $\widehat{\mathcal{D}}^\hbar (x_{\mathbb{N}^*}, \partial_{\mathbb{N}})$ -- we need to interpret the zero mode $J_0$ as a derivative in the Weyl algebra.
	
\end{proof}

Now that we know that the left ideal $\mathcal{I}$ generated by the modes of the strong generators of $\mathcal{W}(\mathfrak{sp}_6)$ is an Airy ideal, we obtain an immediate Corollary from Theorem \ref{t:airy} (see also Section \ref{s:second}).

\begin{Cor}
Let $\mathcal{I}$ be the left ideal in $\widehat{U}^\hbar(\mathfrak{h})$ generated by the non-negative modes $\{H^2_k, H^4_k, H^6_k \}$, with $k \geq 0$. Then $\widehat{U}^\hbar(\mathfrak{h}) / \mathcal{I}$ is a cyclic left module canonically isomorphic to the ($\hbar$-adically completed) submodule of $M$ (see Section \ref{s:second}) for the rank 3 free boson VOA generated by $|x_0\rangle$, but twisted by some stable transvection on $\widehat{U}^\hbar(\mathfrak{h})$.

$\widehat{U}^\hbar(\mathfrak{h}) / \mathcal{I}$ is also canonically isomorphic to a module of exponential type generated by a state
\begin{equation}\label{eq:pfssts}
v:= Z |x_0 \rangle= \exp \left( \sum_{\substack{g \in \frac{1}{2} \mathbb{N}, n \in \mathbb{N}^* \\ 2g-2+n>0}} \hbar^{2g-2+n} F_{g,n}(J_{-1}, J_{-2}, \ldots) +  \sum_{g \in \frac{1}{2} \mathbb{N}^*} \hbar^{2g-1} F_{g,1} (\tilde J_0) \right) |x_0 \rangle,
\end{equation}
for some polynomials $F_{g,n}$ homogeneous of degree $n$ in the respective modes, with $F_{g,n}(0) = 0$. Here the $\tilde J_0$ are the modes conjugate to the zero modes $J_0$.

Furthermore, by construction the state $v$ is annihilated by all non-negative modes $\{H^2_k, H^4_k, H^6_k \}$, with $k \geq 0$:
\begin{equation}
H^m_k v = 0, \qquad m = 2,4,6, \quad k \in \mathbb{N}.
\end{equation}
Therefore, the action of the negative modes $H^m_k$, $k <0$ on $v$ generates a  ($\hbar$-adically completed) Fock module for $\mathcal{W}(\mathfrak{sp}_6)$.
\end{Cor}

\begin{Rem}
What is particularly interesting here is that the state $v$ does not live in the $\hbar$-completion of the Fock module generated by $|x_0 \rangle$; indeed, the conjugate modes $\tilde J_0$ appear in $v$. This is a direct consequence of the fact that we need to interpret the zero mode $J_0$ as a derivative instead of a variable -- see Section \ref{s:AiryHeisenberg}, and in particular Section \ref{s:second}.

\end{Rem}

\subsection{Airy ideals for $\mathcal{W}(\mathfrak{sp}_{2N})$}\label{s:genN}
In this section we generalize the above construction for all $N \ge 3$. We follow closely the methods and logic of the previous section.

\subsubsection{The strong generators}

We use the notation in Section \ref{s:VOA} and Section \ref{s:twisted}. We consider the rank $N$ free boson VOA, generated by states $\{ \chi^0, \chi^1, \ldots, \chi^{N-1} \}$. We consider the fully cyclic automorphism $\sigma$, and the corresponding $\sigma$-twisted module with map $Y_\sigma$. Let $v^0, v^1, \ldots, v^{N-1}$ be the diagonal basis. For clarity, we write
\begin{equation}
\chi^i(z) := Y_\sigma(\chi^i,z), \qquad v^i(z) = Y_\sigma(v^i,z), \qquad i=0,1,\ldots,N-1.
\end{equation}
We index the modes of the twisted fields in the diagonal basis as usual:
\begin{align}
v^k(z) &= \sum_{n \in \mathbb{Z}} J_{N n+k} z^{-n-1-k/N}, \quad k=0,1, \ldots, N-1
\end{align}
We introduce $\hbar$ in our module as usual, which turns the algebra of modes into the graded Rees algebra associated to the filtration by conformal weight.

We obtained the twisted fields $W^m_{\sigma,\hbar}(z)$ for the strong generators of $\mathcal{W}(\mathfrak{sp}_{2N})$ in our $\hbar$-deformed $\sigma$-twisted module in Lemma \ref{l:finalform}. As before, for clarity we  drop the subscripts $\sigma,\hbar$ on the generators, and denote them simply by $W^m(z)$. We define the modes of the strong generators as
\begin{equation}
W^m(z) = \sum_{k \in \mathbb{Z}}W^m_{k} z^{-k-1}.
\end{equation}

From Lemma \ref{l:finalform} it is easy to see that each of these modes takes the form
\begin{equation}
W^m_k = \hbar^m p^m_k(J_a),
\end{equation}
where $p^m_kl(J_a)$ is a polynomial (a sum of normal ordered monomials) in the bosonic modes $J_a$ of degree $ \leq m$. As for the $N=3$ case, the modes $W^m_k$ with $k \geq 0$ certainly do not generate an Airy ideal  as they are homogeneous of degree $m$ in $\hbar$, and $m>1$. To obtain an Airy ideal we introduce an automorphism of $\widehat{U}^\hbar(\mathfrak{h})$ (dilaton shift) that breaks the $\hbar$-homogeneity.

\begin{Def}\label{d:dilatonshiftN}
Let $\phi: \widehat{U}^\hbar(\mathfrak{h}) \to \widehat{U}^\hbar(\mathfrak{h})$ be the automorphism given by
\begin{equation}
\phi: (\hbar, \hbar J_m) \mapsto (\hbar, \hbar J_m + \delta_{m,-N} + \delta_{m,-N-1}).
\end{equation}
In other words, it simply shifts the modes $\hbar J_{-N} \mapsto \hbar J_{-N} + 1$ and $\hbar J_{-N-1} \mapsto \hbar J_{-N-1} + 1$. It can be understood as acting by conjugation; for any $P \in \widehat{U}^\hbar(\mathfrak{h})$, we can think of $\phi(P)$ as being given by
\begin{equation}
\phi(P) = \exp \left( \frac{J_N}{N \hbar} +\frac{J_{N+1}}{(N+1) \hbar} \right) P  \exp \left( - \frac{J_N}{N \hbar} -\frac{J_N}{(N+1) \hbar} \right). 
\end{equation}
This is of course a natural generalization of the dilaton shift Definition \ref{d:dilatonshift} for $N=3$.
\end{Def}

The strong generators $W^m(z)$ of $\mathcal{W}(\mathfrak{sp}_{2N})$ are mapped to new fields $\phi(W^m(z))$ under this automorphism. We introduce the following notation for the image fields and their modes:
\begin{equation}
H^m(z) := \frac{N^{m-1} m!}{2} \phi(W^m(z)), \qquad H^m_k := \frac{N^{m-1} m!}{2} \phi(W^m_k).
\end{equation}

Clearly, the action of $\phi$ breaks homogeneity in $\hbar$, which is what we want. Our goal is to show that the non-negative modes $H^m_k$, $k \geq 0$, generate an Airy ideal in $\widehat{U}^\hbar(\mathfrak{h})$. As in the $N=3$ case, we need to prove Conditions (1)--(4) in Lemma \ref{l:aairy}. We first focus on Condition (3), which amounts to studying the $O(\hbar^0)$ and $O(\hbar^1)$ terms in the $H^m_k$.

\begin{Lem}
\label{l:ope}
The modes $H^m_k$ satisfy, for $k \geq 0$ and $m=2,4,6,\ldots, 2N$:
\begin{equation}
H^m_k =  \hbar \sum_{i=0}^{m-1}\frac{m!}{(m-i-1)!i!}  J_{Nk+i} + O (\hbar^2).	
\end{equation}
\end{Lem}
\begin{proof}
We start with Lemma \ref{l:finalform} for the strong generators $W^m(z)$. We mentioned before, the modes take the form
\begin{equation}
W^m_k = \hbar^m p^m_k(J_a),
\end{equation}
where $p^m_k(J_a)$ is a polynomial in the bosonic modes of degree $\leq m$. As in the proof of Lemma \ref{l:linearterms} for $N=3$, to study the $O(\hbar^0)$ and $O(\hbar^1)$ terms in the modes $H^m_k$ we are only interested in the monomials of degree $m$ and $m-1$ in  the polynomials $p^m_k(J_a)$.

First, Lemma \ref{l:finalform} implies that
\begin{equation}\label{eq:genfree}
		\frac{1}{\hbar^{m}}W^{m}(z)= \sum_{i=0}^{N-1} \left[ \frac{2:\chi^i(z)^{m}:}{m!}  - \frac{1}{(m-1)!} \partial_z :\chi^i(z)^{m-1}: \right] + \text{lower degree},
	\end{equation}
	where ``lower degree'' stands for terms of degree $\leq m-2$ in the bosonic modes.

These expressions are in terms of the bosonic fields $\chi^i(z)$. We need to rewrite them in terms of the twisted fields $v^i(z)$ in the diagonal basis. Recall that
\begin{equation}
\chi^i(z) = \frac{1}{N} \sum_{a=0}^{N-1} \theta^{i a} v^a(z),
\end{equation}
where $\theta = e^{2 \pi i/N}$.

We consider the degree $m$ terms in the fields $W^m(z)$. They read:
\begin{equation}
\frac{m!N^{m-1}}{2 \hbar^m}W^m(z) = \sum_{\substack{ \sum_{i=0}^{N-1} i\alpha_i |N  \\ \sum_{i=0}^{N-1} \alpha_i = m}} \frac{m!}{\alpha_0! \alpha_1! \ldots \alpha_{N-1}!} \prod_{p=0}^{N-1} \vi{p}^{\alpha_p}	+\text{lower degree}	.
\end{equation}
In terms of the modes, we get:
\begin{align}\label{eq:wmodesg}
\frac{m!N^{m-1}}{2\hbar^m}W^m_k =  \sum_{\substack{ \sum_{i=0}^{N-1} i\alpha_i |N  \\ \sum_{i=0}^{N-1} \alpha_i = m}} \frac{m!}{\alpha_0! \alpha_1! \ldots \alpha_{N-1}!}  \prod_{i=1}^{\alpha_0} :J_{\beta^0_i}: \prod_{i=1}^{\alpha_1} :J_{\beta^1_i}: \ldots \prod_{i=1}^{\alpha_{N-1}} :J_{\beta^{N-1}_{i}}:  +\text{lower degree}
\end{align}
where $\beta^j_i \equiv j \pmod N$ for all $i$ and $ \sum_{j=0}^{N-1} \sum_{i=1}^{\alpha_j}\beta^j_i = N(k + 1 -m) .$
	 
	 Now we want to implement the dilaton shift (the automorphism $\phi$ of Definition \ref{d:dilatonshiftN} on these highest degree terms. From \eqref{eq:wmodesg}, the $O(\hbar^0)$ terms can only come from a term proportional to $(J_{-N})^{m}$, but these terms only come up in the mode expansion of $W^m_{-1}$. Therefore the dilaton shift does not produce $O(\hbar^0)$ terms in the non-negative modes.

	As for the $O(\hbar^1)$ terms, they are essentially determined by the conditions:
	\begin{equation}\label{eq:cond}
	 \sum_{i=0}^{N-1} i \alpha_i | N, \quad \sum_{i=0}^{N-1} \alpha_i = m.
	\end{equation}
	 To get a  $O(\hbar^1)$ term we have to shift all modes except for one in the mode expansion given by \eqref{eq:wmodesg}, hence we only need to consider terms with $\alpha_i =0$ for $i=1,2, \ldots, N-2$ or $\alpha_i=1$ for some $i=1,2, \ldots, N-2.$ 
	 For a term of the form $\prod_i J_{\gamma_i}$ we recall that the modes add up as follows, 
	 \begin{equation}
	   \sum_i \gamma = N(k+1-m).
	 \end{equation}
	  The lowest and highest index of the $O(\hbar^1)$ terms produced in $H^m_k$ are due to the terms with $\alpha_0=m$ and $\alpha_{N-1}=m-1$ respectively, that is terms of the form
	 \begin{equation}
	 	:J_{Nk} (J_{-N})^{m-1}:, \qquad b_{m,k}:J_{Nk+m-1}(J_{-(N+1)})^{m-1}:,
	 \end{equation}
	and are given by
	 \begin{equation}
	 	m J_{Nk}, \qquad d_{n,k} J_{Nk+m-1},
	 \end{equation}
	 respectively, where
	 \begin{align}
	 	b_{m,k} &= 1 , \qquad \mathrm{if} \,\, m- 1 \equiv N-1 \pmod N \\
	 	b_{m,k} &= m, \qquad \mathrm{otherwise} 
	 \end{align}
	 and $d_{m,k}=m$.
	  In between these two extreme cases we have terms of the form, 
	 \begin{equation}
	 \frac{m!}{(m-a-1)!a!}	:J_{Nk+a} (J_{-N})^{m-a-1} (J_{-N-1})^{a}: , \quad  a=1,2, \ldots, m-2 ,
	 \end{equation}
	 and after the dilaton shift these yield the following $O(\hbar^1)$ terms:
	 \begin{equation}
	 \frac{m!}{(m-a-1)!a!}	J_{Nk+a}.
	 \end{equation}
	 
	 Finally, we look at the sub-leading degree $m-1$ terms in the fields $W^m(z)$, which may contribute $O(\hbar^1)$ terms after dilaton shift. From \eqref{eq:genfree} the degree $m-1$ term in $\frac{1}{\hbar^m} W^m(z)$ is proportional to
	 \begin{equation}
	 \sum_{i=0}^{N-1}\partial_z: (\chi^i(z))^{m-1}: .
	 \end{equation}
	 After changing to the diagonal basis only terms of the form
	 \begin{equation}\label{eq:dterms}
	\partial_z:(v^0(z))^{m-1}:, \quad \partial_z:(v^{N-1}(z))^{m-1}: \delta_{N|m-1}, \quad \partial_z :(v^0(z))^{m-1-a} (v^{N-1}(z))^{a}: \delta_{N| a }, \quad 0<a<m-1,
	 \end{equation}
	 can yield $O(\hbar^1)$ corrections. Note that these are terms that are firstly invariant under $\sigma$ and secondly have as factors only the dilaton shifted fields $v_0(z)$ and $v_{N-1}(z)$. After performing the dilaton shifts $\hbar J_{-N} \mapsto J_{-N} +1$ and $\hbar J_{-N-1} \mapsto \hbar J_{-N-1}+1$ it is easy to check that $O(\hbar)$ corrections are only produced for negative modes, and are zero for all non-negative modes, as in the $N=3$ case. The easiest way to see this is to reformulate the dilaton shift as, 
	 \begin{align}
	 	v_0(z) \mapsto v_0(z) + \frac{1}{\hbar}, \\
	 	v_{N-1}(z) \mapsto v_{N-1}(z) + \frac{z^{1/N}}{\hbar}.
	 \end{align}
	 Then the  $O(\hbar)$ terms produced by each of the terms mentioned in \eqref{eq:dterms} in the operators $W^m(z)$ from this dilaton shift are respectively of the form, 
	 \begin{align}
	 	 \hbar \partial_z (z^0), \quad \hbar\partial_z(z^{\frac
	 {m-1}{N}}), \quad \hbar\partial_z(z^{\frac{a}{N}}).
	 \end{align}
	 As the powers of $z$ in the above expressions are non-negative, the result follows.

%	 Using \eqref{eq:shiftfds1} and \eqref{eq:shiftfds2} one can compute the $O(\hbar^1)$ terms thus produced to be zero for all non-negative modes.

\end{proof}

\subsubsection{The Airy ideal}

We now prove that the left ideal generated by the modes $\{H^m_k\}_{m=2,4,\ldots,2N}$ with $k \geq 0$ in $\widehat{U}^\hbar(\mathfrak{h})$ is an Airy ideal. We find ourselves in the scenario of Section \ref{s:second}, where the zero mode $J_0$ of the field $v^0(z)$ (which is the only zero mode, see \eqref{eq:bb}) acts as a derivative $\partial_0$.

\begin{Th}\label{t:N}
Let $\mathcal{I}$ be the left ideal in $\widehat{U}^\hbar(\mathfrak{h})$ generated by the $\{H^m_k\}_{m=2,4,\ldots,2N}$ with $k \geq 0$. Then $\mathcal{I}$ is an Airy ideal. 
\end{Th}

\begin{proof}
To prove that $\mathcal{I}$ is an Airy ideal, we need to check that Conditions (1)--(4) in Lemma \ref{l:aairy} are satisfied.

\emph{Condition (1)}. The boundedness condition is automatically satisfied for the modes of the fields of a VOA (see Lemma \ref{lemma:VOAboundedness}).

\emph{Condition (2)}. It is always satisfied for the subset of non-negative modes of the strong generators of a VOA  (see Proposition 3.14 in \cite{Airy}).

\emph{Condition (3)}.
For simplicity, let us re-index our operators $H^m_k$, $i=2,4,\ldots,2N$, as
\begin{equation}
H^m_k =: L_{N k +\frac{m}{2} - 1}.
\end{equation}
Then the operators are indexed by $\{ L_i\}_{i \in I}$ with $I=A= \mathbb{N}$. We want to determine whether
\begin{equation}\label{eq:pi1N}
L_i =   \sum_{a \in \mathbb{N}} M_{i a}\hbar J_a + O(\hbar^2)
\end{equation}
for some coefficients $M_{ia}$ such that for all fixed $a \in \mathbb{N}$, they vanish for all but finitely many $i \in \mathbb{N}$. But we showed in Lemma \ref{l:ope} that
\begin{equation}
\label{eq:pi11}
H^m_k = L_{Nk + \frac{m}{2}-1} =  \hbar \sum_{i=0}^{m-1}\frac{m!}{(m-i-1)!i!}  J_{Nk+i} + O(\hbar^2).	
\end{equation}
In particular, we can write (for $k \geq 0$ and $n \in \{0,1,\ldots,N-1\}$)
\begin{equation}
L_{Nk + n}= \sum_{a \in \mathbb{N}} M_{Nk+n, a} \hbar J_a + O(\hbar^2),
\end{equation}
with $M_{Nk+n, a} = 0$ for all $a < N k$. In other words, for a fixed $a$, the only non-vanishing coefficients $M_{Nk+m,a}$ are for $N k \leq a$. In particular, for any fixed $a \in \mathbb{N}$, the coefficients $M_{i a}$ vanish for all but finitely many $i \in \mathbb{N}$, as required. 
	
	\emph{Condition (4).} We need to show that there exists coefficients $N_{bj}$ such that
	\begin{equation}
	\sum_{i \in \mathbb{N}} N_{b i} M_{i a} = \delta_{a b}, \qquad \sum_{a \in \mathbb{N}} M_{i a} N_{a j} = \delta_{ij},
	\end{equation}
	and such that for all fixed $j \in \mathbb{N}$, the coefficients $N_{bj}$ vanish for all but finitely many $ b \in \mathbb{N}$. Equivalently, we need to show that we can invert the relation \eqref{eq:pi1} to get
	\begin{equation}
	 \sum_{ i \in \mathbb{N}} N_{b i} L_i = \hbar J_b + O(\hbar^2)
	\end{equation}
	with the coefficients such that for all fixed $i \in \mathbb{N}$ they vanish for all but finitely many $b \in \mathbb{N}$.
	
	We proceed as for the $N=3$ case. We rewrite \eqref{eq:pi11} as
	\begin{equation}\label{eq:rewrite}
	 L_{Nk + \frac{m}{2}-1} =  \sum_{i=0}^{m-1}c^m_i  \hbar J_{Nk+i} + O(\hbar^2), \qquad c^m_i = \frac{m!}{(m-i-1)!i!}.
	\end{equation}
	Let $\textbf{J}_k:=\hbar \begin{pmatrix}
	J_{Nk} \\J_{Nk+1} \\ \vdots \\  J_{Nk + N - 1}
	\end{pmatrix}$, $\textbf{K}_k:=\begin{pmatrix}
	L_{Nk} - \sum_{m> Nk+ N+1} c^2_m \hbar J_{Nk+m} \\L_{Nk+1} -   \sum_{m> Nk+ N+1} c^4_m \hbar J_{Nk+m} \\ \vdots \\ L_{Nk+N-1} - \sum_{m> Nk+ N+1} c^{2N}_m  \hbar J_{Nk+m}
	\end{pmatrix}$ and 
	\begin{equation}(\textbf{M})_{ij} = 
	\begin{cases}
	 0, \quad j > 2i    \\
	 c^{2i}_{j-1},  \quad \text{otherwise}
	\end{cases}
	\end{equation}
	for $1 \le i,j \le N$.
	From \eqref{eq:rewrite} we get matrix equations
	\begin{equation}
	\textbf{M} \textbf{J}_k  +O(\hbar^2)= \textbf{K}_k, \qquad k \ge 0. 
	\end{equation}
	One can see that $\textbf{M}$ is invertible. This can be shown in various ways, such as converting to an upper/lower triangular matrix using Gaussian elimination or LU decomposition. Note that dividing the $i$th row by $i$ gives a matrix, 
\begin{equation}(\textbf{M})_{ij} = 
\begin{cases}
0, \quad j>2i  \\
{2i \choose  j}, \quad \text{otherwise}
\end{cases}
\end{equation}
which is a submatrix of the infinite `Pascal matrix'. Submatrices of the Pascal matrix with non-zero diagonal entries are invertible, as proved in \cite{kersey}.

As $\textbf{M}$ is invertible, we can invert the matrix equations to get (for $k \geq 0$ and $i \in \{0,1,\ldots,N-1\}$):
	\begin{align}\label{eq:jpgen}
	\hbar J_{Nk+i} + O(\hbar^2) = \sum_{j=1}^{N} \delta_j L_{Nk+j-1} + \sum_{j=N} ^{2N-1}\epsilon_j \hbar J_{Nk+j}
	\end{align}
	for some constants $\delta_j, \epsilon_j \in \mathbb{Q}$.  Substituting back recursively the formulae for $\hbar J_{Nk+j}$, we end up with
	\begin{equation}
	\hbar J_{Nk+i} + O(\hbar^2) = \sum_{m \in \mathbb{N}} N_{Nk+i,m} \pi^{\leq 1}(L_m),
	\end{equation}
	for $k \geq 0$ and $i \in \{0,1,2,N-1\}$ and where all $N_{Nk+i,m}=0$ for $m < Nk$. Equivalently, for any fixed $m \in \mathbb{N}$, the only non-vanishing coefficients are $N_{Nk+i,m}$ with $Nk \leq m$. In particular, for all fixed $m \in \mathbb{N}$ the coefficients vanish for all but finitely many $Nk+i \in \mathbb{N}$, and the condition is satisfied.
	
	As all conditions of Lemma \ref{l:aairy} are satisfied, we conclude that the left ideal $\mathcal{I}$ generated by the $\{H^m_k\}$, $k \geq 0$, $m \in \{2,4,\ldots,2N\}$ is an Airy structure. Furthermore, from the calculation above we see that we can also think of the left ideal as being generated by the differential operators ($k \geq 0$, $i \in \{0,1,2,\ldots,N-1\}$):
	\begin{equation}
	\tilde{L}_{Nk+i} =  \sum_{m \in \mathbb{N}} N_{Nk+i, m}  L_m = \hbar J_{Nk+i} + O (\hbar^2).
	\end{equation}
In particular, as in the $N=3$ case we have an operator $\tilde L_0 = \hbar J_0 + O(\hbar^2)$. Therefore, we are again in the scenario of Section \ref{s:second}, where we need to intepret the zero mode $J_0$ as a derivative in the Weyl algebra.
\end{proof}

Now that we know that the left ideal $\mathcal{I}$ generated by the modes of the strong generators of $\mathcal{W}(\mathfrak{sp}_{2N})$ is an Airy ideal, we obtain an immediate Corollary from Theorem \ref{t:airy} (see also Section \ref{s:second}).

\begin{Cor}
Let $\mathcal{I}$ be the left ideal in $\widehat{U}^\hbar(\mathfrak{h})$ generated by the non-negative modes $\{H^m_k \}_{m=2,4,\ldots,2N}$, with $k \geq 0$. Then $\widehat{U}^\hbar(\mathfrak{h}) / \mathcal{I}$ is a cyclic left module canonically isomorphic to the ($\hbar$-adically completed) submodule of $M$ (see Section \ref{s:second}) for the rank $N$ free boson VOA generated by $|x_0\rangle$, but twisted by some stable transvection on $\widehat{U}^\hbar(\mathfrak{h})$.

$\widehat{U}^\hbar(\mathfrak{h}) / \mathcal{I}$ is also canonically isomorphic to a module of exponential type generated by a state
\begin{equation}
v:= Z |0 \rangle= \exp \left( \sum_{\substack{g \in \frac{1}{2} \mathbb{N}, n \in \mathbb{N}^* \\ 2g-2+n>0}} \hbar^{2g-2+n} F_{g,n}(J_{-1}, J_{-2}, \ldots) +  \sum_{g \in \frac{1}{2} \mathbb{N}^*} \hbar^{2g-1} F_{g,1} (\tilde J_0) \right) |0 \rangle,
\end{equation}
for some polynomials $F_{g,n}$ homogeneous of degree $n$ in the respective modes, with $F_{g,n}(0) = 0$. Here the $\tilde J_0$ are the modes conjugate to the zero modes $J_0$.

Furthermore, by construction the state $v$ is annihilated by all non-negative modes $\{H^m_k\}_{m=2,4,\ldots,2N}$, with $k \geq 0$:
\begin{equation}
H^m_k v = 0, \qquad m = 2,4, \ldots, 2N \quad k \in \mathbb{N}.
\end{equation}
Therefore, the action of the negative modes $H^m_k$, $k <0$ on $v$ generates a  ($\hbar$-adically completed) Fock module for $\mathcal{W}(\mathfrak{sp}_{2N})$.
\end{Cor}

\begin{Rem}
As in the $N=3$ case, the state $v$ does not live in the $\hbar$-completion of the Fock module generated by $|x_0 \rangle$ for the  free boson VOA; indeed, the conjugate modes $\tilde J_0$ appear in $v$. This is a direct consequence of the fact that we need to interpret the zero mode $J_0$ as a derivative instead of a variable -- see Section \ref{s:AiryHeisenberg}, and in particular Section \ref{s:second}.

\end{Rem}

\end{document}